\documentclass[11pt]{article}
\usepackage[utf8]{inputenc}
\usepackage[T1]{fontenc}
\usepackage{amsthm, amsmath}
\usepackage{amsfonts}
\usepackage[margin=1in]{geometry}
\usepackage{graphicx} % Required for inserting images
\usepackage[draft]{fixme}
\usepackage[colorlinks]{hyperref}
\usepackage{thm-restate}
\usepackage{tikz}
\usepackage[inline]{enumitem}
\usepackage{mathtools}
\usepackage[ruled,noend,resetcount,linesnumbered]{algorithm2e}
\usepackage[capitalize]{cleveref}
\usepackage[colorinlistoftodos]{todonotes}
\newtheorem{theorem}{Theorem}[section]
\newtheorem{lemma}[theorem]{Lemma}
\newtheorem{fact}[theorem]{Fact}
\newtheorem{corollary}[theorem]{Corollary}
\newtheorem{observation}[theorem]{Observation}
\crefname{observation}{Observation}{Observations}
\newtheorem{definition}[theorem]{Definition}
\newtheorem{claim}[theorem]{Claim}

\fxsetup{mode=multiuser,theme=color, layout=inline}
\FXRegisterAuthor{tk}{atk}{\color{blue} {\bf Tomasz}}
\fxsetup{mode=multiuser,theme=color, layout=inline}
\FXRegisterAuthor{bs}{abs}{\color{orange} {\bf Barna}}
\FXRegisterAuthor{jg}{ajg}{\color{purple} {\bf Jacob}}

\newcommand{\Oh}{O}
\newcommand{\tOh}{\tilde{\Oh}}
\newcommand{\Dyck}{\mathsf{Dyck}}
\newcommand{\ed}{\mathsf{ed}}
\newcommand{\edd}{\mathsf{ed_d}}
\newcommand{\ded}{\mathsf{ded}}
\newcommand{\ted}{\mathsf{ted}}
\newcommand{\dedd}{\mathsf{ded_d}}
\newcommand{\C}{\mathcal{C}}
\newcommand{\D}{\mathcal{D}}
\newcommand{\X}{\mathcal{X}}
\newcommand{\hvy}{\mathcal{H}}
\newcommand{\hvys}{\mathcal{S}}
\newcommand{\Lp}{\mathcal{L}}
\newcommand{\Rp}{\mathcal{R}}
\newcommand{\T}{\mathcal{T}}
\newcommand{\dd}{\mathinner{.\,.\allowbreak}}
\newcommand{\lcp}{\mathsf{LCP}}
\newcommand{\lcpt}{\mathsf{LMP}}

\newcommand{\add}{\mathsf{add}}
\newcommand{\addT}{\mathsf{addT}}
\newcommand{\spl}{\mathsf{split}}
\newcommand{\tw}{\mathsf{twin}}
\newcommand{\rg}{\mathsf{range}}
\newcommand{\con}{\mathsf{concat}}
\newcommand{\out}{\mathbf{O}}
\newcommand{\heavy}[1]{\lfloor \lg |#1| \rfloor}
\newcommand{\chd}{\mathsf{child}}

\newcommand{\piece}{\sigma}
\newcommand{\laq}{\mathsf{LAQ}}
\newcommand{\lca}{\mathsf{LCA}}
\newcommand{\ipm}{\mathsf{IPM}}

\SetKwProg{Fn}{}{:}{}

\newcommand{\A}{\mathcal{A}}
\newcommand{\onto}{\to}

\newcommand{\Zz}{\mathbb{Z}_{\ge 0}}

\newcommand{\F}{\mathcal{F}}
\newcommand{\G}{\mathcal{G}}
\newcommand{\Hcal}{\mathcal{H}}

\newcommand{\str}[2][]{\mathsf{P}_{#1}(#2)}
\newcommand{\ta}{\mathsf{TA}}
\newcommand{\hld}{\mathsf{hld}}
\newcommand{\paren}{\mathsf{P}_{\hld}}

\newcommand{\dedt}{\mathsf{ted_d}}

\usetikzlibrary{shapes}
\usepackage{authblk}

\title{\vspace{-1.5cm}Dynamic Dyck and Tree Edit Distance: \\ 
 Decompositions and Reductions to String Edit Distance%
 \thanks{This work is partially supported by DARPA QuICC, ONR MURI 2024 award on Algorithms, Learning, and Game Theory, Army-Research Laboratory (ARL) grant W911NF2410052, NSF AF:Small grants 2218678, 2114269, 2347322.}}

\author[1]{Debarati Das}
\author[2]{Jacob Gilbert}
\author[2]{MohammadTaghi Hajiaghayi}
\author[3]{Tomasz Kociumaka}
\author[4]{Barna Saha}

\affil[1]{Pennsylvania State University, United States}
\affil[ ]{\texttt{debaratix710@gmail.com}}
\affil[2]{University of Maryland, United States}
\affil[ ]{\texttt{jgilber8@umd.edu}\; \texttt{hajiaghayi@gmail.com}}
\affil[3]{Max Planck Institute for Informatics, Germany}
\affil[ ]{\texttt{tomasz.kociumaka@mpi-inf.mpg.de}}
\affil[4]{University of California, San Diego, United States}
\affil[ ]{\texttt{bsaha@ucsd.edu}}

\date{\vspace{-1.5cm}}

\begin{document}

\maketitle
\setcounter{page}{0}
\thispagestyle{empty}
\begin{abstract} 
In this paper, we present the first dynamic algorithms for Dyck edit distance and tree edit distance that achieve subpolynomial update times. The Dyck edit distance measures how far a parenthesis string is from a well-parenthesized expression (i.e., the Dyck language), while the tree edit distance quantifies the minimum number of node insertions, deletions, and substitutions required to transform one rooted, ordered, and labeled tree into another. These problems have been studied extensively since the 1970s, with recent advances in both algorithmic efficiency and fine-grained complexity lower bounds.  

Despite this progress, no prior work has addressed efficient dynamic algorithms for these problems, even though many real-world applications involve evolving structured data such as LaTeX, JSON, XML, HTML, hierarchical datasets, and RNA secondary structures. We make the first step in this direction by designing new approximation algorithms for Dyck and tree edit distances in the dynamic setting.  

Our key technical contribution is a set of novel reduction and decomposition techniques that transform instances of Dyck and tree edit distance into efficiently maintainable instances of string edit distance. Leveraging existing dynamic algorithms for string edit distance, we obtain an \( n^{o(1)} \)-approximation for Dyck edit distance with \(n^{o(1)}\) update time. This builds upon but significantly extends prior work on Dyck language decomposition ([Saha, FOCS'14] and [Koucký, Saks, SODA'23]). For tree edit distance, we introduce a new static reduction that improves the best-known approximation bound from \( n^{3/4} \) [Akutsu et al., Algorithmica, 2010] to \( \tilde{O}(\sqrt{n}) \), while generalizing beyond the prior constraint of constant-degree trees. While Akutsu's result only held for constant degree trees, our's hold for any trees. We then extend this reduction dynamically, yielding a dynamic tree edit distance algorithm with an approximation factor of \(n^{1/2+o(1)}\) and update time \(n^{o(1)} \).  

A core component of our approach is a new dynamic maintenance algorithm for heavy-light decomposition, a widely used technique in tree algorithms. Given its broad applicability, we believe this result is of independent interest. Finally, we introduce a novel static and dynamic decomposition method that achieves an $\tilde{O}(k)$-approximation for tree edit distance when the tree edit distance is at most \( k \); combined with the trivial bound $k\le n$, this yields a deterministic \( \tilde{O}(\sqrt{n}) \)-approximation.
While similar decompositions exist for strings, no prior work has successfully extended them to trees. Our approach breaks this barrier, improving the best-known approximations for tree edit distance both in the static and dynamic setting. In the static setting, our algorithm runs in $\tilde{O}(n)$ time; in the dynamic setting, it only requires a poly-logarithmic worst-case update time. The best known linear-time static algorithm for tree edit distance previously achieved an $O(\sqrt{n})$-approximation [Boroujeni, Ghodsi, Hajiaghayi, Seddighin, STOC'19].
\end{abstract}
\newpage
\setcounter{page}{1}
\section{Introduction}

The \emph{Dyck language}  represents the set of well-balanced sequences of parentheses and is one of the most fundamental context-free languages. Early work by Chomsky and Schützenberger \cite{CHOMSKY1963118} showed that a nondeterministic version of Dyck is equivalent to the entire set of context-free grammars. Since then, the Dyck language has not only played a pivotal role in formal language theory but has also found numerous applications in other fields. For example, the popular \emph{RNA folding} problem in bioinformatics requires a similar balance condition between various nucleotides \cite{Nussinov1980, bringmann2019, DasKS2022} and the modern-day programming languages HTML, XML, and JSON maintain balanced sequences of parentheses to determine the structure of stored information. Most of these applications also require correcting errors, and thus determining how far the underlying data semantics deviate from the well-balanced property. This is formulated as the \emph{Dyck edit distance} and has been studied since the early seventies \cite{aho1972minimum}.  The Dyck edit distance counts the minimum number of insertions, deletions, and substitutions needed for a given parenthesis string to become well balanced, that is to become a member of the Dyck language. Furthermore, in many of the above applications, the parenthesis strings being considered change frequently (e.g., mutations in RNA sequences, edits in programming text, etc.) and the corresponding Dyck edit distance must be recomputed many times over. This motivates the study of \emph{dynamic Dyck edit distance} where the goal is to maintain the Dyck edit distance quickly over such updates.

Along with Dyck edit distance, we consider a fundamental measure of distances over trees: \emph{tree edit distance}. The tree edit distance involves comparing two ordered, rooted, and labeled trees to determine the minimum number of node insertions, deletions, and relabelings required to transform one tree into another. Originally introduced by Selkow~\cite{SELKOW1977184}, tree edit distance has found a wide range of applications including in computational biology (e.g., in analyzing RNA molecules, whose secondary structure is depicted as a rooted tree)~\cite{10.1016/j.tcs.2004.12.030,10.5555/262228,10.1137/0213024,DBLP:journals/bioinformatics/ShapiroZ90}, image analysis~\cite{10.1016/S0167-8655(97)00179-7}, and compiler optimization~\cite{10.1145/1644015.1644017}, as well as in NoSQL big databases like MongoDB~\cite{MongoDB}. Like in Dyck edit distance, the change in underlying data requires updating the tree edit distance, and the goal of \emph {dynamic tree edit distance} is to maintain these updates quickly.

In this work, we initiate the study of dynamic Dyck and tree edit distance problems. In this setting, the underlying sequence goes through dynamic insertions, deletions, and substitutions. The goal is to maintain Dyck edit distance or tree edit distance fast with such updates. The dynamic variants are particularly important for practical applications, such as efficiently tracking and updating changes in structured documents or correcting compilation errors. However, even the fastest known algorithms in the static setting for these problems are quite slow.  The static Dyck edit distance problem admits a simple folklore $\Oh(n^3)$-time dynamic programming solution which was improved in a series of papers \cite{bringmann2019, williams2020truly,CDX22,chi2022faster1,NPSVXY25} to $\tilde{O}(n^{\frac{3+\omega}{2}})$\footnote{$\tilde{O}(n)$ notation is commonly used to hide $\text{poly}\log{n}$ factors.} where $\omega$ is the exponent of fast matrix multiplication, $\omega \leq 2.371339$ from \cite{ADVXXZ24}. For tree edit distance, Mao \cite{Xiao21} gave the first subcubic algorithm with a runtime of $\Oh(n^{2.9546})$ for tree edit distance, which was then slightly improved by D\"{u}rr to 2.9148 \cite{Durr23}. While this improvement came after a long series of works \cite{10.1145/322139.322143,ZhangS89,Klein98,10.1145/1644015.1644017}, it also uses fast matrix multiplication. For both of these problems, the best combinatorial algorithms that do not use fast matrix multiplication stands at $O(n^3)$, and fine-grained complexity states that this running time is conditionally optimal \cite{AbboudBW2018, 10.1145/3381878}. 

To obtain faster static running time, a series of work have studied designing approximation algorithms \cite{Saha2014, Saha2015, DasKS2022, KSSODA2023, DBLP:conf/stoc/BoroujeniGHS19,DBLP:conf/innovations/SeddighinS22}. If we concentrate our attention to nearly-linear time algorithms, that is algorithms running in $O(n^{1+\epsilon})$ time for any $\epsilon >0$, the best known approximation factor known for the Dyck edit distance is $O(\log{n})$ in $O(n^{1+\epsilon})$ time for any constant $\epsilon >0$ and $O(2^{\sqrt{\log{n}\log{\log{n}}}})$ in $O(n^{1+o(1)})$ time\footnote{These bounds come from \cite{Saha2014, KSSODA2023} utilizing the best bounds known for the string edit distance problem in the desired time bounds \cite{10.1145/1536414.1536444,ANFOCS20}.}, whereas for the tree edit distance, it is $O(\sqrt{n})$  in $\tilde{O}(n)$ time \cite{DBLP:conf/stoc/BoroujeniGHS19}. This raises the following question.

\begin{center}
{\it 
Can we obtain a dynamic algorithm for the Dyck and Tree edit distance problems that have update times of $n^{\epsilon}$, $\epsilon >0$, and achieve approximation factors that are close to known near-linear time static algorithms for these problems?}
\end{center}

\paragraph*{Results.} In this paper, we make significant progress towards resolving the above question. We show for the Dyck edit distance, if we allow $n^{o(1)}$ update time, then a $O(2^{\sqrt{\log{n}\log{\log{n}}}})$ approximation is possible. On the other hand, if we let the update time to be $n^{\epsilon}$, for any constant $\epsilon >0$, then we can achieve an approximation factor of $(\log{n})^{O(\frac{1}{\epsilon})}$.
\sloppy

\begin{theorem}
\label{thm:dyck-improved-simple}
{\normalfont{(Informal)}}
   There exists a randomized dynamic algorithm, that maintains a $\text{poly}\log n$-approximation of Dyck edit distance (correctly with high probability against an oblivious adversary) for an (initially empty) string $X$ of length at most $n$ undergoing edits. The expected amortized update time of the algorithm is $n^{O(\epsilon)}$ per update for any constant $\epsilon >0$. Moreover, if we allow an approximation factor of $O(2^{\sqrt{\log{n}\log{\log{n}}}})$, then an expected amortized update time of $n^{o(1)}$ is possible.
\end{theorem}

For tree edit distance, we get the following result, which achieves an $\Oh(\sqrt{n\log n})$-factor approximation in the worst case
(since the maintained value can be capped by $n$ without loss of generality). However, when $\ted(\F,\G)$ is small (less than $\sqrt{n/\log{n}}$), it achieves a significantly better approximation ratio. Moreover, the update time is only polylogarithmic.

\newcommand{\apted}{\widetilde{\smash{\ted}}}

\begin{theorem}\label{thm:dyn-tree-kk2-simple}
   There exists a deterministic dynamic algorithm that, for (initially empty) forests $\F$ and $\G$ of size at most $n$ undergoing edits, maintains a value $\apted(\F,\G)$  such that $\ted(\F,\G) \le \apted(\F,\G) \le \Oh(\ted(\F,\G)^2 \log n)$.
   The algorithm supports updates in  $\Oh(\log^{3+o(1)}n)$ worst-case time.
\end{theorem}

\paragraph*{{\bf Reductions and Decompositions to String Edit Distance.}}
At the core of two of our main results lie new dynamic reductions and decomposition methods to string edit distance. Simplistically speaking, we show an instance of a Dyck edit distance problem can be decomposed into multiple disjoint string edit distance instances in such a way that it is possible to maintain this decomposition fast with dynamic updates, and we do not lose much in the approximation factor. In the static setting, Saha~\cite{Saha2014} gave the first such decomposition method for Dyck edit distance that loses a $\log {n}$ factor in approximation. While her method was randomized, Kouck{\'{y}} and Saks \cite{KSSODA2023} simplified the reduction and made it deterministic. In this paper, one of our main contributions is to implement a novel decomposition inspired by these reductions under dynamic updates. 

For tree edit distance, we also take a similar approach. We show there exists a reduction from tree edit distance to a \emph{single} instance of string edit distance which loses only an $O(\sqrt{n})$ factor in the approximation. This significantly improves the best-known such result by Akutsu, Fukagawa, and Takasu~\cite{AFT10} that provides an $O(n^{3/4})$-approximate reduction, and only for trees that have constant maximum degree. Not only do we improve the approximation factor in the static setting, we also show how this reduction can be maintained dynamically. Using this novel tree to string reduction, we are able to obtain the following dynamic result.

\begin{theorem}{\normalfont{(Informal)}}
\label{thm:dyn-tree-simple}
   There exists a randomized dynamic algorithm, that maintains an $O(n^{\frac{1}{2}+o(1)})$-approximation of tree edit distance between $\F$ and $\G$ (correctly with high probability against an oblivious adversary) for (initially empty) forests $\F, \G$ of size at most $n$ undergoing edits. The expected amortized update time of the algorithm is  $n^{o(1)}$ per edit.
\end{theorem}

The above approach is both simple and powerful as it enables us to use the known dynamic approximation algorithms for string edit distance \cite{KMS2023} as a black box. 

Compared to dynamic graph algorithms, the development in dynamic string algorithms is relatively recent~\cite{chen2013dynamic, CKM20, KMS2023, bringmann2024dynamic}. We add several important tools to this growing field of research. As part of our Dyck and tree reductions to string edit distance, we propose a new dynamic heavy-light decomposition (see Section~\ref{subsec:heavy}) method where the heavy paths do not need to be split or merged, and therefore can be maintained efficiently over updates. 
Given the significance of heavy-light decomposition in the literature of strings, we believe our approach will be useful for many other applications as well. 

\paragraph*{Improved Static Approximation Algorithm for Tree Edit Distance}
Our results on tree edit distance not only provide the first nontrivial algorithm in the dynamic setting, but also surpass the best known approximation bound for tree edit distance achievable in near-linear time.
Specifically, our approximation algorithm from Theorem~\ref{thm:dyn-tree-kk2-simple} distinguishes whether a given tree edit distance instance has distance at most $k$ or at least $c k^2 \log n$, for an integer $k$, input size $n$, and a sufficiently large constant $c$.
Consequently, whenever $\ted(\F, \G) = o\big(\sqrt{n / \log n}\big)$, our method outperforms the state-of-the-art near-linear-time algorithm, achieving an $O(\sqrt{n})$-factor approximation~\cite{DBLP:conf/stoc/BoroujeniGHS19}.

The $k$-versus-$k^2$ gap problem has been studied extensively for string edit distance approximation in near-linear~\cite{LandauV11988,GRS:20,KSSODA2023} and sublinear~\cite{GKS19,KS20a,GRS:20} time, as well as for Dyck edit distance approximation in near-linear time~\cite{Saha2014,KSSODA2023}.
However, none of the existing techniques for edit distance problems in this regime extend naturally to trees.
We develop a new string decomposition approach that not only generalizes to the tree setting—yielding improved approximation guarantees and runtime over prior work on tree edit distance—but can also be efficiently maintained in the dynamic setting.

\subsection*{Roadmap}
In \cref{sec:prelim}, we provide a short preliminaries section with useful definitions, notations, and brief observations.  
This is followed by a technical overview (\cref{sec:to}) describing our reductions to string edit distance, our novel tree edit distance approximation algorithm, and a discussion of related work in Section~\ref{sec:related}.  
In Section~\ref{sec:fastapprox}, we present our dynamic Dyck approximation algorithm, along with proofs of its runtime and correctness.  
In \cref{sec:treetostring}, we introduce our new tree-to-string reduction in both the static and dynamic settings.  
Section~\ref{sec:kTreeApprox} provides our state-of-the-art tree edit distance approximation algorithm with a proof of correctness and its adaptation to the dynamic setting.  
Finally, \cref{sec:approx,sec:exact} present additional results on exact dynamic Dyck edit distance using techniques from prior work.

\section{Preliminaries}
\label{sec:prelim}
Given a string $X = x_0 x_1\ldots x_{n-1}$, we denote $x_i$ as $X[i]$ and substrings $x_i x_{i+1} \ldots x_{j-1}$ as $X[i \dd j-1] = X(i-1\dd j)$. $X[i\dd]$ denotes $X[i\dd|X|)$, the \emph{suffix} of $X$ starting at $i$. Given an interval $I \subseteq [0 \dd n)$, $X[I]$ is the substring of $X$ of all characters in $X$ with indices that fall in interval $I$. Any substring of $X$ starting with the first character of $X$ is a \emph{prefix} of $X$, i.e. $X[0\dd i)$ for any $i \in [0\dd n]$. We let $\overline{X}$ represent the reverse string of $X$, i.e. $x_{n-1} x_{n-2} \ldots x_0$. For two strings $X, Y$, $X \cdot Y$ is the concatenation of $X$ and $Y$. For a sequence of strings $X_1, X_2, \ldots, X_k$, $\bigodot_{i = 1}^k X_i = X_1 \cdot X_2 \cdot \ldots \cdot X_k$. We denote $\lcp(X, Y)$ as the length of the longest common prefix of $X$ and $Y$. Given two strings $X$ and $Y$, the \emph{edit distance} $\ed(X, Y)$ is the minimum number of insertions, deletions, and substitutions needed to transform $X$ into $Y$.

\subsection{Parenthesis Strings}
For a parenthesis string $X$, let the \emph{transpose} of $X$, denoted by $T(X)$, be the string corresponding to $\overline{X}$ with each opening parenthesis replaced by a closing parenthesis of the same type and each closing parenthesis replaced by an opening parenthesis of the same type.  For example, if $X = ``(([)["$, then $T(X) = ``](]))"$.
We define the height of a character in $X$ as the height of the index at which it occurs, where the height of an index $i \in [0\dd n]$, denoted $h_X(i)$, is defined as the difference in the number of opening and closing parentheses in $X[0\dd i)$.
A parenthesis string is \emph{monotone} if it is composed of only opening parentheses or only closing parentheses. For a string $X$, $\lcpt(X)$ denotes the longest monotone prefix of $X$.

A \emph{Dyck} sequence is a sequence of opening and closing parentheses that is well-parenthesized. We denote $\Dyck$ as the set of all Dyck sequences.

\begin{definition}\label{def:twins}
Let $X$ be a parenthesis string.
Indices $0 \le i < j < |X|$ are called \emph{twins} if $h_X(i) = h_X(j+1)$ and $h_X(j') > h_X(i)$ holds for all $j' \in (i\dd j]$.\footnote{Note that $X[i]$ is an opening parenthesis since $h_X(i+1) > h_X(i)$, and $X[j]$ is a closing parenthesis since $h_X(j) > h_X(j+1)$.}
We write $\tw_X(i) := j$ and $\tw_X(j) := i$ when $i$ and $j$ are twins.\footnote{When the parenthesis string $X$ is clear from context, we omit the subscript $X$ in $h_X(\cdot)$ and $\tw_X(\cdot)$.}
\end{definition}

The following observation is easy to see from the above definition.

\begin{observation}
\label{obs:noncrossingtwins}
The set of twins is \emph{non-crossing}; that is, for twins $i < i'$ and $j < j'$ with $i' = \tw_X(i)$ and $j' = \tw_X(j)$, it cannot occur that $i < j < i' < j'$.
\end{observation}

In this work, we consider the problem of Dyck edit distance, defined as follows.
\begin{definition}[Dyck edit distance]
    Given a parenthesis string $X$, the \emph{Dyck edit distance} of $X$, denoted as $\ded(X)$, is the minimum number of character deletions, insertions, and substitutions needed to place $X$ in $\Dyck$.
\end{definition}

\subsection{Tree Edit Distance Preliminaries}

\subsubsection{Alignments}
\label{sec:tree-prelim}
Recall that given two strings $X$ and $Y$, the edit distance $\ed(X, Y)$ is the minimum number of insertions, deletions, and substitutions needed to transform $X$ into $Y$.  Edit distance may be equivalently defined via the minimum cost \emph{alignment} $\A: X \rightarrow Y$. 

\begin{definition}[Alignments]
\label{def:alignment}
  A sequence $\A = (x_t,y_t)_{t=0}^m$ is an \emph{alignment}
  of a string $X\in \Sigma^*$ onto a string $Y\in\Sigma^*$, denoted $\A : X \onto Y$, if $(x_0,y_0)=(0,0)$,
  $(x_m,y_m)=(|X|,|Y|)$, and $(x_{t+1},y_{t+1})\in \{(x_t+1,y_t+1),(x_t+1,y_t),(x_t,y_t+1)\}$ for $t\in [0\dd m)$.
\end{definition}
Given an alignment $\A = (x_t,y_t)_{t=0}^m : X \onto Y$, for every $t\in [0\dd m)$:
\begin{itemize}
  \item If $(x_{t+1},y_{t+1})=(x_t+1,y_t)$, we say that $\A$ \emph{deletes} $X[x_t]$.
  \item If $(x_{t+1},y_{t+1})=(x_t,y_t+1)$, we say that $\A$ \emph{deletes} $Y[y_t]$.
  \item If $(x_{t+1},y_{t+1})=(x_t+1,y_t+1)$, we say that $\A$ \emph{aligns} $X[x_t]$ and $Y[y_t]$, denoted $X[x_t] \sim_\A Y[y_t]$. If~additionally $X[x_t]= Y[y_t]$, we say that $\A$ \emph{matches} $X[x_t]$ and $Y[y_t]$, denoted $X[x_t] \simeq_\A Y[y_t]$.
  Otherwise, we say that $\A$ \emph{substitutes} $X[x_t]$ for $Y[y_t]$.
\end{itemize}

The \emph{cost} of an edit distance alignment $\A$ is the total number characters that $\A$ deletes, inserts, or substitutes.
We denote the cost by $\ed_\A(X,Y)$.
The cost of an alignment $\A=(x_t,y_t)_{t=0}^m$ is at least its \emph{width}, defined as $\max_{t=0}^m |x_t-y_t|$.
Observe that $\ed(X,Y)$ can be defined as the minimum cost of an alignment of $X$ and~$Y$,
that is, $\ed(X,Y) = \min_{\A : X \onto Y} \ed_{\A}(X,Y)$.
An alignment of $X$ and $Y$ is \emph{optimal} if its cost is equal to $\ed(X, Y)$. We often consider a restricted version of edit distance, the \emph{deletion-only edit distance}, which is the minimum cost of an alignment of $X$ and $Y$ in which only deletions and matches are allowed (no substitutions are allowed), denoted $\ded(X, Y)$. It is easy to see that $\ded(X, Y) \leq 2 \cdot \ed(X, Y)$ as any substitution $(x_t, y_t), (x_t + 1, y_t + 1)$ can be replaced with two deletions $(x_t, y_t), (x_t + 1, y_t), (x_t + 1, y_t + 1)$. We define analogous deletion-only versions for Dyck and Tree edit distances as well, denoted $\dedd$ and $\dedt$, respectively.

We say that a set $M \subseteq \Zz \times \Zz$ is a \emph{monotone matching} if there do not exist distinct pairs $(x, y), (x', y') \in M$ such that $x \le x'$ and $y \ge y'$.
Observe that for every alignment $\A$ between $X$ and $Y$, the set $\{(x,y)\in [0\dd |X|)\times [0\dd |Y|): X[x]\sim_{\A} Y[y]\}$
is a monotone matching.

\subsubsection{Forests}
\label{subsubsec:forests}
A \emph{tree} $\T = (V, E, \lambda)$ is a directed acyclic graph with node set $V$, edge set $E$, and node labeling $\lambda: V \to \Sigma$ for some alphabet $\Sigma$, such that every node except the unique \emph{root} of $\T$ has exactly one incoming edge. The \emph{size} of $\T$, written $|\T|$, is $|V|$. 
For a node $u \in V$, we write $u \in \T$. If $(u, v) \in E$, then $u$ is the \emph{parent} of $v$ and $v$ is a \emph{child} of $u$; let $\chd(v)$ denote the set of children of $v$. A node $v$ is \emph{internal} if $|\chd(v)| \ge 1$, and a \emph{leaf} otherwise. 
For $v \in V$, let $P$ denote the unique path from the root to $v$ (inclusive). Every node on $P$ is an \emph{ancestor} of $v$, and $v$ is a \emph{descendant} of each such node. A path from the root to a leaf is called a \emph{branch} of $\T$. 
For $v \in V$, $\T(v)$ denotes the subtree of $\T$ rooted at $v$, i.e., the subgraph of $\T$ induced by the descendants of $v$.

Throughout this work, we consider \emph{ordered} trees, where the children of each node have a fixed left-to-right order. 
We also say that a \emph{forest} is a sequence of trees. For every node $v$ in a forest $\F$, we write $\F(v)$ for the subtree $\T(v)$ of the tree $\T$ containing $v$.

Given two forests $\F=(V_{\F}, E_{\F}, \lambda_{\F})$ and $\G=(V_{\G}, E_{\G}, \lambda_{\G})$, the \emph{tree edit distance} $\ted(\F, \G)$ is the minimum number of node insertions, node deletions, and node label substitutions needed to transform $\F$ into $\G$. In a forest $\F$, a \emph{node deletion} of a node $v \in \F$ with parent $u$ and children $w_1, w_2, \ldots, w_\ell$ removes $v$ from $\F$ and sets the parent of $w_1, w_2, \ldots, w_\ell$ to $u$. If $v$ has no parent, then each of $w_1, w_2, \ldots, w_\ell$ become root nodes of its own tree (its respective subtree) in $\F$. In a forest $\F$ containing node $u$ with children $w_1, w_2, \ldots, w_\ell$, a \emph{node insertion} adds a new node $v$ as the child of some node $u \in \F$ such that $u$ is to the right of $w_i$, to the left of $w_{j+1}$, and $u$ is set to be the parent of $w_{i+1}, w_{i+2}, \ldots, w_j$ for any $0 \leq i < j \leq \ell$. If $i = 0$ or $j = \ell$, then $u$ will have no left or right sibling, respectively, after insertion since there is no $w_0$ or $w_{\ell + 1}$ node. Furthermore, if a node $v$ is inserted as a new root of $\F$, $v$ becomes the child of the ``virtual root'' of $\F$, the parent of all root nodes in $\F$, and so $v$ will become the parent of a subset of the root nodes of $\F$.

Given a forest $\F$, a folklore mapping of $\F$ to a string, called the \emph{parenthesis representation} of $\F$ and denoted $\str{\F}$, is defined as follows.  
Let $\F$ consist of trees $\T_1, \T_2, \ldots, \T_m$ with respective root nodes $r_1, r_2, \ldots, r_m$, arranged from left to right so that $\T_1$ is the leftmost tree in $\F$. Then,
$\str{\F} = \bigodot_{i = 1}^m \texttt{(}_{\lambda(r_i)} \str{\T_i \setminus r_i} \texttt{)}_{\lambda(r_i)}$.  
In this definition, the symbols $\texttt{(}_\ell$ and $\texttt{)}_\ell$ represent parentheses in the constructed string (as opposed to ordinary algebraic parentheses), and the subscript $\ell$ denotes their \emph{type}, which matches only with parentheses having the same subscript.

For a node $u \in \F$, let $o_P(u)$ and $c_P(u)$ denote the indices of the opening and closing parentheses in $\str{\F}$ corresponding to $u$; note that $o_P(u)$ and $c_P(u)$ are twins (\cref{def:twins}).  
It is easy to verify inductively that, for any node $u$ in a forest $\F$, the substring of $\str{\F}$ between $o_P(u)$ and $c_P(u)$ corresponds exactly to the parenthesis representation of the subtree rooted at $u$, i.e., $\str{\F}[o_P(u) \dd c_P(u)] = \str{\F(u)}$.  
This yields the following useful observation.

\begin{observation}
\label{obs:paren-anc}
Let $\F$ be a forest with nodes $u$ and $v$. Then, $u$ is an ancestor of $v$ if and only if 
$o_P(u) \le o_P(v) < c_P(v) \le c_P(u)$.
\end{observation}

\begin{definition}[Tree Alignments]\label{def:ta}
    We say that an alignment $\A : \str{\F} \onto \str{\G}$ is a \emph{tree alignment} of forests $\F$ and $\G$, denoted by $\A : \F \onto \G$
    if the following \emph{consistency} conditions are satisfied for each $u\in V_\F$:
    \begin{itemize}
        \item either $\A$ deletes both $\str{\F}[o_P(u)]$ and $\str{\F}[c_P(u)]$, or
        \item there exists $v\in V_\G$ such that $\str{\F}[o_P(u)] \sim_\A \str{\G}[o_P(v)]$ and $\str{\F}[c_P(u)] \sim_\A \str{\G}[c_P(v)]$.
    \end{itemize}
\end{definition}

Note that $\ted(\F, \G)$ can be defined as the minimum cost tree alignment of $\str{\F}$ and $\str{\G}$, i.e., $2 \cdot \ted(\F, \G) = \min_{\A \in \ta(\F, \G)} \ted_{\A}(\F, \G)$.

\subsubsection{Heavy-Light Decompositions}
\label{subsec:heavy}

Given a forest $\F$, a heavy-light decomposition of $\F$ assigns a label \emph{heavy} or \emph{light} to each node of $\F$. For a given node $u \in \F$, a child node $v$ of $u$ is \emph{heavy} if $\heavy{\F(v)} = \heavy{\F(u)}$, and \emph{light} otherwise. A node with no parent is light by default. Here, for a node $u \in \F$, $\F(u)$ denotes the subtree of $\F$ rooted at vertex $u$, and $\heavy{\F(u)}$ denotes the \emph{heavy depth} of $u$ in $\F$. A set of heavy nodes in $\F$ is called a \emph{heavy path} if it forms a connected path in $\F$. The following useful observations are folklore:

\begin{observation}
\label{obs:heavy-light}
    Given a forest $\F$ with $|\F| = n$, any root-to-leaf path includes at most $\Oh(\log n)$ light nodes.
\end{observation}

\begin{proof}
    A node $v$ is light if it does not have the same heavy depth as its parent $u$.
    For nodes $u$ and $v$ such that $v$ is in the subtree of $u$, note that the heavy depth of $v$ is at most the heavy depth of $u$ since $|\F(v)| < |\F(u)|$. Therefore, for any root-to-leaf path, the sequence of heavy depth values corresponding to each node on the path is non-increasing.  Since the maximum heavy depth value is $\heavy{\F} = \heavy{n}$ and the minimum heavy depth value is 0, there are at most $\heavy{n} + 1 = \Oh(\log n)$ pairs of nodes $u$ and $v$ where $u$ is the parent of $v$ in a root-to-leaf path such that the heavy depth of $v$ is smaller than the heavy depth of $u$.  Therefore, there are at most $\Oh(\log n)$ light nodes in a root-to-leaf path.
\end{proof}

\begin{observation}
\label{obs:heavy-child}
    Given a forest $\F$, every node $u \in \F$ has at most one child node that is heavy.    
\end{observation}

\begin{proof}
    Assume for sake of contradiction that a node $u$ has at least two heavy children $v_1, v_2$. Let $2^{k+1} > |\F(u)| \geq 2^k$ for some $k \in \mathbb{Z}^+$.  Then, $\heavy{\F(v_1)} = \heavy{\F(v_2)} =k$, and so,
    \begin{align}
        |\F(v_1)| + |\F(v_2)| &\geq 
        2^{\heavy{\F(v_1)}} + 2^{\heavy{\F(v_2)}} = 2*2^k = 2^{k+1} > |\F(u)|.
        \label{ineq:heavychild}
    \end{align}
    Since $v_1$ and $v_2$ are both children of $u$, however, $|\F(u)| \geq |\F(v_1)| + |\F(v_2)|$, which is a contradiction with Inequality~\ref{ineq:heavychild}.  
\end{proof}

\section{Technical Overview}\label{sec:to}

\subsection{Dynamic Dyck Approximation Algorithm.} We start with a technical overview of our first algorithm that approximates Dyck edit distance in the dynamic setting by reducing the Dyck edit distance instance to a set of dynamic string edit distance instances. In the static setting, similar reductions were studied~\cite{Saha2014, KSSODA2023}, where the core idea was to decompose the Dyck edit distance instance into multiple string edit distance instances. Our new dynamic reduction begins with a similar decomposition concept but employs a new technique that diverges significantly from previous approaches. As \cite{KSSODA2023} directly decomposes the input string, even a single dynamic edit in their approach can affect multiple string edit distance instances, making it challenging for efficient dynamic maintenance (we discuss this further in the following paragraphs).
 
We introduce a new reduction framework that, instead of directly decomposing the input Dyck edit distance instance, first constructs a tree representation of the problem.
This tree is then decomposed into paths via a carefully designed heavy-light decomposition, which we subsequently convert back into pairs of strings.
These strings can be interpreted as string edit distance instances, allowing us to approximate the Dyck edit distance of the original parenthesis sequence.
Our heavy-light decomposition is specifically designed to support polylogarithmic maintenance time, yielding subpolynomial update time overall.
Although the steps in our new reduction differ entirely from those in the static reduction, we leverage one of its core analytical techniques to prove that our reduction achieves a logarithmic approximation factor in the static setting and a subpolynomial approximation factor in the dynamic setting.

We now provide a brief overview of the reduction and analysis framework of \cite{KSSODA2023} and discuss the challenges of maintaining the reduction dynamically. Following this, we introduce our new dynamic reduction and provide an overview of the technical details and analysis for our main Dyck edit distance result.

\paragraph{Reduction from Dyck to String Edit Distance in the Static Setting \cite{KSSODA2023}.} Given an input string $X$, the Dyck to string reduction of \cite{KSSODA2023} begins by constructing a collection $\C(X)$ of strings corresponding to subsequences of $X$ such that each string contains a sequence of opening parentheses followed by a sequence of closing parentheses of equal length.  Each string in $\C(X)$ corresponds to an instance of the string edit distance problem, and the resulting string edit distance instances are used to compute an approximation for the Dyck edit distance of $X$.

The construction of $\C(X)$ begins by partitioning $X$ according to its \emph{LR-decomposition} where each part of the LR-decomposition is a maximal-length disjoint substring composed of a sequence of opening parentheses followed by a sequence of closing parentheses, called \emph{LR-segments}. For each LR-segment, either a prefix of opening parentheses or a suffix of closing parentheses is trimmed so that the remaining string has an equal number of opening and closing parentheses. The resulting substring of $X$ is added to $\C(X)$ and removed from $X$ for the remainder of the construction algorithm.  After each LR-segment is processed, the input string $X$ is comprised of exactly the trimmed prefixes and suffixes. This process is then repeated iteratively until only a sequence of opening or a sequence of closing parentheses remains in $X$.  See Table~\ref{tab:C-example} for an example of the construction of $\C(X_1)$ for a string $X_1$.

The collection $\C(X)$ constructed as above corresponds to a set of string edit distance instances that are used to approximate the Dyck edit distance of $X$. For each LR-segment in $\C(X)$, recall that the segment contains a sequence of opening parentheses followed by a sequence of closing parentheses. 
Each segment is split in half, the closing parenthesis sequence is reversed, and each closing parenthesis is replaced with an opening parenthesis of the same type.\footnote{e.g., segment ``(([((\}))])'' would become pair ``(([(('', ``([((\{'' .} The sum of the edit distances between each pair of halves for each string of the LR-decomposition yields an $\Oh(\log n)$-approximation of the Dyck edit distance for the original input string $X$.

\begin{table}
    \centering
    \begin{tabular}{c|c|c}
        Iteration & $X_1$ & $\C(X_1)$ \\ \hline
        1 & \texttt{\textcolor{blue}{(([(]}\textcolor{red}{([]]]}\textcolor{blue}{(()}\textcolor{red}{(()))]]}} & \texttt{(], ([]], (), (())} \\ \hline
        2 & \texttt{\textcolor{blue}{(([]}\textcolor{red}{()]]}} & \texttt{(], ([]], (), (()), [], ()} \\ \hline
        3 & \texttt{\textcolor{blue}{((]]}} & \texttt{(], ([]], (), (()), [], (), ((]]}
    \end{tabular}
    \caption{The above table contains an example of the construction routine for $\C(X_1)$ given string $X_1$. The red and blue parenthesis coloring separate the LR-segments of $X_1$ in each iteration.}
    \label{tab:C-example}
\end{table}

\begin{table}
    \centering
    \begin{tabular}{c|c|c}
        Iteration & $X_2$ & $\C(X_2)$ \\ \hline
        1 & \texttt{\textcolor{red}{(([(([]]]}\textcolor{blue}{(()}\textcolor{red}{(()))]]}} & \texttt{(([]]], (), (())} \\ \hline
        2 & \texttt{\textcolor{red}{(([()]]}} & \texttt{(([]]], (), (()), ([()]]} \\ \hline
        3 & \texttt{(} & \texttt{(([]]], (), (()), ([()]]}
    \end{tabular}
    \caption{The above table contains an example of the construction routine for $\C(X_2)$ given string $X_2$. The red and blue parenthesis coloring separate the LR-segments of $X_2$ in each iteration. $X_1$ and $X_2$ differ by a single parenthesis, $X_1[4]$ is not present in $X_2$. However, $\C(X_1)$ and $\C(X_2)$ differ in over half of their strings.}
    \label{tab:C-example-2}
\end{table}

The core concept behind the LR-decomposition comes from a natural observation. Note that the Dyck edit distance can be defined as the number of deletions or substitutions needed to balance a parenthesis string. In a balanced parenthesis string, every parenthesis has a ``match'': the parenthesis with opposite parity (open or closed), of the matching type and matching height. Here the height of a parenthesis is defined as the difference between the number of opening and closing parentheses preceding it, see Figure~\ref{fig:parenthesis-tree}a. Consider an opening parenthesis followed immediately by a closing parenthesis contained in a parenthesis string. If these parentheses do not match after all edits in an optimal sequence, one of the two parentheses must have been deleted. Instead, the sequence could be modified to match these two parentheses, which may require a substitution to align their types and an additional deletion for a later parenthesis that was intended to match with one of them. This adjustment would require 2 edits instead of the original 1 to handle this pair of parentheses.

This simple idea can be extended to LR-segments of $X$ that contain an equal number of opening and closing parentheses.  If $d$ of the opening parentheses in an LR-segment are not matched to the following sequence of $d$ closing parentheses after performing an optimal sequence of edits, then at least $d$ deletions must take place on this LR-segment. These $d$ deletions can be replaced with $d$ substitutions with an additional $d$ deletions of the later parentheses to obtain a 2-approximation for the number of edits required to balance the parentheses in this LR-segment. By computing this approximation for each segment of the LR-decomposition of $X$, a 2-approximation of the number of edits needed to balance all characters contained in these segments is found excluding a prefix of opening parentheses or a suffix of closing parentheses per segment. The same steps are repeated iteratively on the string comprised of the remaining prefixes and suffixes until no more LR-segments are contained in $X$. Since over half of the string is removed in each iteration, there are $\Oh(\log n)$ iterations. Determining the substitutions and deletions needed for each LR-segment can be done by computing the string edit distance between the two halves of the LR-segment. Thus, the LR-decomposition provides a way to achieve a $2 \cdot \Oh(\log n)$-approximation for the Dyck edit distance.

\paragraph{Reduction from Dyck to String Edit Distance in the Dynamic Setting.} The key challenge in implementing this reduction dynamically is maintaining $\C(X)$ as the input string $X$ undergoes updates. Reconstructing the collection of strings after each update following the above construction method is too time-consuming. Even a single pass to parse the entire string to compute the LR-decomposition, and then to trim prefixes and suffixes would require linear time in the number of LR-segments, which could be $\Oh(n)$ in the worst case.

Furthermore, a single parenthesis insertion or deletion in $X$ can change a super-constant number of strings in $\C(X)$, see Table~\ref{tab:C-example-2} for an example of how deleting a single parenthesis greatly changes $\C(X)$. A key challenge, therefore, lies in structuring $\C(X)$ to efficiently identify and update only the affected strings. To address these challenges, rather than constructing a collection of subsequences of  $X$ directly, we first represent $X$ with a tree $T_X$. We then re-frame the reduction by representing $X$ through a collection of paths within this tree $T_X$.

To build a tree $T_X$ from $X$, we begin by pairing together the opening parentheses and closing parentheses of equal height, which we refer to as \emph{twins}. For each opening parenthesis, its twin is the closest closing parenthesis to its right that has the same height. Note, that this procedure defines a one-to-one mapping between opening and closing parentheses in $X$.  Each pair of twins corresponds to a node in the tree. 
For an opening parenthesis at index $i$ in $X$, referred to as $X[i]$, its corresponding node in the tree is the child of the node that corresponds to the nearest opening parenthesis $X[j]$ with $j < i$ and of height exactly 1 less than that of $X[i]$. 
Any opening parentheses with no parent represent the root of their own tree. 
In this tree representation of $X$, each path corresponds to an LR-segment of $X$ by arranging the sequence of opening parentheses, followed by the sequence of closing parentheses, that corresponds to the nodes appearing in the path.
See Figure~\ref{fig:parenthesis-tree}b for an example of $T_X$. 

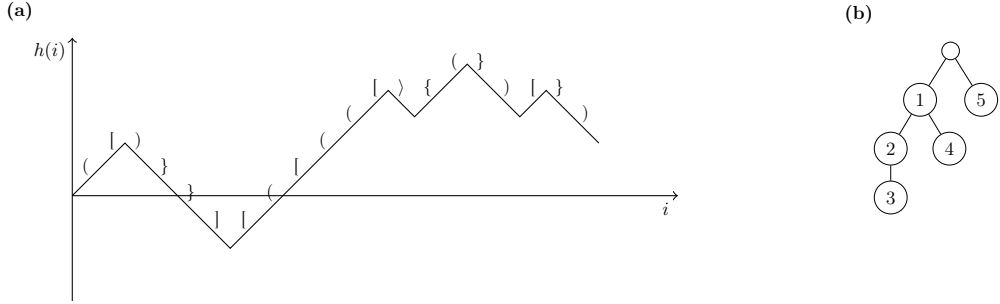
\begin{figure}
    \centering

    \scalebox{1}{\begin{tikzpicture}[scale=0.35,every node/.style={scale=0.65}]
%\path[use as bounding box] (-1,-4) rectangle (23,6);
\draw[->] (0,-4)--(0,6) node[pos=0.95,left] {{$h(i)$}}; 
\draw[->] (0,0)--(23,0) node[pos=0.98,below] {{$i$}};

\foreach \height/\para [count=\xi] in {1/(,2/[,2/),1/\},0/\},-1/],-1/[,0/(,1/[,2/(,3/(,4/[,4/$\rangle$,4/\{,5/(,5/\},4/),4/[,4/\},3/)}
	\node at (\xi-0.5,\height+0.1) {\small{\para}};  
	
\draw (0,0)--(2,2)--(6,-2)--(12,4)--(13,3)--(15,5)--(17,3)--(18,4)--(20,2);
\draw (-2, 7) node {\textbf{(a)}};
\end{tikzpicture}}\hspace{2cm}\scalebox{.65}{\begin{tikzpicture}
    \draw (-1.25, 1.75) node {\textbf{(b)}};
    \draw (.625, 1) node[circle, draw] (p0) {};
    \draw (0, 0) node[circle,draw] (p1) {1};
    \draw (-.6, -1) node[circle,draw] (p2) {2};
    \draw (-.6, -2) node[circle,draw] (p3) {3};
    \draw (.6, -1) node[circle,draw] (p4) {4};
    \draw (1.25, 0) node[circle,draw] (p5) {5};
    \draw (p1) -- (p4);
    \draw (p1) -- (p2);
    \draw (p2) -- (p3);
    \draw (p0) -- (p1);
    \draw (p0) -- (p5);
    \draw (0, -4) node {};
\end{tikzpicture}}
    \caption{(a) Example of a graph of heights of parentheses in a given input string from \cite{DBLP:conf/soda/FriedGKKPS22}. (b) The tree representation $T_X$ of parenthesis string $X = $ $\texttt{(}_1$$\texttt{(}_2$$\texttt{(}_3$$\texttt{)}_3$$\texttt{)}_2$$\texttt{(}_4$$\texttt{)}_4$$\texttt{)}_1$$\texttt{(}_5$$\texttt{)}_5$ where matching indices in the subscripts of the parentheses indicate pairs of twins. The path containing nodes labeled 1, 2, 3 corresponds to string $\texttt{(}_1$$\texttt{(}_2$$\texttt{(}_3$$\texttt{)}_3$$\texttt{)}_2$$\texttt{)}_1$.}
    \label{fig:parenthesis-tree}
\end{figure}

To reduce Dyck to string edit distance with our tree representation of $X$, we first perform a heavy-light decomposition on $T_X$. We recall the definition of heavy-light decomposition here.

\begin{definition}[Heavy-Light Decomposition]
\label{def:heavy-light}
A \emph{heavy-light decomposition} of a forest $\F$ assigns each node a label of \emph{heavy} or \emph{light}.
A node $u \in \F$ with parent $v$ is heavy if $\heavy{\F(u)} = \heavy{\F(v)}$, and light otherwise; roots are light by default.
Here, $\F(u)$ denotes the subtree rooted at $u$, and $\heavy{\F(u)}$ its \emph{heavy depth}.
A maximal sequence of heavy nodes forming a path is called a \emph{heavy path}.
\end{definition}

As in a typical heavy-light decomposition, every node in a forest has at most one heavy child. We also make one other useful observation that is folklore for heavy-light decompositions.

\begin{observation}
\label{obs1:heavy-light}
    Given a forest $\F$ with $|\F| = n$, any root-to-leaf path includes at most $\Oh(\log n)$ light nodes.
\end{observation}

Once the heavy-light decomposition is obtained, we partition the tree into paths based on the heavy-light status of each node, ensuring that each path consists solely of heavy nodes except for the starting node, which is a light node.
Since each path of the tree corresponds to a sequence of opening and closing parentheses of matching height, the core principle of the reduction from~\cite{KSSODA2023} still holds; for each string corresponding to these heavy node paths, we can find a 2-approximation for the number of edits needed to balance the parentheses in this string using the corresponding string edit distance estimation. We compute the string edit distance between the opening and closing parenthesis sequences for each path, and as a heavy-light decomposition ensures each root-to-leaf path has at most $\Oh(\log n)$ light nodes, the 2-approximation factor only compounds additively $\Oh(\log n)$ times. Thus, the re-framed Dyck to string reduction using the tree still provides an $\Oh(\log n)$-approximation.

\paragraph{New Dynamic Heavy-light Decomposition.} 
With our new tree representation, we have shifted the challenge of the dynamic approach from maintaining $\C(X)$ to maintaining the heavy-light decomposition of the tree representation of $X$ in the dynamic setting.
In a traditional heavy-light decomposition, which typically designates the child with the largest subtree as heavy while categorizing the remaining children as light, a single dynamic edit to the tree can result in many nodes changing their type from heavy to light and vice versa.
By defining our heavy-light decomposition based on the logarithm of the size of the subtrees, a dynamic update to $X$ creates a structured impact on the nodes of the tree, which our algorithm utilizes. This approach ensures that any root-to-leaf path contains only $\mathcal{O}(\log n)$ light nodes, which is consistent with the conventional heavy-light decomposition.
We observe that when a parenthesis is inserted or deleted at index $i$, all the parentheses appearing at an index greater than $i$ have their height changed by 1. Further, any pair of twin parentheses, where both indices are greater than $i$, have their height shifted and thus, these pairs are unaffected. Thus, the only pairs of twins, with opening parenthesis appearing at an index $j_1 < i$ and closing parenthesis appearing at an index $j_2 > i$ are affected and need to be handled. These are exactly the ancestor nodes of the node corresponding to $X[i]$. Since any root-to-leaf path can contain at most $\Oh(\log n)$ heavy paths, we only need to update a logarithmic number of heavy paths in the path decomposition of $T_X$ for each edit. Moreover, since the heights of the parentheses only change by 1, we are primarily concerned with a shift of just 1 node between adjacent affected heavy paths, specifically the first or last nodes of these heavy paths. These are specifically the nodes that may change their status from heavy to light or vice versa. 
For each of these first and last nodes of affected heavy paths, we recompute their heavy-light status by comparing the size of their subtree to that of their parent’s subtree. Since we only need to check $\Oh(1)$ starting and ending nodes per heavy path across $\Oh(\log n)$ heavy paths, we may maintain the entire heavy-light decomposition in polylogarithmic time.

When a dynamic update occurs on $X$, we can identify the set of heavy paths that contain the ancestor nodes of the node corresponding to the inserted or deleted parenthesis. Subsequently, we can adjust the first and last nodes of these heavy paths as needed, all within logarithmic time.
 For all the affected paths, we then perform $\Oh(1)$ edits on the corresponding string edit distance instances according to these shifts. By utilizing a dynamic string edit distance algorithm that can handle dynamic insertions and deletions in sub-polynomial time, we recompute the distances for these updated instances and then update our Dyck edit distance approximation with these new distances.

We highlight that our proposed dynamic heavy-light decomposition introduces significantly new ideas, allowing heavy paths to remain intact without needing to be split or merged. Traditional dynamic heavy-light decompositions, such as those in link-cut trees, lack this property and could not be applied in our reduction (since the dynamic string edit distance approximation does not support splits and concatenations). Given the heavy-light decomposition’s broad utility, we believe our approach has potential for further applications and could be of independent interest.

\subsection{Dynamic Approximate Tree to String Edit Distance}
We now concentrate our attention on an overview of the $\tilde{O}(\sqrt{n})$-approximation algorithm for dynamic tree edit distance. Our method begins with a new $\tilde{O}(\sqrt{n})$-approximate reduction from tree to string edit distance in the dynamic setting. This reduction enables us to use the $\Oh(b \log_b n)$-approximation dynamic string edit distance algorithm where $2 \leq b \leq n$ is an integer parameter we may set, which has an update time of $\Oh(b^2 (\log n)^{\Oh(\log_b n)})$~\cite{KMS2023} to achieve our final result. We start by discussing the new reduction in the static setting and then outline the additional techniques required to implement it dynamically.

\paragraph{Reduction from Tree to String Edit Distance in the Static Setting.}
In the static setting, the state-of-the-art reduction from tree to string edit distance, by Akutsu, Fukagawa, and Takasu~\cite{AFT10}, achieves an $\tilde{O}(n^{3/4})$ approximation. Despite over a decade passing, no improvement on this bound has been proposed; moreover, their reduction only applies to trees of constant degree. Until now, it remained open whether a reduction for trees of unbounded degree could achieve approximation loss $o(n)$.
We not only significantly improve the previous bound, providing a reduction with approximation $\tilde{O}(\sqrt{n})$, but also show that it works for all trees and can be maintained under dynamic updates.
In~\cite{AFT10}, two rooted labeled trees $\T$ and $\T'$ are converted into strings by labeling edges based on specific structural properties. A node $v$ is defined as \emph{special} if its subtree is small while its parent’s subtree is large. Labels are then assigned to edges according to the presence of special children: each non-special edge $(u,v)$ is labeled by $v$, while each special edge receives a unique label that encodes information about the subtrees of $v$’s special children. After labeling, $\T$ and $\T'$ are transformed into strings $S$ and $S'$ via an Euler tour. An optimal string alignment $\A$ between $S$ and $S'$ is then computed, but converting $\A$ into a valid tree alignment incurs an approximation loss~of~$\tilde{O}(n^{3/4})$.

Our algorithm is inspired by a similar framework but introduces a new approach from the outset, beginning with the label generation step. Instead of assigning labels to edges, our method assigns labels to nodes using a heavy-light decomposition, rather than the heavy subtree approach in~\cite{AFT10}. Incorporating heavy-light decomposition in label generation gives us more flexibility in exploiting the structural properties of the Euler tour representation, allowing the optimal string alignment to be transformed directly into a valid tree alignment with an improved approximation guarantee. Our string-to-tree alignment step is also entirely new and simpler than that in~\cite{AFT10}. Further, the cost analysis of this step ensures a $\tilde{O}(\sqrt{n})$ approximation and introduces several foundational techniques and concepts, serving as a core contribution of our work.

\subparagraph*{Parentheses representation.} 
Our reduction begins by performing a heavy-light decomposition (\cref{def:heavy-light}) on the input trees $\T$ and $\T'$; recall that a node is \emph{light} if it does not share the same heavy depth as its parent.
Each node $v$ is assigned a label encoding information about the subtree rooted at $v$, excluding its heavy child, if any, and the subtree rooted at that child, which is represented by a dummy character. We then perform an Euler tour of each tree using these labels, producing strings that serve as instances for string edit distance (see formal definition in Subsection~\ref{subsubsec:forests} above).

\subparagraph*{String to tree alignment.} The second step begins by computing an optimal alignment $\A$ between $\str{\T}$ and $\str{\T'}$. A key observation here is that, under $\A$, a match between a node $u \in \T$ and a node $u' \in \T'$ satisfies the tree alignment constraints if $\A$ aligns both $\str{\T}[o(u)]$ with $\str{\T'}[o(u')]$ and $\str{\T}[c(u)]$ with $\str{\T'}[c(u')]$. Nodes that do not meet this condition are referred to as misaligned and thus violate the tree alignment constraints.

For a node $u \in \T$, misalignments can be categorized as follows: (i) $\A$ aligns only one of the parentheses corresponding to $u$ while deleting the other, making $u$ \emph{partially deleted}; (ii) $\A$ aligns $o(u)$ with $o(v_1)$ and $c(u)$ with $c(v_2)$, where $v_1$ and $v_2$ lie on the same root-to-leaf path in $\T'$, classifying $u$ as a \emph{single-branch} misaligned node; (iii) similar to (ii), but $v_1$ and $v_2$ appear on different root-to-leaf paths in $\T'$, making $u$ a \emph{multi-branch} misaligned node.

Next, we discuss how our algorithm $\mathsf{TreeAlign}$ transforms $\A$ to eliminate all misaligned nodes, resulting in a valid tree alignment. If a misaligned node $u$ is either light, partially deleted, or a multi-branch misaligned node, we update $\A$ by removing the matches related to $u$.  
For a heavy single-branch misaligned node, more complex processing is required. One common structure in single-branch misaligned nodes involves a sequence of single-branch nodes alternating between $\T$ and $\T'$ where each node aligns only one of its parentheses with the previous node in the sequence and its other parenthesis with the next node. We formally define this as a \emph{chain}.
To handle a heavy single-branch misaligned node $ u $, we modify $\A$ by removing the matches associated with $ u $ if $ u $ is not part of a chain, if it appears near either end of a chain, or if there is a large subtree (size $\ge \sqrt{n}$) rooted at $ u $.
Otherwise, we adjust $\A$ to align the subtree rooted at $ u $ (including $ u $) with the subtree rooted at $ v $ (including $ v $), where $ v \in \T' $ is the node following $ u $ in the chain sequence, while maintaining $\A$ as a valid alignment. It becomes fairly intuitive that after these modifications, no nodes remain misaligned under $\A$, allowing $\A$ to be interpreted as a valid tree alignment between $\T$ and $\T'$.

\vspace{1mm}
\textbf{Cost analysis.} We now prove that performing $\mathsf{TreeAlign}$ increases the number of edits in alignment $\A$ by at most a factor of $\tOh(\sqrt{n})$. Before, we delve into details, we present example instances of tree edit distance that shed light on the functioning of our reduction and give an intuitive idea of why we achieve a $\sqrt{n}$ approximation factor.

\begin{figure}[h!]
    \centering
        \scalebox{.6}{\begin{tikzpicture}
    \draw (-2, .75) node {\textbf{\Large{(a)}}};
    \draw (-1, .5) node {$T_1$};
    \draw (0, 0) node[circle,draw] (p1) {};
    \draw (0, -2) node[circle,draw] (p2) {};
    \draw (0, -4) node[circle, draw] (p3) {};
    \draw (0, -7) node[circle, draw] (p4) {};
    \draw (p1) -- (p2);
    \draw (p2) -- (p3);
    \draw[dashed] (p3) -- (p4);
    \draw (p1) -- (-1, -.8)
    node [isosceles triangle,minimum width=.75cm,draw, rotate=90, anchor=apex] {};
    \draw (-1, -1.4) node {$A_1$};
    \draw (p1) -- (1, -.8)
    node [isosceles triangle,minimum width=.75cm, draw,rotate=90, anchor=apex] {};
    \draw (1, -1.4) node {$B_1$};
    \draw (p2) -- (-1, -2.8)
    node [isosceles triangle,minimum width=.75cm,draw, rotate=90, anchor=apex] {};
    \draw (-1, -3.4) node {$A_2$};
    \draw (p2) -- (1, -2.8)
    node [isosceles triangle,minimum width=.75cm, draw,rotate=90, anchor=apex] {};
    \draw (1, -3.4) node {$B_2$};
    \draw (p3) -- (-1, -4.8)
    node [isosceles triangle,minimum width=.75cm,draw, rotate=90, anchor=apex] {};
    \draw (-1, -5.4) node {$A_3$};
    \draw (p3) -- (1, -4.8)
    node [isosceles triangle,minimum width=.75cm, draw,rotate=90, anchor=apex] {};
    \draw (1, -5.4) node {$B_3$};
    \draw (p4) -- (-1, -7.8)
    node [isosceles triangle,minimum width=.75cm,draw, rotate=90, anchor=apex] {};
    \draw (-1, -8.4) node {$A_i$};
    \draw (p4) -- (1, -7.8)
    node [isosceles triangle,minimum width=.75cm, draw,rotate=90, anchor=apex] {};
    \draw (1, -8.4) node {$B_i$};
\end{tikzpicture}\hspace{2cm}\begin{tikzpicture}
    \draw (-1, .5) node {$T_2$};
    \draw (0, 0) node[circle,draw] (p1) {};
    \draw (0, -2) node[circle,draw] (p2) {};
    \draw (0, -4) node[circle, draw] (p3) {};
    \draw (0, -7) node[circle, draw] (p4) {};
    \draw (p1) -- (p2);
    \draw (p2) -- (p3);
    \draw[dashed] (p3) -- (p4);
    \draw (p1) -- (-1, -.8)
    node [isosceles triangle,minimum width=.75cm,draw, rotate=90, anchor=apex] {};
    \draw (-1, -1.4) node {$A_2$};
    \draw (p1) -- (-2, -.8)
    node [isosceles triangle,minimum width=.75cm,draw, rotate=90, anchor=apex] {};
    \draw (-2, -1.4) node {$A_1$};
    \draw (p1) -- (1, -.8)
    node [isosceles triangle,minimum width=.75cm, draw,rotate=90, anchor=apex] {};
    \draw (1, -1.4) node {$B_1$};
    \draw (p2) -- (-1, -2.8)
    node [isosceles triangle,minimum width=.75cm,draw, rotate=90, anchor=apex] {};
    \draw (-1, -3.4) node {$A_3$};
    \draw (p2) -- (1, -2.8)
    node [isosceles triangle,minimum width=.75cm, draw,rotate=90, anchor=apex] {};
    \draw (1, -3.4) node {$B_2$};
    \draw (p3) -- (-1, -4.8)
    node [isosceles triangle,minimum width=.75cm,draw, rotate=90, anchor=apex] {};
    \draw (-1, -5.4) node {$A_4$};
    \draw (p3) -- (1, -4.8)
    node [isosceles triangle,minimum width=.75cm, draw,rotate=90, anchor=apex] {};
    \draw (1, -5.4) node {$B_3$};
    \draw (p4) -- (1, -7.8)
    node [isosceles triangle,minimum width=.75cm, draw,rotate=90, anchor=apex] {};
    \draw (1, -8.4) node {$B_i$};
\end{tikzpicture}\hspace{1cm}
        \begin{tikzpicture}
    \draw (-2, .75) node {\textbf{\Large{(b)}}};
    \draw (-1, .5) node {$T_1$};
    \draw (0, 0) node[circle,draw] (p1) {};
    \draw (0, -2) node[circle,draw] (p2) {};
    \draw (0, -4) node[circle, draw] (p3) {};
    \draw (0, -7) node[circle, draw] (p4) {};
    \draw (p1) -- (p2);
    \draw (p2) -- (p3);
    \draw[dashed] (p3) -- (p4);
    \draw (p1) -- (-1, -.8)
    node [isosceles triangle,minimum width=.75cm,draw, rotate=90, anchor=apex] {};
    \draw (-1, -1.4) node {A};
    \draw (p1) -- (1, -.8)
    node [isosceles triangle,minimum width=.75cm, draw,rotate=90, anchor=apex] {};
    \draw (1, -1.4) node {B};
    \draw (p2) -- (-1, -2.8)
    node [isosceles triangle,minimum width=.75cm,draw, rotate=90, anchor=apex] {};
    \draw (-1, -3.4) node {A};
    \draw (p2) -- (1, -2.8)
    node [isosceles triangle,minimum width=.75cm, draw,rotate=90, anchor=apex] {};
    \draw (1, -3.4) node {B};
    \draw (p3) -- (-1, -4.8)
    node [isosceles triangle,minimum width=.75cm,draw, rotate=90, anchor=apex] {};
    \draw (-1, -5.4) node {A};
    \draw (p3) -- (1, -4.8)
    node [isosceles triangle,minimum width=.75cm, draw,rotate=90, anchor=apex] {};
    \draw (1, -5.4) node {B};
    \draw (p4) -- (-1, -7.8)
    node [isosceles triangle,minimum width=.75cm,draw, rotate=90, anchor=apex] {};
    \draw (-1, -8.4) node {A};
    \draw (p4) -- (1, -7.8)
    node [isosceles triangle,minimum width=.75cm, draw,rotate=90, anchor=apex] {};
    \draw (1, -8.4) node {B};
\end{tikzpicture}\hspace{2cm}\begin{tikzpicture}
    \draw (-1, .5) node {$T_2$};
    \draw (0, 0) node[circle,draw] (p1) {};
    \draw (0, -2) node[circle,draw] (p2) {};
    \draw (0, -4) node[circle, draw] (p3) {};
    \draw (0, -7) node[circle, draw] (p4) {};
    \draw (p1) -- (p2);
    \draw (p2) -- (p3);
    \draw[dashed] (p3) -- (p4);
    \draw (p1) -- (-1, -.8)
    node [isosceles triangle,minimum width=.75cm,draw, rotate=90, anchor=apex] {};
    \draw (-1, -1.4) node {A};
    \draw (p1) -- (-2, -.8)
    node [isosceles triangle,minimum width=.75cm,draw, rotate=90, anchor=apex] {};
    \draw (-2, -1.4) node {A};
    \draw (p1) -- (1, -.8)
    node [isosceles triangle,minimum width=.75cm, draw,rotate=90, anchor=apex] {};
    \draw (1, -1.4) node {B};
    \draw (p2) -- (-1, -2.8)
    node [isosceles triangle,minimum width=.75cm,draw, rotate=90, anchor=apex] {};
    \draw (-1, -3.4) node {A};
    \draw (p2) -- (1, -2.8)
    node [isosceles triangle,minimum width=.75cm, draw,rotate=90, anchor=apex] {};
    \draw (1, -3.4) node {B};
    \draw (p3) -- (-1, -4.8)
    node [isosceles triangle,minimum width=.75cm,draw, rotate=90, anchor=apex] {};
    \draw (-1, -5.4) node {A};
    \draw (p3) -- (1, -4.8)
    node [isosceles triangle,minimum width=.75cm, draw,rotate=90, anchor=apex] {};
    \draw (1, -5.4) node {B};
    \draw (p4) -- (1, -7.8)
    node [isosceles triangle,minimum width=.75cm, draw,rotate=90, anchor=apex] {};
    \draw (1, -8.4) node {B};
\end{tikzpicture}}
    \caption{Two instances the tree edit distance problem. (a) Trees $T_1$ and $T_2$ with $A_i \ne A_j$, $B_i \ne B_j$ for $i \ne j$.  (b) Trees $T_1$ and $T_2$ with repeating subtrees $A$ and $B$ to the left and right, respectively, of the center path.}
    \label{fig:tree-intuition}
\end{figure}
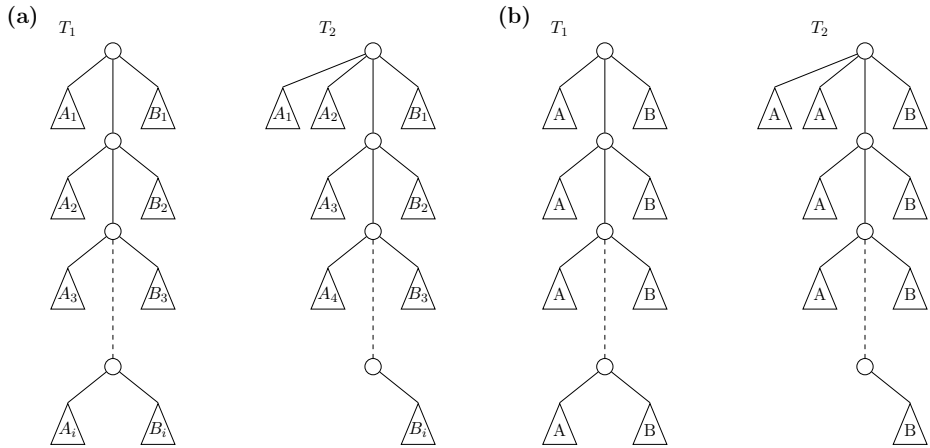

\textbf{Examples illustrating the Reduction.}
In Figure~\ref{fig:tree-intuition}a, both trees $T_1$ and $T_2$ have subtrees labeled $A_1, A_2, \ldots, A_i, B_1, B_2, \ldots, B_i$, while Figure~\ref{fig:tree-intuition}b presents subtrees simply labeled as $A$ and $B$ (without indexing). We refer to the central path through the tree as the \emph{spine}. Furthermore, we assume that for all $i \ne j$, $A_i \ne A_j$ and $B_i \ne B_j$. 

First, let's study the pair of trees depicted in Figure~\ref{fig:tree-intuition}a. Ignoring any special labeling, the conventional Euler tour string representations of $T_1$ and $T_2$ are given as follows: $T_1$ is represented as $(A_1 ( A_2 ( A_3 \ldots B_3 ) B_2 ) B_1 )$, and $T_2$ as $(A_1 A_2 ( A_3 ( A_4 ( \ldots ) B_3 ) B_2 ) B_1 )$, where each pair of parentheses ``('', ``)'' corresponds to a node along the spines of the trees. It is evident that only a few deletions are necessary to make these two strings similar. However, the only way to achieve similarity between $T_1$ and $T_2$ is either by deleting all the nodes in the spines or by removing all the subtrees attached to the spines. If the spine is quite long, for instance, if $i = \Oh(n)$, then the tree edit distance is $\Oh(n)$, whereas the string edit distance is $\Oh(1)$. Figure~\ref{fig:tree-intuition}a serves as a representative example of the challenges we tackle using our label embedding technique. By embedding subtree information in the labels of the nodes, the resulting Euler tour strings for $T_1$ and $T_2$ are respectively $ (_{A_1, B_1} A_1 (_{A_2, B_2} A_2 \ldots $ and $ (_{A_1, A_2, B_1} A_1 A_2 (_{A_3, B_2} A_3 \ldots $. Thus, for these strings, any optimal sequence of edits must remove all parentheses corresponding to the nodes on the spine, as each one carries a distinct label after the embedding process. Consequently, the gap between the string and tree edit distances reduces to a mere $\Oh(1)$.

However, we can still encounter a gap of $\Oh(\sqrt{n})$ in certain cases. For instance, if we modify the pair of trees to have more periodicity, as shown in Figure~\ref{fig:tree-intuition}b, the resulting strings are $ (_{A, B} A (_{A, B} A \ldots $ and $ (_{A, A, B} A A (_{A, B} A (_{A, B} \ldots $. In this scenario, only $\Oh(1)$ deletions are necessary to make these strings equal. 
 Conversely, to make $T_1$ and $T_2$ identical trees, we have two options: either delete all the nodes in the spine, as previously mentioned or remove a few copies of $A$ and $B$. It is important to note that in Figure~\ref{fig:tree-intuition}a, we were required to delete all $A_i$ and $B_i$, whereas, in the periodic trees, we may only need to delete $\Oh(1)$ copies of them.
If $|A|, |B| < \sqrt{n}$, then deleting $\Oh(1)$ copies of $A$ and $B$ will require only $\Oh(\sqrt{n})$ deletions. On the other hand, if $|A|, |B| \geq \sqrt{n}$, we know that the length of the spine can be at most $n/|A| = \Oh(\sqrt{n})$. Thus, deleting the spine nodes will also result in $\Oh(\sqrt{n})$ deletions. In both scenarios, the string edit distance remains $\Oh(1)$, while the tree edit distance is $\Oh(\sqrt{n})$. In Section~\ref{subsec:tree-reduction}, we formally establish that this is the maximum gap attainable through our reduction. We show that for any optimal string alignment between the Euler tour strings generated by our reduction, we can construct a corresponding tree alignment for the original input trees with at most $\tOh(\sqrt{n})$ additional deletions.

In the remainder of the cost analysis overview, we give more details that helps to establish the formal bound in \cref{subsec:tree-reduction}.

If a node $u$ is partially deleted, since the string alignment $\A$ already incurs a cost of 1 for deleting one of its parentheses, we can modify $\A$ to delete the other parenthesis as well, thereby only doubling the cost.

Next, we address the processing of misaligned nodes that appear in a chain. We make two key observations. First, by definition, any chain consists of either all light nodes or all heavy nodes. Second, for every maximal chain, since $\A$ aligns the opening and closing parentheses of each node with at least a one-position shift, it incurs at least one deletion. Thus, the cost of $\A$ is at least one for every maximal chain. We begin by considering chains of light nodes. Since the length of such a chain is at most $\Oh(\log n)$, deleting all light nodes may increase the cost by at most a factor of $\Oh(\log n)$.

We now analyze the cost of handling a chain that consists solely of heavy nodes. First, consider a heavy node $ u $ for which the subtree rooted at $ u $, excluding the subtree of its heavy child if one exists, is large (i.e., of size $ \ge \sqrt{n} $). The number of such heavy nodes is bounded by $ \sqrt{n} $. Therefore, removing all these nodes (i.e., their corresponding matches from $ \A $) incurs a cost of at most $ \Oh(\sqrt{n}) $. Since each chain contributes at least one deletion under $ \A $, the overall increase in cost is at most a factor of $ \Oh(\sqrt{n}) $.
Therefore, we now need to address the heavy nodes that have small subtrees, excluding the subtree of its heavy child (if it exists), rooted at them. If such a node $ u $ does not appear close to either end of the chain, we modify $ \A $ to align the subtree rooted at $ u $ (including $ u $) with the subtree rooted at $ v $ (including $ v $), where $ v \in \T' $ is the node that follows $ u $ in the chain sequence. This adjustment maintains the same cost.
Alternatively, for a constant number of nodes that root small subtrees and are located at either end of a chain, we delete the subtree that is rooted at those nodes. Since this is performed only a constant number of times for each chain, and the size of each subtree is at most $ \sqrt{n} $, the total cost incurred remains $ \Oh(\sqrt{n}) $. Furthermore, as each chain contributes at least one deletion under $ \A $, the overall increase in cost is again bounded by a factor of $ \Oh(\sqrt{n})$. 

We note that the threshold used to distinguish between large and small subtrees is crucial in determining the approximation guarantee. For instance, if we set a higher threshold, the number of nodes with large subtrees will decrease, but we will have to delete more nodes associated with small subtrees located close to either end of the chain, which will raise the overall cost. On the other hand, if we adopt a lower threshold, the number of nodes with large subtrees will increase, again contributing to a higher total cost. This analysis supports our decision to use a threshold of $\sqrt{n}$ and explains the rationale behind the approximation factor.

Finally, we analyze the cost of handling misaligned nodes that do not belong to a chain.
Note these are essentially the multi-branch misaligned nodes. We represent these nodes using a concept called an \emph{extended chain}. An extended chain is a sequence of nodes (possibly multi-branch) alternating between $\T$ and $\T'$, where, under $\A$, each node aligns only one of its parentheses with the previous node in the sequence and its other parenthesis with the next node.
We observe that an extended chain contains at most one multi-branch misaligned node and one partially deleted node. Since $\A$ incurs at least a cost of one for each partially deleted node, deleting the multi-branch misaligned nodes increases the overall cost by at most a constant factor. We remark that the formal analysis includes many additional details and must address more complex scenarios, such as cases where multiple chains or extended chains overlap. Further, instead of focusing on ordered trees, we consider a more general case where the input may consist of a pair of ordered forests, with each forest being an ordered collection of trees.

\paragraph{Reduction from Tree to String Edit Distance in the Dynamic Setting.} To implement our reduction dynamically, the primary challenge is updating the label of each node following an edit operation to keep the Euler tour string representation up-to-date. Since our label encodings rely on the heavy-light decomposition of the input trees, the first step is to maintain this decomposition dynamically. We achieve this using a strategy similar to that used for dynamic heavy-light decomposition in the Dyck edit distance algorithm (Theorem~\ref{thm:dyck-improved}), which enables a logarithmic update time.

Next, we observe that a dynamic edit (such as insertion, deletion, or relabeling of a node) may require adding, removing, or updating a pair of parentheses in the Euler tour for the affected node. The main challenge, however, is that any ancestor nodes containing the modified parentheses in their labels will also need corresponding label updates.
 In our string reduction, it is specifically the parents of light nodes who encode information about the subtrees rooted at these light nodes within their labels. Therefore, any update within these light subtrees must also be reflected in the label of their parent. Consequently, for any edit, the only affected nodes are those that serve as parents to light nodes along the path from the root to the edited node.
 Since any root-to-leaf path in the heavy-light decomposition contains at most $\Oh(\log n)$ light nodes, only $\Oh(\log n)$ nodes need their labels updated for each edit.
After identifying changes in the heavy-light decomposition, we apply split and concatenate operations from the dynamic strings data structure~\cite{DBLP:conf/soda/GawrychowskiKKL18} to update the labels of each affected node in $\Oh(\log n)$ time. Therefore, the total update time for maintaining the Euler tour representation is bounded by $\Oh(\text{polylog}(n))$.

In addition to the updates mentioned above, there are further updates required due to nodes outside the root-to-edited-node path that may switch from light to heavy or vice versa.
We note that only the nodes whose status has shifted between heavy and light, along with their parent nodes, need label modifications.
 Specifically, when a node changes from light to heavy, the label of its parent must substitute the child's subtree with a dummy character. Conversely, if a node transitions from heavy to light, we must replace the dummy character in the parent's label with the encoding of the light child's subtree.
In a conventional heavy-light decomposition, an arbitrary number of nodes can change their status between heavy and light with each update. However, as we discussed while analyzing the heavy-light decomposition for Dyck edit distance, our specialized decomposition has the property that only a logarithmic number of nodes can have their statuses switched.
Since each switch requires updating the labels of only two nodes, the total number of nodes needing label updates is bounded by $\Oh(\log n)$. Consequently, using the dynamic strings data structure allows us to perform all label updates, and thus update the Euler tour representation, in a total update time of $\Oh(\text{polylog}(n))$.

\vspace{2mm}
As in the final step of our approximate dynamic Dyck edit distance algorithm, we similarly pass the updated parenthesis representations of each tree to the dynamic string edit distance algorithm. This algorithm then manages dynamic insertions and deletions of characters within the parenthesis string, producing an updated alignment in sub-polynomial update time.

\subsection{Improved Static and Dynamic Tree Edit Distance}

We end with a technical overview for our $\Oh(k^2 \log n)$-approximate static and dynamic tree edit distance algorithm where $k$ is an upper bound on the distance and $n$ is the input size. Our algorithm actually uses a decomposition of the parenthesis string representation of the input forests to distinguish between the two cases of whether the distance of the input instance is at most $k$ or at least $\tOh(k^2)$. Our algorithm works straightforwardly in linear time up to logarithmic factors in the static setting. We show that using a clever combination of the dynamic string data structure of \cite{DBLP:conf/soda/GawrychowskiKKL18} and the dynamic tree query data structure of \cite{Navarro2014}, we can implement our novel decomposition algorithm in the dynamic setting with amortized $\tOh(1)$-update time.

\paragraph{Tree Decomposition in the Static Setting.}

We begin with a brief description of our novel tree edit distance approximation algorithm for distance bound $k$. Given forests $\F$ and $\G$, the algorithm maintains a decomposition $S$ of $\str{F}$, where recall $\str{\F}$ is the parenthesis representation of $\F$. Initially, $S$ is just the entire string $\str{\F}$. 

We use the following routine for a limited number of iterations to find our approximation. First, we remove the first part $\sigma$ in $S$. If $\sigma$ has size at most $\Oh(k)$, we do nothing and move on to the next part of $S$.  If instead the size of $\sigma$ is larger than our threshold, we check if there is a part in $\G$ that exactly matches $\sigma$ and has the same starting index up to a left or right shift of length $k$. If such a match is found, we remove $\sigma$ from $S$.  If a match is not found, we partition $\sigma$ into a constant number of pieces with size a constant fraction of the previous size of $\sigma$ and add these pieces to $S$. We repeat these steps for $\Oh(k \log n)$ iterations, and if $S$ is empty after the last iteration, we output $\ted(\F, \G) \leq k$, otherwise $\ted(\F, \G) \geq \Oh(k^2 \log n)$.

To analyze the correctness of our algorithm, we observe that we can view the parts contained in $S$ at any timestep of the algorithm as a tree $\T$. Each node of the tree corresponds to a single part contained in $S$. The root of $\T$ is the node corresponding to the entire forest $\str{\F}$. If a part $\sigma$ is broken into a constant number of smaller parts $\sigma_1, \sigma_2, \ldots, \sigma_M$ in the last case of our algorithm, then the node $v_\sigma \in \T$ corresponding to $\sigma$ will be the parent of the nodes $v_{\sigma_1}, v_{\sigma_2}, \ldots, v_{\sigma_M} \in \T$ corresponding to $\sigma_1, \sigma_2, \ldots, \sigma_M$.  Each level of the tree contains each parenthesis of $\str{\F}$ at most once since at any timestep, $S$ cannot contain the same parenthesis in two different parts. Since $\ted(\F, \G) \leq k$, there are at most $k$ edits made to $\str{\F}$ and furthermore, at most $k$ nodes on each level of the tree whose corresponding part will not find matches in $\G$.  Therefore, the total number of nodes in $\T$ cannot ever exceed $\Oh(k \cdot h)$ where $h$ is the height of the tree.  Since we either match or partition each $\sigma \in S$ into pieces of a constant fraction of the previous size, $h$ must be $\Oh(\log n)$. Thus, if $\ted(\F, \G) \leq k$, we are guaranteed to find $S$ empty after $\Oh(k \cdot h) = \Oh(k \log n)$ iterations, which matches our algorithm. Moreover, since our algorithm matches all parts of $S$ exactly except those with size at most $\Oh(k)$, the set of matched elements actually defines an alignment between $\F$ and $\G$ with $\Oh(k^2 \log n)$ deletions needed for non-matched elements. 

\paragraph{Tree Decomposition in the Dynamic Setting.}

Implementing the above algorithm in the dynamic setting introduces two main challenges. First, how we partition each unmatched part of set $S$ into a constant number of pieces with reduced size is straightforward in the static setting, i.e., using a tree centroid decomposition to guide our partitioning. In the dynamic setting, however, we are not aware of any such quickly computable decomposition for trees. Second, we need a fast way to find matches for each part in $S$. The latter challenge we discuss in more detail after explaining how we perform our partitioning, as the partitioning method determines the types of parts of $\str{\F}$ that will be contained in $S$.

Our partitioning method will take as input a \emph{subforest} $\str{F}[i \dd j)$, a balanced parenthesis substring of $\str{\F}$, or a \emph{context} $(\str{\F}[i_1 \dd j_1), \str{\F}[i_2 \dd j_2))$. A context is a subforest of $\F$ excluding the subtree of some node in this subforest. More formally, both $\str{\F}[i_1 \dd j_2)$ and $\str{\F}[j_2 \dd i_2)$ must be a subforest of $\F$.  See Figure~\ref{fig:context} for an example.

To partition a subforest piece $\sigma$, we find a node $u$ in $\sigma$ such that its subtree has size at most $|\sigma|/2$ and its parent $u'$ (if it exists) has subtree size at least $|\sigma|/2$. Furthermore, we choose $u$ such that it is centered in $\sigma$, i.e., $o(u) \leq \lceil|\sigma|/2 \rceil$, $c(u) \geq \lceil |\sigma|/2\rceil$.  We then divide the subforest into six subforest and context parts around this node.  See Figure~\ref{fig:forest-part} for a picture of the six parts.  To see that each of the six parts has size at most $|\sigma|/2$, we first observe that we chose $u$ to be a central node in $\sigma$, so any of the parts fully contained to the left or right of $u$ must have size at most $|\sigma|/2$. For the remaining two parts, we observe that $u$ and its parent $u'$ are chosen such that the subtree of $u$ has size at most $|\sigma|/2$ and the set of remaining nodes has size at most $|\sigma| - |\sigma|/2 = |\sigma|/2$, respectively.

Partitioning a context piece $\sigma$ is a bit more difficult and results in two cases. In the first case, if there is a node $u$ in $\sigma$ satisfying the same constraints as in the case for subforests, we may perform the same decomposition as we did before. However, it is possible all valid nodes $u$ satisfying these constraints are exactly in the gap of the context.  Therefore, in the second case, no such $u$ exists. Let $v$ be the node in $\sigma$ whose subtree is excluded from $\sigma$ as per the definition of contexts. The subtree of $v$ in $\sigma$ must already be size $|\sigma|/2$ or else we would fall into the first case. We set the first part of our partition to be all nodes of $\sigma$ excluding $v$ and its subtree.  Then, we note that at most one of the left half of the children of $v$ or the right half of the children of $v$ must have a combined size of greater than $|\sigma|/2$. For the smaller half, we set it to be its own part. For the remaining half of the subtree of $v$, this half is itself a subforest, and so we partition the remaining subforest into 6 more pieces as we did for subforest pieces before.  This results in at most 8 total pieces of size at most $|\sigma|/2$.

For both subforest and context partitioning, in order to check for a node $u$ in $\sigma$ satisfying our size constraints, we perform a binary search on a root-to-leaf path in $\sigma$ aided by the dynamic tree query data structure of \cite{Navarro2014} (see Theorem~\ref{thm:navarro} for the specific set of queries).  The tree query data structure supports $\Oh(\log n)$-time dynamic insertions and deletions and answers queries in time $\Oh(\log n)$. Thus, our partitioning can be done in time $\Oh(\log^2 n)$ per part in $S$.

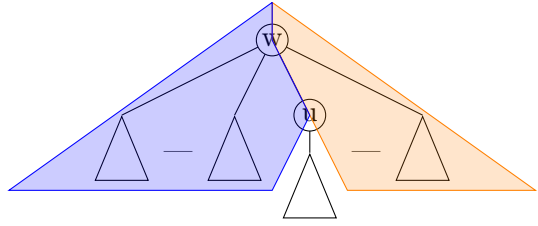
\begin{figure}
    \centering
    \begin{tikzpicture}
    \draw (0, 0) node[circle,draw,minimum size=12pt, inner sep=0pt, outer sep=0pt] (p1) {w};
    \draw (p1) -- (-2, -1)
    node [isosceles triangle,minimum size=20pt,draw, rotate=90, anchor=apex] {};
    \draw (p1) -- (-.5, -1)
    node [isosceles triangle,minimum size=20pt,draw, rotate=90, anchor=apex] {};
    \draw (p1) -- (2, -1)
    node [isosceles triangle,minimum size=20pt,draw, rotate=90, anchor=apex] {};
    \draw (.5, -1) node [circle,draw,minimum size=12pt, inner sep=0pt, outer sep=0pt] (p4) {u};
    \draw (p4) -- (.5, -1.5)
    node [isosceles triangle,minimum size=20pt,draw, rotate=90, anchor=apex] {};
    \draw (-1.25, -1.5) node[circle] (dash1) {\textemdash};
    \draw (1.25, -1.5) node[circle] (dash2) {\textemdash};
    \draw (p1) -- (p4);
    \draw (0, .5) -- (3.5, -2) -- (1, -2) -- (.5, -1) -- (0, 0) -- (0, .5)[color=orange,fill=orange,fill opacity=.2];
    \draw (0, .5) -- (-3.5, -2) -- (0, -2) -- (.5, -1) -- (0, 0) -- (0, .5) [color=blue,fill=blue,fill opacity=.2];
\end{tikzpicture}
    \caption{An example of a context $(\str{F}[o(w) \dd o(u)],\str{F}[c(u) \dd c(w)])$ for the subtree of node $w$ in a forest $\F$ formed by excluding the subtree of node $u$. The left part of the context, depicted in blue, is $\str{F}[o(w) \dd o(u)]$ and the right part of the context, depicted in orange, is $\str{F}[c(u) \dd c(w)]$.}
    \label{fig:context}
\end{figure}

    \begin{figure}
    \centering  \scalebox{.75}{\begin{tikzpicture}
    \draw (0, 0) node[circle,draw,minimum size=12pt, inner sep=0pt, outer sep=0pt] (p1) {$w$};
    \draw (p1) -- (-3, -1)
    node [isosceles triangle,minimum size=20pt,draw, rotate=90, anchor=apex] {};
    \draw (p1) -- (-2, -1)
    node [isosceles triangle,minimum size=20pt,draw, rotate=90, anchor=apex] {};
    \draw (p1) -- (3, -1)
    node [isosceles triangle,minimum size=20pt,draw, rotate=90, anchor=apex] {};
    \draw (.5, -1) node [circle,draw,minimum size=12pt, inner sep=0pt, outer sep=0pt] (p4) {$u'$};
    \draw (-.5, -2) node[circle,draw,minimum size=12pt, inner sep=0pt, outer sep=0pt] (p5) {$u$};
    \draw (p4) -- (-1.5, -2)
    node [isosceles triangle,minimum size=20pt,draw, rotate=90, anchor=apex] {};
    \draw (p4) -- (.5, -2)
    node [isosceles triangle,minimum size=20pt,draw, rotate=90, anchor=apex] {};
    \draw (p4) -- (1.5, -2)
    node [isosceles triangle,minimum size=20pt,draw, rotate=90, anchor=apex] {};
    \draw (p5) -- (-.5, -2.5)
    node [isosceles triangle,minimum size=20pt,draw, rotate=90, anchor=apex] {};
    \draw (-1, -2.5) node[circle] (dash1) {\textemdash};
    \draw (1, -2.5) node[circle] (dash2) {\textemdash};
    \draw (p4) -- (p5);
    \draw (-2.5, -1.5) node[circle] (dash1) {\textemdash};
    \draw (2.25, -1.5) node[circle] (dash2) {\textemdash};
    \draw (p1) -- (p4);
    \draw (-5, -2) node [isosceles triangle, minimum size=60pt,draw,rotate=90] {};
    \draw (5, -2) node [isosceles triangle, minimum size=60pt,draw,rotate=90] {};
    \draw (-6.25, 0.25) -- (-6.25, -3) -- (-3.75, -3) -- (-3.75, .25) -- (-6.25, .25) [color=red, fill=red, fill opacity=.2];
    \draw (3.75, .25) -- (3.75, -3) -- (6.25, -3) -- (6.25, .25) -- (3.75, .25) [color=green, fill=green, fill opacity=.2];
    \draw (-1, -1.75) -- (-1, -3.5) -- (0, -3.5) -- (0, -1.75) -- (-1, -1.75) [draw=blue, fill=blue, fill opacity=.2];
    \draw (-2, -2) -- (-2, -3) -- (-1, -3) -- (-1, -2) -- (-2, -2) [draw=orange, fill = orange, fill opacity=.2];
    \draw (0, -2) -- (0, -3) -- (2, -3) -- (2, -2) -- (0, -2) [draw=magenta, fill=magenta, fill opacity=.2];

    \draw (-6.25, .75) node {$\str{F}[i \dd j)$};
\end{tikzpicture}}
    \caption{An example of partitioning a forest piece $\str{F}[i \dd j)$.  The six parts as defined in the proof of Lemma~\ref{lem:partitionpiece} are colored 1 - blue, 2 - orange, 3 - magenta, 4 - uncolored, 5 - red, 6 - green.}
    \label{fig:forest-part}
    \end{figure}
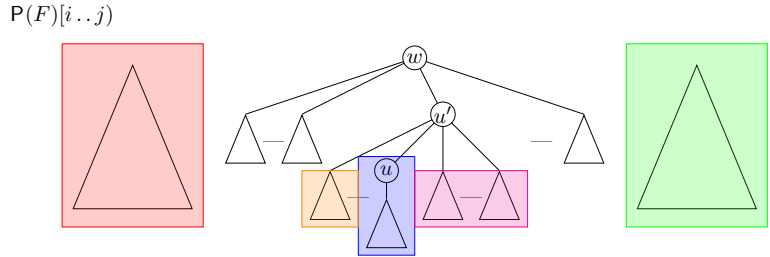

Lastly, we describe our efficient matching routine for subforests and contexts for the matching step of our algorithm. We utilize the dynamic strings data structure of \cite{DBLP:conf/soda/GawrychowskiKKL18} to answer $\ipm$, \emph{internal pattern matching}, queries in $\Oh(\log n)$-time each, formally given in \cite{KK22}. Given a pattern $P$ and text $T$, $\ipm(P, T)$ reports all occurrences of $P$ in $T$. For $|T| < 2|P|$, all occurrences are reported as a single arithmetic progression of starting positions of $P$.

By definition, subforests are subtrings of $\F$, so a single $\ipm$ query straightforwardly finds a match for a given subforest $\sigma$ of $\F$ if one exists. Contexts are defined as two substrings in $\F$ that correspond to a balanced set of parentheses when combined, and so, we can use two $\ipm$ queries to find the set of matches in $\G$ for each substring individually. The remaining difficulty is to find some pair of matches for each part that correspond to a single valid context in $\G$. Again using the data structure of \cite{Navarro2014}, we may rewrite the arithmetic progression of starting indices from each $\ipm$ query as instead an arithmetic progression of depths.  See Figure~\ref{fig:context-match} for an example.  The intersections of these arithmetic progressions corresponds to potential context matches. We prove that we can check precisely the smallest depth intersection of the two depth progressions, and if this intersection is a valid context, we output this match, otherwise there is no match. 

    \begin{figure*}
    \centering \begin{tikzpicture}
    \draw (-.5, .25) node {$\G$};
    \draw (0, 0) node[circle,draw,minimum size=12pt, inner sep=0pt, outer sep=0pt] (p1) {};
    \draw (p1) -- (.75, -.5) node [isosceles triangle,minimum size=10pt,draw,rotate=90,anchor=apex, fill=blue, fill opacity=.2] {};
    \draw(0, -1) node[circle,draw,minimum size=12pt, inner sep=0pt, outer sep=0pt, fill=green, fill opacity=.2] (p2) {};
    \draw (p2) -- (.75, -1.5) node [isosceles triangle,minimum size=10pt,draw,rotate=90,anchor=apex, fill=blue, fill opacity=.2] {};
    \draw (p2) -- (-.75, -1.5) node [isosceles triangle,minimum size=17pt,draw,rotate=90,anchor=apex, fill=orange, fill opacity=.2] {};
    \draw(0, -2) node[circle,draw,minimum size=12pt, inner sep=0pt, outer sep=0pt] (p3) {};
    \draw (p3) -- (.75, -2.5) node [isosceles triangle,minimum size=10pt,draw,rotate=90,anchor=apex, fill=blue, fill opacity=.2] {};
    \draw (p3) -- (-.75, -2.5) node [isosceles triangle,minimum size=10pt,draw,rotate=90,anchor=apex, fill=magenta, fill opacity=.2] {};
    \draw(0, -3) node[circle,draw,minimum size=12pt, inner sep=0pt, outer sep=0pt, fill=green, fill opacity=.2] (p4) {};
    \draw (p4) -- (.75, -3.5) node [isosceles triangle,minimum size=10pt,draw,rotate=90,anchor=apex, fill=blue, fill opacity=.2] {};
    \draw (p4) -- (-1.25, -3.5) node [isosceles triangle,minimum size=17pt,draw,rotate=90,anchor=apex, fill=orange, fill opacity=.2] {};
    \draw(0, -4) node[circle,draw,minimum size=12pt, inner sep=0pt, outer sep=0pt] (p5) {};
    \draw (p5) -- (1.5, -4.5) node [isosceles triangle,minimum size=10pt,draw,rotate=90,anchor=apex, fill=blue, fill opacity=.2] {};
    \draw (p5) -- (-1.5, -4.5) node [isosceles triangle,minimum size=10pt,draw,rotate=90,anchor=apex, fill=magenta, fill opacity=.2] {};
    \draw(-1, -5) node[circle,draw,minimum size=12pt, inner sep=0pt, outer sep=0pt] (p6) {};
    \draw (p6) -- (-1.75, -5.5) node [isosceles triangle,minimum size=17pt,draw,rotate=90,anchor=apex, fill=orange, fill opacity=.2] {};
    \draw(-1, -6) node[circle,draw,minimum size=12pt, inner sep=0pt, outer sep=0pt] (p7) {};
    \draw (p7) -- (-1.75, -6.5) node [isosceles triangle,minimum size=10pt,draw,rotate=90,anchor=apex, fill=magenta, fill opacity=.2] {};
    \draw (p1) -- (p2);
    \draw (p2) -- (p3);
    \draw (p3) -- (p4);
    \draw (p4) -- (p5);
    \draw (p5) -- (p6);
    \draw (p6) -- (p7);
    \draw (1, -5) node[circle,draw,minimum size=12pt, inner sep=0pt, outer sep=0pt] (p8) {};
    \draw (p8) -- (1.75, -5.5) node [isosceles triangle,minimum size=10pt,draw,rotate=90,anchor=apex, fill=blue, fill opacity=.2] {};
    \draw (1, -6) node[circle,draw,minimum size=12pt, inner sep=0pt, outer sep=0pt] (p9) {};
    \draw (p5) -- (p8);
    \draw (p8) -- (p9);

    \draw (-5.5, .25) node {$\piece$};
    \draw(-5, 0) node[circle,draw,minimum size=12pt, inner sep=0pt, outer sep=0pt] (f1) {};
    \draw (f1) -- (-4.25, -.5) node [isosceles triangle,minimum size=10pt,draw,rotate=90,anchor=apex, fill=blue, fill opacity=.2] {};
    \draw (f1) -- (-5.75, -.5) node [isosceles triangle,minimum size=17pt,draw,rotate=90,anchor=apex, fill=orange, fill opacity=.2] {};    
    \draw(-5, -1) node[circle,draw,minimum size=12pt, inner sep=0pt, outer sep=0pt] (f2) {};
    \draw (f2) -- (-4.25, -1.5) node [isosceles triangle,minimum size=10pt,draw,rotate=90,anchor=apex, fill=blue, fill opacity=.2] {};
    \draw (f2) -- (-5.75, -1.5) node [isosceles triangle,minimum size=10pt,draw,rotate=90,anchor=apex, fill=magenta, fill opacity=.2] {}; 
    \draw (f1) -- (f2);
    \draw (f2) -- (-5, -2) node [isosceles triangle,minimum size=20pt,draw,rotate=90,dotted,anchor=apex, fill=gray, fill opacity=.2] {};
\end{tikzpicture}
    \caption{ An example of a context $\piece$ and forest $\G$ for which  $\mathsf{HasMatch}(\piece) = \mathsf{true}$. The arithmetic progression of depths for the left and right part of $\sigma$ are $f_1(i) = 2i + 1$, $f_2(i) = i$, respectively.  Green colored nodes represent root nodes of valid matching contexts corresponding to intersections in $f_1$ and $f_2$.}
    \label{fig:context-match}
    \end{figure*}

Since both our partitioning and our matching routines can be done in poly-logarithmic time dynamically, our entire dynamic algorithm can be performed in $\tOh(k)$-time. We may amortize this time cost to obtain $\tOh(1)$-time  per update, see Section~\ref{sec:kTreeApprox} for details.

\noindent\textbf{Further Results.}
We provide further new results for Dynamic Dyck edit distance, specifically when we are allowed higher update times, and when we want to compute the Dyck edit distance exactly. 
Specifically,
our result is a dynamic algorithm with sub-linear update time that computes an $\Oh(\log n)$-approximation of Dyck edit distance. 
For a parenthesis string $X$, our algorithm uses two separate subroutines to efficiently handle the cases when the Dyck edit distance of $X$, $\ded(X)$, is large ($\Omega(\sqrt{n})$) and small ($\Oh(\sqrt n)$). For large distance, we may straightforwardly recompute the Dyck to string edit distance reduction of \cite{Saha2014, KSSODA2023} after a number of dynamic edits occur that is proportional to the distance. We can use a string edit distance algorithm on the resulting instances to obtain the following result.  See Subsection~\ref{subsec:largeed} for details.

\begin{restatable}{theorem}{largeDed}
\label{thm:largeded}
    There exists a randomized dynamic algorithm that, for a fixed parameter $\epsilon\in \mathbb{R}_+$, maintains an $\Oh(f(\epsilon)\log |X|)$-approximation of $\ded(X)$ (correctly with high probability against an adaptive adversary) for a string $X\in \Sigma^*$ undergoing edits. The algorithm has amortized update time $\tOh_\epsilon(|X|^{1+\epsilon}/\ded(X))$ per edit and pre-processing time $\tOh_\epsilon(|X|^{1+\epsilon})$. 
\end{restatable}

When $ \ded(X) $ is small, maintaining the same approximation guarantee requires more frequent recomputations due to amortization constraints. We show how these recomputations can be done efficiently using the dynamic strings data structure of \cite{DBLP:conf/soda/GawrychowskiKKL18} that we equip with a few specific additional operations useful for our algorithm. For this case we obtain the following. See Subsection~\ref{subsec:smalled} for details.

\begin{restatable}{theorem}{smallDed}
\label{thm:smallded}
    There exists a randomized dynamic algorithm that maintains an $\Oh(\log (\ded(X)))$-approximation of $\ded(X)$ (correctly with high probability against an adaptive adversary) for a string $X\in \Sigma^*$ undergoing edits. This algorithm has amortized update time $\tOh(\ded(X) + 1)$ per edit and pre-processing time of $\tOh(|X|)$.
\end{restatable}

We combine the two above results to obtain the following theorem.

\begin{restatable}{theorem}{combinedDed}
\label{thm:combinedded}
    There exists a randomized algorithm that, for a fixed parameter $\epsilon \in \mathbb{R}_+$, maintains an $\Oh(f(\epsilon) \log \ded(X))$-approximation of $\ded(X)$ (correctly with high probability against an adaptive adversary) for a string $X\in \Sigma^*$ undergoing edits. This algorithm has amortized update time $\tOh(|X|^{.5 + \epsilon})$ per edit and pre-processing time of $\tOh_\epsilon(|X|^{1+\epsilon})$.
\end{restatable}

In addition to approximation algorithms, we design an exact algorithm for the bounded-distance regime, where the Dyck edit distance is upper bounded by $k$. This algorithm achieves an update time of $ \tOh((1+\ded(X))^5) $ per edit. Consequently, when $ k $ is small, this approach beats the traditional cubic-time solution for Dyck edit distance in the dynamic setting. To obtain this result, we adapt the Dyck edit distance algorithm of \cite{DBLP:conf/soda/FriedGKKPS22}. The algorithm of \cite{DBLP:conf/soda/FriedGKKPS22} separates parentheses of the input string into trapezoids and clusters, and we show this decomposition can be maintained in the dynamic setting without much additional work. See Section~\ref{sec:exact} for details.

\begin{restatable}{theorem}{exactK}
\label{thm:exactk}
    There exists a dynamic algorithm that maintains $\ded(X)$ (correctly with high probability against an adaptive adversary) for a string $X \in \Sigma^*$ undergoing edits, with update time $\tOh((1+\ded(X))^5)$ per edit. The algorithm has a pre-processing time $\tOh(|X|+\ded(X)^5)$.
\end{restatable}

\section{Related Works}\label{sec:related}

\paragraph{String Edit Distance:}
One of the most fundamental problems in computer science is \emph{string edit distance}, with roots dating back to the 1960s~\cite{L65,NW70,WF74}. Given two strings, the objective of string edit distance is to find the minimum number of edit operations (insertions, deletions, and substitutions) required to transform one string into the other. An algorithm by Landau and Vishkin~\cite{LandauV11988} builds upon insights by Ukkonen~\cite{Ukk85} and Myers~\cite{Myers1986} to achieve this task in $\Oh(n+k^2)$ time by leveraging suffix trees and an elegant approach combining greedy and dynamic programming component. This algorithm has linear running time in the length of the strings $n$ as long as the edit distance $k$ is at most $\Oh(\sqrt{n})$.
For larger values of $k$, the study of approximation algorithms for edit distance has been extensive~\cite{LMS98,10.5555/874063.875596,10.1109/FOCS.2004.14,10.5555/1109557.1109644,10.1145/1536414.1536444,10.1109/FOCS.2010.43}, particularly in recent years~\cite{10.1145/3313276.3316371,10.1145/3422823,GRS:20,KS20,BR19,DBLP:journals/jacm/BoroujeniEGHS21,K19}. Andoni and Nosatzki~\cite{ANFOCS20} obtained the current best bound, presenting a constant-factor approximation algorithm running in $\Oh(n^{1+\epsilon})$-time for any constant $\epsilon > 0$.
Recently, Das, Gilbert, Hajiaghayi, Kociumaka, and Saha~\cite{DGHKS23} proposed an $\Oh(n + k^5)$-time algorithm for weighted string edit distance in which a weight function determines the cost of edits, further improved to $\tOh(n + \sqrt{nk^3})$ by Cassis, Kociumaka, and Wellnitz~\cite{CKW23}. This latter paper also shows that any polynomial-factor improvement of this runtime (for any $\sqrt{n}\le k \le n$) would violate the All-Pairs Shortest Paths Hypothesis (i.e., yield an All-Pairs Shortest Paths algorithm with truly sub-cubic runtime).

\paragraph{Dyck Edit Distance:}
In this paper, we delve deeply into the variant of edit distance known as \emph{Dyck edit distance}, a special case of general language edit distance~\cite{aho1972minimum, M95,saha2017fast,bringmann2019}, with various practical applications such as fixing hierarchical data files, notably XML and JSON files~\cite{h:78,k:12}. For a parenthesis string of length $n$, the Dyck edit distance is the minimum number of edits (character insertions, deletions, and substitutions) required to make the string well-balanced. Various algorithms have been developed for both exact~\cite{bringmann2019,CDX22,Durr23,DBLP:conf/soda/FriedGKKPS22} and approximate~\cite{Saha2014,DasKS2022,KSSODA2023} versions of this problem.
Finding the exact Dyck edit distance is at least as hard as Boolean matrix multiplication~\cite{AbboudBW2018}. Recent studies have also explored the bounded-distance Dyck edit problem in which an upper bound $k$ is known for the Dyck edit distance of a string. Backurs and Onak~\cite{DBLP:conf/pods/BackursO16} initially presented an algorithm with a runtime of $\Oh(n+k^{16})$, later improved to $\Oh(n+k^5)$ \cite{DBLP:conf/soda/FriedGKKPS22}, and further improved to $\Oh(n+k^{4.5442})$ through the use of fast matrix multiplication \cite{DBLP:conf/soda/FriedGKKPS22,Durr23}. However, except for the $\Oh(n^3)$-time exact algorithm for language edit distance~\cite{M95}, these results do not extend to the weighted setting. 
When the weight function is anti-symmetric and satisfies the triangle inequality, an $\Oh(n + k^{12})$-time algorithm is known~\cite{DGHKS23} for the weighted setting.

\paragraph{Tree Edit Distance:}
 The \emph{tree edit distance} problem, initially introduced by Selkow~\cite{SELKOW1977184}, is another generalization of edit distance aiming to compute a measure of dissimilarity between two rooted ordered trees with node labels. Given two trees, the tree edit distance is the minimum number of insertions, deletions, and relabelings needed to transform one tree into the other. This problem finds applications in diverse fields such as compiler optimization~\cite{10.1145/1644015.1644017}, structured data analysis~\cite{DBLP:conf/vldb/Chawathe99,10.5555/1315451.1315465,10.1145/1613676.1613680}, image analysis~\cite{10.1016/S0167-8655(97)00179-7}, and computational biology~\cite{10.1137/0213024,DBLP:journals/bioinformatics/ShapiroZ90,10.5555/262228,LMS98,10.1016/j.tcs.2004.12.030}.
Currently the best algorithm for finding exact tree edit distance has runtime $\tOh(n^{(3+\omega)/2})$ due to Nogler et al.~\cite{NPSVXY25}, following a series of improvements from $\Oh(n^6)$ \cite{10.1145/322139.322143}, $\Oh(n^4)$ \cite{zhang1989simple}, $\Oh(n^3\log n)$ \cite{Klein98}, $\Oh(n^3)$ \cite{10.1145/1644015.1644017}, $\Oh(n^{2.9546})$ \cite{Xiao21}, and $\Oh(n^{2.9149})$~\cite{Durr23}. Furthermore, a $(1+\epsilon)$-approximation algorithm for tree edit distance with a running time of $\tOh(n^2)$ was presented by Boroujeni, Ghodsi, Hajiaghayi, and Seddighin \cite{DBLP:conf/stoc/BoroujeniGHS19}.
Recently, Seddighin and Seddighin \cite{DBLP:conf/innovations/SeddighinS22} introduced an $\Oh(n^{1.99})$-time $(3+\epsilon)$-approximation algorithm for tree edit distance. For the case when the tree edit distance doe not exceed a threshold $k$, Akmal and Jin~\cite{DBLP:conf/icalp/AkmalJ21} presented an $\tOh(nk^2)$-time algorithm (improving upon an $\Oh(nk^3)$-time solution~\cite{10.1007/11496656_29}), whereas Das, Gilbert, Hajiaghayi, Kociumaka, and Saha~\cite{DGHKS23} achieved a runtime of $\tOh(n+k^{7})$ algorithm (improving upon an $\Oh(n+k^{15})$-time solution~\cite{GHKSS22}). For weighted tree edit distance, the fastest algorithm, by Demaine, Mozes, Rossman, and Weimann~\cite{10.1145/1644015.1644017}, operates in $\Oh(n^3)$ time, matching the conditional lower bound of Bringmann, Gawrychowski, Mozes, and Weimann~\cite{10.1145/3381878}. However, when the weighted tree edit distance $k$ is small and the weight function satisfies the triangle inequality, the $\Oh(n + k^{15})$-time algorithm of \cite{DGHKS23} may outperform the $\Oh(n^3)$ upper bound.

\paragraph{Dynamic Edit Distance:}
Dynamic algorithms address real-world scenarios where data evolves rapidly, necessitating efficient updates to maintain solutions. Over recent decades, substantial progress has been made in understanding fundamental problems in graphs and sequences within the dynamic setting. Notable examples include maximum matchings in graphs \cite{Sankowski07, GP13, BKS23}, connectivity \cite{NS17,ChuzhoyGLNPS20}, maximum flow \cite{ChenGHPS20, BLS23}, minimum spanning trees \cite{NSW17,ChuzhoyGLNPS20}, clustering \cite{CharikarCFM04, HK20, BateniEFHJMW23}, diameter estimation \cite{BN19}, independent set \cite{AOSS18}, pattern matching \cite{ABR00, DBLP:conf/soda/GawrychowskiKKL18,DBLP:conf/focs/Charalampopoulos20, CGKMU22}, lossless compression \cite{NIIBT20}, string similarity \cite{ABR00, CGP20}, longest increasing subsequence \cite{MS20, KS21, GJ21}, suffix arrays \cite{KK22}, and others.

Various dynamic models with differing complexities have been explored, including incremental, decremental, and fully-dynamic models, supporting insertion, deletion, or both \cite{BhattacharyaK20, GutenbergWW20, KMS22, BLS23, Bernstein16, HKN16, BernsteinGW20, BernsteinGS20,DI04, GP13, RZ16, HHS22}. The evolution of dynamic edit distance algorithms initially tackled simpler variants of the problem, with early focus on updates limited to string endpoints \cite{LMS98, KP04, IIST05,Tis08}. Notably, Tiskin's linear-time algorithm \cite{Tis08} accommodates edit distance maintenance with character insertions or deletions at either endpoint of the strings. Recently, Charalampopoulos, Kociumaka, and Mozes \cite{CKM20} presented an $\tilde{O}(n)$-update-time algorithm for the general dynamic edit distance problem, allowing updates anywhere within the strings, with the caveat that SETH prohibits sublinear update time even in the most restrictive scenarios.
A very recent preprint of Gorbachev and Kociumaka~\cite{GK24} circumvents this lower bound and achieves $\tOh(k)$ time per update, where $k$ is the edit distance.
Kociumaka, Mukherjee, and Saha presented the first approximate dynamic edit distance algorithm that approximates edit distance in $n^{o(1)}$ time within a $n^{o(1)}$ approximation factor \cite{KMS2023}. Improving the approximation to constant in sub-polynomial update time remains a fascinating open question.

In addition to dynamic edit distance algorithms, dynamic data structures for string problems have been extensively studied. These structures maintain dynamic strings while accommodating queries such as equality testing, longest common prefix, pattern matching, and more \cite{m97,ABR00,DBLP:conf/soda/GawrychowskiKKL18,DBLP:conf/focs/Charalampopoulos20, CGKMU22,DBLP:conf/soda/GawrychowskiKKL18,NIIBT20,KK22,CGP20,chen2013dynamic,MS20,KS21,GJ21,ABCK19}. Some works also consider split and concatenate operations, enabling functionalities beyond character edits, such as cut-paste and copy-paste operations \cite{m97,ABR00,DBLP:conf/soda/GawrychowskiKKL18}. Meanwhile, significant strides have been made in fully dynamic graph algorithms, addressing central problems like spanning forests, transitive closure, shortest paths, maximum matching, maximal independent set, maximal matching, set cover, vertex cover and others \cite{NS17,ChuzhoyGLNPS20, NSW17,Sankowski04, BNS19,KMS22,DI04, BN19,Sankowski07, GP13,AOSS18,BaswanaGS15, Solomon16,BHN19, AAGPS19,BK19} (see also survey~\cite{HHS22} for recent advances in fully dynamic graph algorithms).

\section{Approximate Dynamic Dyck Edit Distance}\label{sec:fastapprox}
In this section, we prove our main result for dynamic Dyck edit distance, a sub-polynomial approximation with sub-polynomial update time for dynamic edits. Specifically, for parameters $2\le b \le n$ and a dynamic parenthesis string $X$ of length at most $n$, our algorithm achieves an $\Oh(b \log_b n \log n)$-approximation and supports updates in time $\Oh(b^2 (\log n)^{\Oh(\log_b n)})$.  We show that using a tree representation of parenthesis string $X$ determined by the heights of the parentheses in $X$, the heavy-light decomposition of the tree provides a novel natural reduction from Dyck to string edit distance. Each update to $X$ will cause at most $\Oh(\log n)$ updates to the heavy paths and the corresponding string edit distance instances.  We then use a dynamic approximation of~\cite{KMS2023} to maintain an approximate string edit distance for each heavy path.

\begin{definition}[Reduction tree]
\label{def:reduc}
    Given a parenthesis string $X$, the \emph{height tree} $\T_X$ is defined via pairs of twins in $X$.

    Each node in $\T_X$ corresponds to an opening parenthesis $X[i]$ and the corresponding closing parenthesis $X[\tw(i)]$. For any closing parentheses $X[i]$ with $X[\tw(i)]$ undefined (there is no earlier opening parenthesis twin in $X$), we add an opening parenthesis of an unused type to the start of $X$, which we refer to as a \emph{dummy} opening parenthesis. The parent of the node corresponding to $X[i]$ is the node corresponding to $X[j]$ where $j =\max\{j' \in [0\dd i): h(j') = h(i) - 1\}$ if  such a $j$ exists. We add an additional node to be the root of the tree and set the root to be the parent of all nodes with no parents according to our above construction.

    Given a node $u$, we let $i_u$ represent the index of the corresponding opening parenthesis of $u$. We say $u$ \emph{contains} $X[i_u]$. Note these indices are adjusted after inserting the dummy opening parentheses.
\end{definition}

We next state a useful fact regarding the height of a tree.  
Since the parent $u$ of a given node $v$ must have a smaller index and smaller height, it follows that any node $w$ is a descendant of $u$ if and only if the indices $i_w$ and $\tw(i_w)$ lie within the range $[i_u \dd \tw(i_u)]$.  
Using this fact, we will efficiently determine whether a node is heavy or light in $\Oh(\log n)$ time by maintaining a self-balancing binary tree over $X$ that returns the indices corresponding to any given node of $\T_X$.

\begin{fact}
\label{fct:descendants}
For every two nodes $u, v$ in the height tree $\T_X$, the node $v$ is a (proper) descendant of $u$ if and only if $i_u < i_v < \tw(i_v) < \tw(i_u)$.
\end{fact}

\begin{proof}
    First, let $v$ be a proper descendant of $u$ and assume for contradiction we have $i_u < \tw(i_u) < i_v < \tw(i_v)$ or $i_v < \tw(i_v) < i_u < \tw(i_u)$ (other orderings contradict the non-crossing properties of twins).  We begin with the case when $i_v < \tw(i_v) < i_u < \tw(i_u)$. Consider the path from $u$ to $v$.  Observe that by Definition~\ref{def:reduc}, each parent-child pair in this path must have the opening parenthesis corresponding to the parent to the left of the opening parenthesis corresponding to the child.  Therefore, $i_u < i_v$, which contradicts our assumption.
    
    Now, assume $v$ is a proper descendant of $u$ with $i_u < \tw(i_u) < i_v < \tw(i_v)$. Consider the path $w_0, w_1, w_2, \ldots, w_\ell$ in $\T_X$ from $u$ to $v$ with $w_0 = u$ and $w_\ell = v$. By  Definition~\ref{def:reduc}, $h(i_{w_i}) < h(i_{w_{i+1}})$ for all $0 \leq i < \ell$ and so, $h(i_u) < h(i_v)$. Furthermore by our definition of parents, $u$ must be the closest opening parenthesis to $i_v$ with height $h(i_u)$ satisfying $i_u < i_v$. By the definition of twins, we know that $h(\tw(i_u) + 1) = h(i_u)$. Since height only increments by 1 per opening parenthesis and decrements by 1 per closing parenthesis in the parenthesis string and $h(i_v) > h(i_u)$, there must be some opening parenthesis at index $x > \tw(i_u)$ (which may be $\tw(i_u) + 1$) such that $h(x) = h(i_u)$.  This contradicts that $i_u$ is the closest opening parenthesis to $i_v$ satisfying $i_u < i_v$ with height $h(i_u)$.  Thus, if $v$ is a proper descendant of $u$, we must have that $i_u < i_v < \tw(i_v) < \tw(i_u)$.

    Now, assume we have $i_u < i_v < \tw(i_v) < \tw(i_u)$. Let $u'$ be the ancestor of $v$ with height $h(i_{u'}) = h(i_u)$; by Definition~\ref{def:reduc}, $u'$ is the node corresponding to the closest opening parenthesis to $i_v$ with height $h(i_u)$ and satisfying $i_{u'} < i_v$. If $u' = u$, we are done. Otherwise, by our proof above and the non-crossing property of twins, we know that $i_u < i_{u'} < i_v < \tw(i_v) < \tw(i_{u'}) < \tw(i_u)$.  However, this contradicts the definition of twins for $i_u$ and $\tw(i_u)$, since $\tw(i_{u'})$ satisfies $h(\tw(i_{u'}) + 1) = h(i_u)$ and is closer to $i_u$ than $\tw(i_u)$. Thus, $u$ must be an ancestor of $v$.
\end{proof}

We decompose $\T_X$ into a heavy-light decomposition. Recall that in our heavy-light decomposition, for a given node $u \in \T_X$, a child node $v$ of $u$ is \emph{heavy} if $\lfloor \lg|\T_X(v)| \rfloor = \lfloor \lg|\T_X(u)| \rfloor$, and \emph{light} otherwise. Note that $u$ has at most one heavy child by Observation~\ref{obs:heavy-child}.

%\begin{observation}
%\label{obs:light}
%    Any path from the root of a tree to a leaf contains at most $\lg n + 1$ light nodes.
%\end{observation}

A \emph{heavy} path of $\T_X$ is a path beginning with a light node $u$ and the longest path of descendants of $u$ which are all heavy nodes (some heavy paths may consist of a single light node). Note in Section~\ref{sec:prelim}, we defined heavy paths to exclude light nodes, but in the remainder of this section we assume that light nodes are included as defined above. Given a heavy path $P= (u_1, u_2, u_3, \ldots, u_\ell)$, the \emph{heavy string} corresponding to $P$ is the string $X[i_{u_1}] \cdot X[i_{u_2}] \cdot \ldots \cdot X[i_{u_\ell}] \cdot X[\tw(i_{u_\ell})] \cdot \ldots \cdot X[\tw(i_{u_1})]$. We denote the set of heavy paths of tree $\T_X$ as $\hvy_{\T_X}$, and we note that $\hvy_{\T_X}$ partitions $\T_X$. We denote $\hvys_X$ as the set of heavy strings of $\X$, which partitions $X$. We assume each heavy path in $\hvy_{\T_X}$ stores a pointer to its head node, tail node, and the corresponding heavy string in $\hvys_{X}$.  We also assume each heavy path knows its length (which may be calculated during pre-processing and maintained in the head nodes of each path for constant-time access).  We will maintain these dynamically, and we do not explicitly store $T_\X$.

We give an $\Oh(\log^4 n)$-time algorithm for maintaining $\hvys_X$ when $X$ undergoes dynamic edits. Furthermore, we show that there are at most $\Oh(\log n)$ changes to the set of heavy path strings, and therefore only $\Oh(\log n)$ changes to the string edit distance instances we obtain from our set of heavy paths. To show this, we partition the set of nodes in $\T_X$ into three types based on their location relative to the update: Type I nodes whose corresponding opening and closing parenthesis come before the edit, Type II nodes whose corresponding opening parenthesis comes before the edit and whose corresponding closing parenthesis comes after the edit, and Type III nodes whose corresponding parentheses come after the edit. It is easy to see that the subtrees of Type I and Type III nodes remain completely the same since the heights of all parentheses in these subtrees are shifted by the same amount by the dynamic edit. As for Type II nodes, it turns out that these are exactly the set of nodes on the path from the root of $\T_X$ to the node containing the edited parenthesis.  As such, we know there are at most $\Oh(\log n)$ light nodes on this path by Observation~\ref{obs:heavy-light}.

Through a careful case analysis of the interactions between Type II nodes and their parents and children, we show that any heavy node that becomes light must have a neighboring light Type II node that may be charged for updating the heavy path.  Since there are $\Oh(\log n)$ light Type II nodes and only one child of a Type II node may be heavy, there are at most $\Oh(\log n)$ edits to $\hvys_{X}$ per edit to $X$.  To identify these changes, we look at the light Type II nodes of $\T_X$ and check all the nearby nodes to see if any have changed from light to heavy or heavy to light.  When checking to see if a node is light or heavy, we use Fact~\ref{fct:descendants} as well as a segment tree to answer range queries of the form $(i, h)$ where we want to know the index of the next parenthesis after index $i$ whose height is $h$.  To maintain this segment tree efficiently subject to updates, we use a common technique called lazy propagation. With these range queries, we can find the size of a node's subtree quickly and check whether they are heavy or light. While the lemma is significant, we leave the proof to the end of this section, in Subsection~\ref{subsec:dynheavylight} since it is a long case-by-case analysis and introduces a few additional algorithmic tools, e.g. the range queries discussed above.

\begin{lemma}
\label{lem:heavylight}
    For a string $X$ with $n = |X|$, $\hvy_{\T_X}$ and $\hvys_{X}$ undergoes $\Oh(\log n)$ edits, which can be identified and maintained in $\Oh(\log^5 n)$ time, when $X$ undergoes a dynamic insertion or deletion.
\end{lemma}

\begin{proof}
    See Subsection~\ref{subsec:dynheavylight} for details.
\end{proof}

We now prove that the heavy string set provides a good approximation of $\ded(X)$ for a parenthesis string $X$.  Recall that $\dedd(X)$ denotes the deletion-only Dyck edit distance where substitutions and insertions are not allowed, and that $\ded(X) \leq \dedd(X) \leq 2 \cdot \ded(X)$.

\begin{lemma}
\label{lem:dycktostrnew}
    Given a parenthesis string $X$ and its heavy string set $\hvys_X$, 
    \[\dedd(X) \leq \sum_{s \in \hvys_X} \dedd(s) \leq 2(\lg(|X|) + 1) \dedd(X).\]
\end{lemma}
\begin{proof}
    Given a string $s \in \hvys_X$, let $P_s \in \hvy_{\T_X}$ be the corresponding heavy path in $\T_X$ of $s$ and let $r_s$ represent the head node of $P_s$ (which is a light node). We define the \emph{depth} of $s$ to be $\lfloor \lg|\T_X(r_s)| \rfloor$. We partition the strings of $\hvys_X$ according to their depths. We let $M_i = \{M^i_1, M^i_2, \ldots, M^i_{\ell_i}\}$ represent the set of heavy path strings whose depth is $i$ for $0 \leq i < h$, where $h \le \lfloor\lg |X|\rfloor+1$.
    For example, $M_0$ is the set of single-node heavy paths at the leaves of $T$. For $0\le i \le h$, we define $X_i$ to be the string of $X$ with $\bigcup_{i' =0}^{i-1} \bigcup_{j=0}^{\ell_{i'}} M^{i'}_j$ removed.
    We will prove that the following holds for every $0 \le i < h$:
    \[\dedd(X_{i+1}) \le \dedd(X_{i})\le \dedd(X_{i+1}) + \sum_{j=1}^{\ell_i} \dedd(M^i_j) \qquad\text{and}\qquad  \sum_{j=1}^{\ell_i} \dedd(M^i_j) \le 2\dedd(X_i).\]

    Let us first consider an optimal deletion-only edit sequence $\textsf{OPT}_i$ for $X_i$. For each $1\le j \le \ell_i$, let $D_j^i$ be the number of deletions in $M^i_j$ made by $\textsf{OPT}_i$ and $N_j^i$ be the number of parentheses in $M^i_j$ that are matched to parentheses outside of $M^i_j$. We note that, since $M^i_j$ is a balanced LR-segment, whenever a parenthesis of $M^i_j$ is matched to a parenthesis outside of this segment, we must delete a parenthesis in the other half of $M^i_j$ since a matching must be monotone. Therefore, \begin{align} N_j^i \leq D_j^i\label{ineq:NvsD}.\end{align}  Now, we define $\hat{M}^i_0, \hat{M}^i_1, \ldots, \hat{M}^i_{\ell_i}$ as the remaining substrings of $X_i$ between each $M^i_j$, so that $X_i = \hat{M}^i_0 \cdot M^i_1 \cdot \hat{M}^i_1 \cdot M^i_2 \cdot \ldots \cdot M^i_{\ell_i}\cdot \hat{M}^i_{\ell_i}$.  We define $\hat{D}^i_j$ as the number of deletions in $\hat{M}^i_j$ made by $\textsf{OPT}_i$ and $\hat{N}^i_j$ as the number of parentheses in $\hat{M}^i_j$ that are matched to a parenthesis in some $M^i_{j'}$.  We note that by definition, $\sum_{j = 0}^{\ell_i} \hat{N}^i_j = \sum_{j = 1}^{\ell_i} N^i_j$, and furthermore, \[\dedd(X_i) = \sum_{j = 0}^{\ell_i} \hat{D}^i_j + \sum_{j = 1}^{\ell_i} D^i_j.\]
    Now, we consider $\dedd(X_{i+1})$, which does not need to delete any parentheses from $M_i$ since all of these segments are removed but may need to delete parentheses of $X_{i+1}$ which were matched to parentheses in $M_i$. Therefore,
    \begin{align*}
        \dedd(X_{i+1}) &\le \dedd(X_i) + \sum_{j=0}^{\ell_i} \hat{N}^i_j - \sum_{j = 1}^{\ell_i} D^i_j \\
            &= \dedd(X_i) + \sum_{j = 1}^{\ell_i} \left(N^i_j - D^i_j\right) \\
            &\leq \dedd(X_i) + \sum_{j=1}^{\ell_i} \left(D^i_j - D^i_j\right) = \dedd(X_i)
    \end{align*}
    where the second inequality follows from \eqref{ineq:NvsD}.
    Furthermore, \eqref{ineq:NvsD} yields \begin{align*}\sum_{j = 1}^{\ell_i}\dedd(M^i_j) &\le \sum_{j=1}^{\ell_i} \left(D_j^i + N_j^i\right) \leq \sum_{j=1}^{\ell_i} 2D_j^i 
    \leq 2\dedd(X_{i}). \end{align*}
    It remains to prove that $\dedd(X_{i})\le \dedd(X_{i+1}) + \sum_{j=1}^{\ell_i} \dedd(M^i_j)$.
    For this, note that each string $M^i_j$ is a substring of $X_i$ and $X_{i}$ is obtained from $X_{i+1}$ by removing all these substrings $M^i_j$. 
    If $Y,Z\in \Dyck$, then also $Y[0\dd y)\cdot Z\cdot Y[y\dd |Y|)\in \Dyck$ for every $y\in [0\dd |Y|]$; hence, in order to place $X_i$ in $\Dyck$, it suffices to independently perform the deletions that place $X_{i+1}$ and all strings $M^i_j$ in $\Dyck$.

    Finally, observe that 
    \begin{align*}
        \sum_{s \in \hvys_X} \dedd(s) = \sum_{i = 0}^{h-1} \sum_{j = 1}^{\ell_i} \dedd(M_j^i) &\leq \sum_{i = 0}^{h-1} 2\dedd(X_{i}) \\ &\leq \sum_{i =0}^{h-1} 2\dedd(X_0) \\ &= 2h \cdot \dedd(X) \\&\leq 2 (\lg(|X|) + 1) \dedd(X).
    \end{align*}
    and 
    \begin{align*}
        \sum_{s \in \hvys_X} \dedd(s) = \sum_{i = 0}^{h-1} \sum_{j = 1}^{\ell_i} \dedd(M_j^i) &\ge \sum_{i = 0}^{h-1} \left(\dedd(X_{i})-\dedd(X_{i+1})\right) \\ &= \dedd(X_0)-\dedd(X_h)= \dedd(X)-\dedd(\varepsilon)= \dedd(X).
    \end{align*}
    This completes the proof.    
\end{proof}

Combining Lemma~\ref{lem:dycktostrnew}, Observation~\ref{obs:LRed}, and the facts that $\ded(X)\le \dedd(X)\le2\ded(X)$ and $\ed(Y,Z)\le \edd(Y,Z)\le 2\ed(Y,Z)$, we immediately obtain the following result.
\begin{corollary}\label{cor:dycktostrnew}
        Given a parenthesis string $X$ and its heavy string set $\hvys_X$, 
        \[\ded(X) \leq 2\sum_{s \in \hvys_X} \ed(\Lp(s),T(\Rp(s))) \leq 8(\lg(|X|) + 1) \ded(X).\]
\end{corollary}

We now state the dynamic string edit distance result of \cite{KMS2023}, which we utilize to quickly solve each string edit distance instance we obtain from our reduction.

\begin{theorem}[\cite{KMS2023} Theorem 7.10]
\label{thm:dyned}
    Given integers $2 \leq b \leq n$, there exists a randomized dynamic algorithm that maintains an $\Oh(b \log_b n)$-approximation of $\ed(X, Y)$ (correctly with high probability against an oblivious adversary) for (initially empty) strings $X, Y$ of length at most $n$  undergoing edits. The expected update time of the algorithm is $\Oh(b^2 (\log n)^{\Oh(\log_b n)})$.
\end{theorem}

Our dynamic approximation with a trade-off between approximation factor and update time follows from Lemma~\ref{lem:heavylight}, Corollary~\ref{cor:dycktostrnew}, and Theorems~\ref{thm:dyned}.

\begin{restatable}{theorem}{Dyckimproved}
\label{thm:dyck-improved}
   Given integers $2\le b \le n$, there exists a randomized dynamic algorithm, that maintains an $\Oh(b\log_b n \log n)$-approximation of $\ded(X)$ (correctly with high probability against an oblivious adversary) for an (initially empty) string $X$ of length at most $n$ undergoing edits. The expected update time of the algorithm is $\Oh(b^2 (\log n)^{\Oh(\log_b n)})$ per edit.
\end{restatable}

\begin{proof}
    Note, by Lemma~\ref{lem:heavylight}, $\hvys(X)$ requires $\Oh(\log |X|)$ edits per dynamic edit to $X$ and can be maintained in $\Oh(\log^4 n)$-time. By Theorem~\ref{thm:dyned}, these edits may be handled in time $\Oh(b^2 (\log n)^{\Oh(\log_b n)})$ each, and so in total take expected time $\Oh(\log^4 n + \log n \cdot b^2 (\log n)^{\Oh(\log_b n)}) = \Oh(b^2 (\log n)^{\Oh(\log_b n)})$ to maintain.  We use the sum of the approximations for each heavy path string to obtain our approximation of $\ded(X)$, which may be updated by subtracting the $\Oh(\log n)$ affected paths' old approximations and adding their new approximations after recomputing them. By Theorem~\ref{thm:dyned} and Corollary~\ref{cor:dycktostrnew}, the approximation factor will be $\Oh(b \log_b n) \cdot 8 (\lg(n) + 1) = \Oh(b \log_b n \log n)$.
\end{proof}

Observe that when $b$ is set to $2^{\Theta(\sqrt{\log n \log \log n})}$, the approximation ratio becomes  $\Oh(\log n \cdot 2^{\Oh(\sqrt{\log n \log \log n})})$ and the update time becomes $2^{\Oh(\sqrt{\log n \log \log n})}$. Setting $b=\log^{O(\frac{1}{\epsilon})}{n}$, we get
$(\log{n})^{O(\frac{1}{\epsilon})}$-factor approximation with $O(n^{\epsilon})$ update time for any constant $\epsilon >0$.

Although the dependence on parameter $b$ in the approximation factor and update time may be somewhat high in the above theorem, both quantities match the state-of-the-art for dynamic string edit distance (up to an extra $\Oh(\log n)$ factor present in the approximation only). The Dyck edit distance is at least as hard to approximate as the string edit distance, so this relationship is expected: any significant improvements to the approximation or update time for dynamic Dyck edit distance would therefore automatically improve the state-of-the-art for the dynamic string edit distance, and vice versa.

\subsection{Dynamic Heavy-Light Decomposition}
\label{subsec:dynheavylight}

In the remainder of this section, we describe how the set of heavy paths of a height tree for a parenthesis string $X$ can be maintained in $\Oh(\log^5 n)$ time. Before we show the proof of how the heavy paths of $\T_X$ are maintained, we develop a few useful tools.

First, we discuss \emph{minimum height range queries} $\rg(i, h)$ that return the maximum range starting at $i$ in $X$ such that all parentheses in this range have heights greater than $h$. To answer them in $\Oh(\log^2 n)$ time, we use a well-known technique called lazy propagation to maintain a balanced binary tree with height information on top of $X$ that we may use to quickly determine the answer to such queries.

\begin{lemma}
\label{lem:segtree2}
    A parenthesis string $X$ can be dynamically maintained in $\Oh(\log |X|)$ time per update so that, given $i, h \in [0 \dd |X|)$, the range $\rg(i, h)$ can be computed in $\Oh(\log^2 |X|)$ time.
\end{lemma}

\begin{proof}
    See Appendix~\ref{sec:lazy} for details.
\end{proof}

Additionally, we also wish to know, for a given index $i_v$, its twin index $\tw(i_v)$.  We use a self-balancing binary tree structure to do so, similar to the construction of $\hat{X}$ in the proof of Lemma~\ref{clm:hatX}. 

\begin{lemma}
\label{lem:parenIdx}
    A parenthesis string $X$ can be dynamically maintained in $\Oh(\log |X|)$ time per update so that,
    given any parenthesis $X[i_v]$ corresponding to node $v \in \T_X$, the indices $i_v$ and $\tw(i_v)$ can be computed in $\Oh(\log |X|)$ and $\Oh(\log^2 |X|)$ time, respectively.     Furthermore, $|\T_X(v)|$ may be computed in $\Oh(\log^2 |X|)$ time. %The data structure can be initialized in time $\Oh(|X|\log |X|)$.
\end{lemma}

\begin{proof}
        To answer index queries of a node $v \in \T_X$, we utilize a binary tree $B$ built on top of $X$. Each node $x \in B$ corresponds to a substring of $X$, denoted $S_x$.  Each leaf node of $B$ corresponds to a single character of $X$ and leaves are assigned characters in sorted order according to $X$ such that the leftmost leaf corresponds to $X[0]$.  Each internal node $x$ of $B$ with children $x_L, x_R$ corresponds to the concatenation $S_{x_L} \cdot S_{x_R}$.  Therefore, each level of $B$ partitions $X$. In addition to this partition, each node $x \in B$ stores the size of its substring $|S_{x}|$ (note that $x$ does not store the substring $S_x$ itself, as this would require too much time to maintain). The tree $B$ can be constructed straightforwardly in time $\Oh(|X|)$. When $X$ undergoes a dynamic edit, we only need to update the nodes whose corresponding substrings are affected.  When the height of $B$ is $\Oh(\log |X|)$, there are only $\Oh(\log |X|)$ such nodes. Leaf nodes do not require any work to maintain, and internal nodes may be reconstructed in constant time from the information in their children (the length of the substring $S_x$ at node $x$ is the sum of the lengths of the substrings $S_{x_L}$ and $S_{x_R}$ at the two children of $x$). We may use a self-balancing tree such as an AVL tree to maintain the height of $\Oh(\log |X|)$ while still only affecting $\Oh(\log |X|)$ nodes per update.

    Now, we show how to use $B$ to answer index queries. We assume that each node $v\in \T_X$ stores a pointer to the leaf of $B$ corresponding to the character $X[i_v]$. To find the index $i_v$, we begin at the leaf corresponding to $X[i_v]$ and we traverse $B$ upwards while maintaining a count. Whenever the node we are at is the left child of its parent we do not increase the count.  Whenever the node we are at is the right child of its parent, this means that the left child corresponds to a substring of $X$ lying to the left of $X[i_v]$ and so we add the size of this substring, stored in that left child, to the count.  We do this until we reach the root node so that the count becomes the number of characters to the left of $X[i_v]$, which is $i_v$.  Since the height of $B$ is $\Oh(\log |X|)$ and each traversal to a parent node takes constant time, these queries only take $\Oh(\log |X|)$ time.

    To get the index of $\tw(i_v)$, we perform a minimum height range query $\rg(i_v +1, h + 1)$ which takes $\Oh(\log^2 |X|)$ time by Lemma~\ref{lem:segtree2}. Then, we may compute $|\T_X(v)| =  (1+\tw(i_v)-i_v)/2$ 
\end{proof}

With the two tools above, we can efficiently determine if a given node $v \in \T_X$ is heavy or not.  The subtree of the heavy child of $u$ (if any) is at least half the size of the subtree of $u$.  Therefore, we may make a binary search using minimum range height queries to identify the heavy child.

\begin{lemma}
\label{lem:heavychild}
    Given a parenthesis string $X$ and the data structure of Lemma~\ref{lem:segtree2}, for a parenthesis $X[i_v]$ corresponding to node $v \in \T_X$, we may find $i_u, \tw(i_u)$ where $u$ is the heavy child of $v \in \T_X$, if it exists, in $\Oh(\log^3 n)$ time.
\end{lemma}

\begin{proof}
    We first obtain $i_v$ and $\tw(i_v)$ using Lemma~\ref{lem:parenIdx}, and let $d = \tw(i_v) - i_v$. Additionally, we find $h(i_v)$ by computing the total length of all heavy paths from $v$ to the root of $\X_T$ (recall each heavy path stores its length).  We observe that any heavy child $u$ of $v$ must have $|\T_X(u)| \geq |\T_X(v)| /2$. Therefore, we perform $\rg(i_v + d/4, h(i_v))$ since $h(i_u)$ must be $h(i_v) + 1$.  If the range query returns an index after $i_v + d/2$, we may determine that $i_u \in (i_v\dd i_v + d/4)$, otherwise $i_u \in [i_v + d/4 \dd i_v + d/2]$ if it exists. We set $d' = d /2$ and search the respective half by setting our starting point to either $i_v + d'/4$ again or $i_v + d/4 + d'/4$. We then continue until we find the exact starting index of $i_u$ if it exists, or find that no such index exists. 

    Once we determine $i_u$, we may use Lemma~\ref{lem:parenIdx} to determine if $u$ is indeed a heavy child of $v$.
\end{proof}

Now we prove that the heavy-light decomposition of $\T_X$ can be maintained in poly-logarithmic update time when $X$ undergoes dynamic edits.

\begin{proof}[Proof of Lemma~\ref{lem:heavylight}] 
     We discuss deletions and insertions separately.  
    First, assume the dynamic update is a deletion of $X[j]$, resulting in the string $Y = X[0 \dd j) \cdot X(j \dd |X|)$.  
    Let $\T_X$ and $\T_Y$ be the respective height trees of $X$ and $Y$.  
    Define a function $f : [0 \dd |X|] \to [0 \dd |Y|]$ mapping positions in $X$ to positions in $Y$ as
    \[f(i) = 
    \begin{cases}
    i & \text{if } i \in [0 \dd j),\\
    i-1 & \text{if } i \in [j \dd |X|].
    \end{cases}\]
    Observe that $Y[f(i)] = X[i]$ for every $i \in [0 \dd |X|) \setminus \{j\}$.  
    We also lift $f$ to a mapping between the nodes of $\T_X$ and $\T_Y$: if a node $u \in \T_X$ contains an opening parenthesis $X[i]$ with $i \neq j$, its image in $\T_Y$ is the node $f(u)$ containing the corresponding opening parenthesis $Y[f(i)]$.  
    If $X[j]$ is an opening parenthesis, then $f$ is undefined (denoted $\bot$) for the node containing $X[j]$.

    Furthermore, we classify nodes $u \in \T_X$ according to their relationship with the index $j$.  
    A node $u \in \T_X$ is of \emph{Type I} if $i_u < \tw_X(i_u) < j$, of \emph{Type II} if $i_u < j \le \tw_X(i_u)$, and of \emph{Type III} if $j < i_u < \tw_X(i_u)$.  
    Note that if $X[j]$ is an opening parenthesis, the node $u$ with $i_u = j$ is not assigned any type, since this node will be deleted and is not useful to group with other nodes.  
    We next discuss some relevant properties of each type of node.

    \begin{claim}
    \label{clm:TypeIandIII}
        For all Type I and Type III nodes $u \in \T_X$, we have $\T_X(u) = \T_Y(f(u))$.
    \end{claim}

    \begin{proof}
        First, assume $u$ is a Type I node. By definition, $i_u < \tw_X(i_u) < j$. Then by Fact~\ref{fct:descendants}, every node $v\in\T_X(u)$ satisfies $i_v < \tw_X(i_v) < \tw_X(i_u) < j$, and thus it is also a Type I node.
        Furthermore, the heights within $[0\dd j)$ do not change, i.e., $h_Y(f(i))=h_X(i)$ holds for all $i\in [0\dd j)$.
        Therefore, for any two type I nodes $v,w \in \T_X$, we have that $w$ is the parent of $v$ if and only if $f(w)$ is the parent of $f(v)$. Thus, $\T_X(u) = \T_X(f(u))$.

        The proof of when $u$ is a Type III node follows similarly.  By definition, $j < i_u < \tw_X(i_u)$, and every node $v \in \T_X(u)$ satisfies $j < i_u < i_v < \tw_X(i_u)$, and thus it is also a Type III node.
        Furthermore, the heights within $[j\dd |X|]$ are all incremented, that is, $h_Y(f(i))=h_X(i)+1$ for all $i\in [j\dd |X|]$ if $X[j]$ is closing,
        or all decremented, that is, $h_Y(f(i))=h_X(i)-1$ for all $i\in [j\dd |X|]$ if $X[j]$ is opening.
        Therefore, for any two type III nodes $v,w\in \T_X$, we have that $w$ is the parent of $v$ if and only if $f(w)$ is the parent of $f(v)$. Thus, $\T_X(u) = \T_X(f(u))$.
    \end{proof}

     By Claim~\ref{clm:TypeIandIII}, the structure of Type I nodes and Type III is essentially intact by the mapping~$f$.
     We therefore turn our attention to Type II nodes.

    \begin{claim}
    \label{clm:TypeIIchild}
    Every Type II node in $\T_X$ has exactly one Type II child, except for the following node with no Type II children:
    \begin{itemize}
        \item if $X[j]$ is closing, the node $u\in \T_X$ with $\tw_X(i_u)=j$;
        \item if $X[j]$ is opening, the parent of the node $u$ with $i_u=j$.
    \end{itemize}
    Furthermore, for any two Type II nodes $v,w\in \T_X$, we have that $w$ is the parent of $v$ if and only if $f(w)$ is the parent of $f(v)$.
    \end{claim}

    \begin{proof}
        First, consider a node $u\in \T_X$ and its two children $w,w'$ with $w$ located to the left of $w'$.
        Then, $i_w < \tw_X(i_{w}) < i_{w'} < \tw_X(i_{w'})$. If both $w,w'$ were Type II nodes, we would have $j \le \tw_X(i_{w}) < i_{w'} < j$, a contradiction. Thus, every node $u\in \T_X$ has at most one Type II child. 
        
        Now, we consider the case when $X[j]$ is closing.  First, if $\tw_X(i_u) = j$, all descendant nodes $v$ of $u$ satisfy $i_v < \tw_X(i_v) < \tw_X(i_u) = j$ and are Type I nodes.  So, in this case, $u$ has no Type II children as claimed. Now when $i_u < j < \tw_X(i_u)$, let us assume $u$ has no Type II children.  Therefore, all children $v$ satisfy either $\tw_X(i_v) < j$ or $j < i_v$. Since, twin pairs are non-crossing by Observation~\ref{obs:noncrossingtwins}, the Type II node $w$ with $i_w < \tw_X(i_w) = j$ must be a descendant of $v$ since $i_u < j < \tw_X(i_u)$.  However, $w$ cannot be a child of any Type I and Type III node, again due to twins' non-crossing property. Therefore, $u$ must have at least one Type II child. 

        Now, let us assume that $u \in \T_X$ is a Type II node and $X[j]$ is opening. If $u$ is the parent of node $v$ with $i_v = j$, note that any other Type II node $w$ must satisfy $i_w < j < \tw_X(i_w)$ and cannot be a sibling of $v$. Thus, $u$ has no Type II children in this case.  If $u$ is not the parent of $v$, it must still be an ancestor of $v$ and, by the same argument as above, $u$ must have at least one Type II child.

        To prove the final property of the statement, take Type II nodes $v,w$. By the definition of the mapping function, $f(v)\in \T_Y$ is the node containing $Y[f(i_v)]$ and $f(w)\in \T_Y$ is the node containing $Y[f(i_w)]$. Since $i_v, i_w < j$ and the heights within $[0\dd j)$ do not change, i.e., $h_Y(f(i))=h_X(i)$ holds for all $i\in [0\dd j)$, we conclude that $w$ is the parent of $v$ if and only if $f(w)$ is the parent of $f(v)$.
    \end{proof}

    \begin{claim}
    \label{clm:update}
    Consider a Type II node $u \in \T_X$ and its image $f(u) \in \T_Y$.  
    If $X[j]$ is a closing parenthesis, then $\tw_Y(f(i_u)) = f(\tw_X(i_w))$, where $w$ is the parent of $u$, or undefined if $u$ is the root of $\T_X$.  
    If $X[j]$ is an opening parenthesis, then $\tw_Y(f(i_u)) = f(\tw_X(i_v))$, where $v$ is the Type II child of $u$, or $\tw_Y(f(i_u)) = f(\tw_X(j))$ if no such child exists.
    \end{claim}

    \begin{proof}
        First, suppose that $X[j]$ is closing and $w$ is the parent of $u$ where $w$ is not the root of $\T_X$.
        Note that $\tw_X(i_w)=\min\{r\ge i_w : h_X(r+1)=h_X(i_w)\}$ and $\tw_X(i_u)=\min\{r\ge i_u : h_X(r+1)=h_X(i_u)\}$. Moreover, $h_X(i_u)=h_X(i_w)+1$.
        Consequently, $h_X(r)>h_X(i_u)$ for $r\in (i_u\dd \tw_X(i_u)]$, $h_X(r)\ge h_X(i_u)$ for $r\in (i_w\dd \tw_X(i_w)]$, and $h_X(\tw_X(i_w)+1)=h_X(i_u)-1$.  Informally, we know that all parentheses between $i_u$ and $\tw_\X(i_u)$ have height greater than $h(i_u)$, and moreover, the height of $i_w$ and its twin $\tw(i_w) + 1$ is $h(i_u) - 1$.  After deleting $X[j]$, the height of all parentheses after index $j$ are increased by 1, and so the height of $\tw_X(i_w)$ matches $h(i_u)$ in $\T_Y$ and is the nearest closing parenthesis with matching height to the right of $i_u$, i.e., it is the twin of $i_u$.  Formally put, since $i_w \le i_u < j \le \tw_X(i_u) < \tw_X(i_w)$, as well as $h_Y(f(r))=h_X(r)$ for $r\in [0\dd j]$ and $h_Y(f(r))=h_X(r)+1$ for $r\in (j\dd |X|]$,
        we have $h_Y(q) > h_X(i_u)=h_Y(f(i_u))$ for $q\in (f(i_u)\dd f(\tw_X(i_w))]$ and $h_Y(f(\tw_X(i_w))+1)=h_Y(f(\tw_X(i_w)+1))=h_X(\tw_X(i_w))+1 = h_X(i_u)=h_Y(f(i_u))$.
        Therefore, $\tw_Y(f(i_u)) = f(\tw_X(i_w))$.

        If the parent of $u$ is the virtual root added to $\T_X$ that does not correspond to any parenthesis of $X$, then note that $h_X(i_u) > h_X(r)$ for all $r \in (i_u \dd |X|)$ and so, $h_Y(f(i_u)) > h_Y(q)$ for all $q \in (f(i_u) \dd |X|)$.  Thus, $\tw_Y(f(i_u))$ is undefined in this case.

        Now, suppose that $X[j]$ is opening and $w$ is either a Type II child of $u$ or the deleted node. 
        Note that $\tw_X(i_w)=\min\{r\ge i_w : h_X(r+1)=h_X(i_w)\}$ and $\tw_X(i_u)=\min\{r\ge i_u : h_X(r+1)=h_X(i_u)\}$. Moreover, $h_X(i_u)=h_X(i_w)-1$.
        Consequently, $h_X(r)>h_X(i_u)$ for $r\in (i_u\dd \tw_X(i_u)]$, $h_X(r)> h_X(i_u)+1$ for $r\in (i_w\dd \tw_X(i_w)]$, and $h_X(\tw_X(i_w)+1)=h_X(i_u)+1$.
        Since $i_u \le i_w < j \le \tw_X(i_w) < \tw_X(i_u)$, as well as $h_Y(f(r))=h_X(r)$ for $r\in [0\dd j]$ and $h_Y(f(r))=h_X(r)+1$ for $r\in (j\dd |X|]$,
        we have $h_Y(q) > h_X(i_u)=h_Y(f(i_u))$ for $q\in (f(i_u)\dd f(\tw_X(i_w))]$ and $h_Y(f(\tw_X(i_w))+1)=h_Y(f(\tw_X(i_w)+1))=h_X(\tw_X(i_w))-1 = h_X(i_u)=h_Y(f(i_u))$.
        Therefore, $\tw_Y(f(i_u)) = f(\tw_X(i_w))$.
    \end{proof}

    The following two claims specify how Type I and Type III nodes with Type II parents may move from $\T_X$ to $\T_Y$. 
    We also consider the children of the deleted node (if $X[j]$ is opening); all these children are of Type III.
    \begin{claim}
    \label{clm:TypeIchild}
        If $u \in \T_X$ is a Type I node with a Type II parent $w$, then $f(w)$ is the parent of $f(u)$.
    \end{claim}

    \begin{proof}
        Note that since the heights of opening parentheses corresponding to Type I and Type II nodes do not change, the parent-child relationships between these nodes also do not change.
    \end{proof}

    \begin{claim}
    \label{clm:TypeIIIchild}
        Consider a Type III node $u \in \T_X$ with a Type II or deleted parent $w$.
        If $X[j]$ is closing, then the parent of $f(u)$ is $f(w')$, where $w'$ is the Type II child of $w$.\footnote{If $w$ does not have a Type II child, then $\tw_X(i_w)=j$ and all children of $w$ are of Type I.}
        If $X[j]$ is opening, then the parent of $f(u)$ is $f(w')$, where $w'$ is the parent of $w$.
    \end{claim}

    \begin{proof}
        Now, suppose that $X[j]$ is closing.
        Observe that the heights of opening parentheses corresponding to Type III nodes increase by one whereas the heights of opening parentheses corresponding to Type II nodes stay the same. If a Type III node $u$ has a Type II parent $w$ and $w'$ is the type II child of $w$, 
        then $h_X(i_u) = h_X(i_w) + 1 = h_X(i_{w'})$ and $h_Y(f(i_u))=h_X(i_u)+1=h_X(i_{w'}) + 1=h_Y(f(i_{w'})) + 1$. 
        Furthermore, $\tw_X(i_{w'})=\min\{r\ge i_{w'} : h_X(r+1)=h_X(i_{w'})\}$ and $\tw_X(i_w)=\min\{r\ge i_w : h_X(r+1)=h_X(i_w)\}$. For all $q \in (i_{w'}, i_u)$, we have that $h_Y(q) = h_X(q) + 1 > h_X(i_w) = h_X(i_u) - 2 = h_Y(f(i_u)) - 1.$ Thus, $i_{w'}$ is the maximum index of a parenthesis with $i_{w'} < i_j$ and $h_Y(f(i_{w'})) = h_Y(f(i_u)) - 1$. 
        
        The analysis is analogous if $X[j]$ is opening: then, the heights of opening parentheses corresponding to Type III nodes decrease instead. If $w$ is a child of the root of $\T_X$ that does not correspond to any parenthesis in $X$, then there would have no opening parenthesis that satisfies the parent condition for $f(u)$ and so, in this case, $f(u)$ becomes a child of the root.
    \end{proof}
    
    Next, we discuss how an algorithm can identify the changes to the heavy paths in $\T_X$. For the following analysis, define $k_x = \heavy{\T_X(x)}$ and $k_x' = \heavy{\T_Y(f(x))}$ for a node $x \in \T_X$.
    The following claim implies that a node $x\in \T_X$ needs to change its heavy path only if $k'_x \ne k_x$.
    \begin{claim}
    Let $u,v\in \T_X$ be nodes such that $k_u = k'_u = k'_v = k_v$. Then, $u$ is an ancestor of $v$ if and only if $f(u)$ is an ancestor of $f(v)$.
    \end{claim}
    \begin{proof}
    The characterization of the parent-child relation specified within the claims above leaves very few cases where the ancestor-descendant relation does not need to be preserved by $f$.
    If $X[j]$ is closing, this is only possible when $u$ is a Type II node and $v$ is a descendant of its Type III sibling, in which case $f(v)$ is a descendant of $f(u)$.
    In that case, however, $\T_Y(f(u))$ contains the image under $f$ of all the nodes in $\T_X(u)$ and $\T_X(v)$,
    so $|\T_Y(f(u))|\ge |\T_X(u)|+|\T_X(v)|\ge 2^{k_u}+2^{k_v} \ge 2^{k_u+1}$, so $k'_u \ge k_u+1$ contradicts our assumption.
    If $X[j]$ is closing, this is only possible when $u$ is a Type II node and $v$ is a descendant of its Type III child, in which case $f(v)$ is a descendant of a sibling of $f(u)$.
    In that case, however, $\T_X(u)$ contains the pre-image under $f$ of all the nodes in $\T_Y(f(u))$ and $\T_Y(f(v))$,
    so $|\T_X(u)|\ge |\T_Y(f(u))|+|\T_Y(f(v))|\ge 2^{k'_u}+2^{k'_v} \ge 2^{k'_u+1}$, so $k_u > k'_u+1$ contradicts our assumption.
    \end{proof}
    Thus, it suffices to iterate over all nodes $x\in \T_X$ with $k'_x \ne k_x$ and move them from one heavy path to another.
    By \cref{clm:TypeIandIII}, we only have to concern ourselves with  Type II nodes (and the deleted node).
    We use the following claim to characterize them.

    \begin{claim}\label{clm:kkp}
    Consider a Type II node $v$ with parent $u$.
    If $X[j]$ is closing, then $k_v \le k'_v \le k_u \le k'_u$.
    If $X[j]$ is opening, then $k'_v \le k_v \le k'_u \le k_u$.
    \end{claim}
    \begin{proof}
    Suppose that $X[j]$ is opening.
    By Claims~\ref{clm:TypeIchild} and~\ref{clm:TypeIIIchild}, we have $\T_Y(f(v))=f(\T_X(v)\cup\{w\in \T_X : w\text{ is a descendant of a Type III sibling of }v\})$. Thus, $|\T_X(v)|\le |\T_Y(f(v))|$. The subtree $\T_X(u)$ additionally contains $u$ and the descendants of Type III siblings of $v$, so $|\T_Y(f(v))|< |\T_X(u)|$. The inequality $|\T_X(u)|\le |\T_Y(f(u))|$ is analogous to its counterpart for $v$.

    Suppose that $X[j]$ is closing.
    By Claims~\ref{clm:TypeIchild} and~\ref{clm:TypeIIIchild}, we have $\T_Y(f(u))=f(\T_X(u)\setminus \{w\in \T_X : w\text{ is a descendant of a Type III sibling of }v\text{ or represents } X[j]\})$. Thus, $|\T_Y(f(u))| < |\T_X(u)|$.
    The subtree $\T_X(v)$ additionally contains the node $w$ representing $X[j]$ but misses $u$ and the descendants  of Type III siblings of $v$,
    so $|\T_X(v)|\le |\T_Y(f(u))|$. The inequality $|\T_Y(f(v))|<|\T_X(v)|$ is analogous to its counterpart for $u$.
    \end{proof}

    We are ready to characterize the affected nodes.
    First, suppose that $X[j]$ is closing. By \cref{clm:kkp}, $k_v \le k'_v \le k_u \le k'_u$ holds for every Type II node $v$ with parent $u$.
    If $k_v \ne k'_v$, then $k_v < k'_v \le k_u$, so $v$ is a light node.
    Consequently, our algorithm inspects the light Type II nodes $v$ (in the bottom-up order) and filters out those for which $k_v = k'_v$.
    We remove every remaining node $v$ from its original heavy path and insert its image $f(v)$ into the updated heavy path.
    In order to identify the latter path, we use \cref{lem:heavychild} to find the heavy child of $f(v)$.
    If there is none but the parent $u$ of $v$ satisfies $k'_v=k_u=k'_u$, then we insert $f(v)$ to the heavy path of $u$.
    For all other nodes~$v$, we have $k'_v < k'_u$, so $f(v)$ is light and belongs to a singleton heavy path.

    Next, suppose that $X[j]$ is opening. In this case, the deleted node containing $X[j]$ needs to be removed from its heavy path.
    Let us focus on the other nodes.  
    By \cref{clm:kkp}, we have that $k'_v \leq k_v \leq k'_u \le k_u$ for every Type II node $v$ with parent $u$.
    If $k_u \ne k'_u$, then $k'_v \le k_v \le k'_u < k_u$, so $v$ is a light node.
    Consequently, our algorithm inspects the parents $u$ of the light Type II nodes (as well as the parent of the deleted node, all in the bottom-up order) and filters out nodes $u$ with $k_u = k'_u$.
    We remove every remaining node $u$ from its original heavy path and insert its image $f(u)$ into the updated heavy path.
    In order to identify the latter path, we use \cref{lem:heavychild} to find the heavy child of $f(u)$.
    If there is none, then the parent $w$ of $u$ satisfies $k'_u < k_u \le k'_w$, so $f(u)$ is light and belongs to a singleton heavy path.

    In all cases, by \cref{fct:descendants}, there are $\Oh(\log n)$ Type II nodes to inspect, and they can be found by iterating over the heavy paths containing the node corresponding to $X[j]$ and its ancestors in $\T_X$.
    All the changes described above can be implemented in $\Oh(\log^4 n)$ time since we utilize Lemma~\ref{lem:heavychild} $\Oh(\log n)$ times.
    Thus, in all cases, we may perform the changes to the heavy path set in $\Oh(\log^4 n)$ time and there are at most $\Oh(\log n)$ such changes. This gives us our total time of $\Oh(\log^5 n)$.
    
    Note that when deleting an opening parenthesis, we have to delete the node corresponding to that parenthesis.  We may straightforwardly remove this node's parenthesis from its heavy path in $\T_X$ to fix the heavy path set. Additionally, a closing parenthesis previously matched at the root node may now no longer have a match. To fix this, we add another dummy opening parenthesis and a new node to the root of $\T_Y$.

    We note that a dynamic insertion of an opening parenthesis is symmetric in the case of deleting a closing parenthesis with one difference. The algorithm must create a new node for the new opening parenthesis, which we may add to the set of heavy paths using Lemma~\ref{lem:heavychild} in $\Oh(\log^3 n)$ time. We also may perform an additional clean-up step for both a dynamic insertion of an opening parenthesis and dynamic deletion of a closing parenthesis, since it is possible that a dummy opening parenthesis no longer has a matching closing parenthesis.  This must be the root node of $\T_Y$, and so we may directly check if this is the case in the root node and delete it if so.
    Similarly, a dynamic insertion of a closing parenthesis is symmetric to the deletion of an opening parenthesis except we do not have to add any node for the inserted parenthesis. As before, we may need to add a dummy node to the root of the tree which can be done in $\Oh(\log^3 n)$ time using Lemma~\ref{lem:heavychild}.
\end{proof}

\section{Static and Dynamic Tree to String Edit Distance}
\label{sec:treetostring}
In this section, we discuss the related problem of \emph{tree edit distance} and provide the first approximation algorithm for tree edit distance in the dynamic setting. Given two labeled trees as input, the goal of tree edit distance is to find the minimum number of edits needed to transform one tree into the other. To do so, we utilize a novel reduction technique, which converts an instance of tree edit distance to an instance of string edit distance.  Via a Euler tour of the trees and a heavy-light decomposition, the algorithm builds a string with special labeling where subtrees of a node are embedded into its label. In total, we are able to give an $\tOh(n^{\frac{1}{2}})$-approximation algorithm in the dynamic setting for tree edit distance with sub-polynomial update time. We begin by defining the reduction and discussing structural properties that arise in the reduction, and then in Subsection~\ref{subsec:constr}, we show how to find a tree alignment using a string alignment obtained via our reduction. We analyze the obtained tree alignment in Subsections~\ref{subsec:treeconstraints} and \ref{subsec:costanalysis}, and finally, in Subsection~\ref{subsec:dynTreeSec}, we show how to maintain our reduction in the dynamic setting.

\subsection{$\tOh(\sqrt{n})$-approximate Reduction to String Edit Distance}
\label{subsec:tree-reduction}

To approximate the tree edit distance between two forests using string edit distance, we employ their parenthesis representations.  
Before doing so, we refine the node labels to encode additional structural information.

\begin{definition}[Modified labeling and parenthesis representation]
\label{def:parenF}
Let $\F$ be a forest equipped with a heavy--light decomposition as defined in \cref{def:heavy-light}.  
If a node $v \in \F$ has a heavy child $w$, we define its \emph{light subtree} $\F'(v)$ as the subtree $\F(v)$ rooted at $v$, except that the subtree $\F(w)$ is replaced by a single node labeled with a special symbol~$\#$.  
If $v$ has no heavy child, we set $\F'(v) \coloneqq \F(v)$.  

We then define the \emph{modified node labeling} $\lambda_{\hld}$ so that each label $\lambda_{\hld}(v) \coloneqq (\F'(v),\, \heavy{\F(v)})$ encodes both the light subtree and the heavy depth of $v$.  
The forest $\F$ equipped with this labeling is denoted by $\F_{\hld}$, and its \emph{modified parenthesis representation} is given by $\paren(\F) \coloneqq \str{\F_{\hld}}$.
\end{definition}

By standard properties of heavy--light decompositions, every node has at most $\Oh(\log n)$ light ancestors in a forest~$\F$.  
Consequently, any node $u$ appears in the light subtrees of $\Oh(\log n)$ of its ancestors.  
In particular, a single node relabeling in $\F$ corresponds to $\Oh(\log n)$ relabelings in $\F_{\hld}$.  
A similar argument applies to insertions and deletions: editing a node $u$ in $\F$ may change the heavy depths of up to $\Oh(\log n)$ ancestors, and hence a single node edit in $\F$ can be modeled as $\Oh(\log n)$ node edits in $\F_{\hld}$.  
We formalize this below and conclude that
\[
\ed(\paren(\F), \paren(\G)) \le 2\,\ted(\F_{\hld}, \G_{\hld}) = \Oh\left(\ted(\F, \G) \log \max(|\F|, |\G|)\right).
\]

\begin{observation}\label{obs:node_relabel}
Let $\F$ be a forest of size $n$, and let $\G$ be obtained from $\F$ by relabeling a single node $u \in \F$.  
Then $\ted(\F_{\hld}, \G_{\hld}) \le \Oh(\log n) \cdot \ted(\F, \G)$.
\end{observation}
\begin{proof}
To transform $\F_{\hld}$ into $\G_{\hld}$, it suffices to relabel all nodes whose modified labels change.  
Since heavy depths are unaffected, the modified label of a node $v$ changes only if $u$ belongs to its light subtree $\F'(v)$.  
This occurs precisely when $v = u$ or when $v$ has a light child $w$ that is an ancestor of $u$; equivalently, when $v = u$ or $v$ is the parent of a light ancestor of $u$.  
By \cref{obs:heavy-light}, any node has at most $\Oh(\log n)$ light ancestors, so at most $\Oh(\log n)$ nodes of $\F_{\hld}$ need to be relabeled.  
\end{proof}

\begin{observation}\label{obs:node_delete}
Let $\F$ be a forest of size $n$, and let $\G$ be obtained from $\F$ by deleting a single node $u \in \F$.  
Then $\ted(\F_{\hld}, \G_{\hld}) \le \Oh(\log n) \cdot \ted(\F, \G)$.
\end{observation}
\begin{proof}
To transform $\F_{\hld}$ into $\G_{\hld}$, it suffices to delete $u$ and relabel all nodes whose modified labels change.  
We claim that the modified label of a node $v$ changes only if $v$ is the parent of $u$, the parent of a light ancestor of $u$, or the grandparent of a light ancestor of $u$.

To prove the claim, fix an arbitrary node $v \ne u$.  
If $v$ is not an ancestor of $u$, then $\G(v) = \F(v)$.  
In particular, the heavy depths of $v$ and its descendants remain unchanged, implying $\heavy{\G(v)} = \heavy{\F(v)}$ and $\G'(v) = \F'(v)$.  
Thus, the modified label of $v$ is unchanged.

Now suppose that $v$ is a proper ancestor of $u$.  
Then $\G(v)$ is obtained from $\F(v)$ by deleting $u$, so $|\G(v)| = |\F(v)| - 1$.  
Consequently,
\[
\heavy{\G(v)} =
\begin{cases}
\heavy{\F(v)} - 1 & \text{if } \log |\F(v)| \in \mathbb{Z},\\[4pt]
\heavy{\F(v)} & \text{otherwise.}
\end{cases}
\]
If $\log |\F(v)| \in \mathbb{Z}$, then every child $w$ of $v$ is light, since $|\F(w)| < |\F(v)|$ implies $\heavy{\F(w)} < \heavy{\F(v)}$.  
In particular, this includes the child of $v$ lying on the path to $u$, so $v$ is the parent of a light ancestor of $u$, which is one of the cases we accounted for.

Otherwise, $\heavy{\G(v)} = \heavy{\F(v)}$, and we only need to analyze how the light subtree $\G'(v)$ differs from $\F'(v)$.  
Such a difference arises only in one of the following cases:
\begin{enumerate*}[label=(\roman*)]
\item $u \in \F'(v)$ (and, as in the proof of \cref{obs:node_delete}, $v$ is then the parent of a light ancestor of $u$),
\item $v$ is the parent of $u$, or
\item a child of $v$ changes its heavy/light status under the decomposition.
\end{enumerate*}
The first two cases have already been covered.  
For case~(iii), recall that the status of a child $w$ of $v$ depends on whether $\heavy{\F(w)} < \heavy{\F(v)}$ and $\heavy{\G(w)} < \heavy{\G(v)}$.  
Since $\heavy{\G(v)} = \heavy{\F(v)}$, the status can change only if $\heavy{\G(w)} \ne \heavy{\F(w)}$.  
But as shown earlier, $\heavy{\G(w)} \ne \heavy{\F(w)}$ only when $w$ is the parent of a light ancestor of $u$, i.e., when $v$ is the grandparent of a light ancestor of~$u$.

This completes the case analysis, establishing that the modified label of $v$ changes only if $v$ is the parent of $u$, the parent of a light ancestor of $u$, or the grandparent of a light ancestor of $u$.  
By \cref{obs:heavy-light}, there are at most $\Oh(\log n)$ such nodes.  
Hence, $\G_{\hld}$ can be obtained from $\F_{\hld}$ by one node deletion and $\Oh(\log n)$ node relabelings.
\end{proof}

The analogous statement for node insertions follows immediately by swapping the roles of $\F$ and $\G$ in \cref{obs:node_delete}.

\begin{lemma}
\label{lem:overapprox}
For all forests $\F$ and $\G$ of size at most $n$, we have
\[
\ed(\paren(\F), \paren(\G)) \le 2 \, \ted(\F_\hld,\G_\hld) \le \Oh(\log n) \cdot \ted(\F, \G).
\]
\end{lemma}
\begin{proof}
For the first inequality, recall that $\ed(\paren(\F), \paren(\G)) = \ed(\str{\F_{\hld}}, \str{\G_{\hld}})$ by \cref{def:parenF}.  
Every tree alignment can be interpreted as a string alignment on the underlying parenthesis representations (by \cref{def:ta}), with each node edit corresponding to two character edits.  

For the second inequality, we show that $\ted(\F_{\hld}, \G_{\hld}) \le c (1+\log n) \cdot \ted(\F, \G)$ for a sufficiently large constant $c$.  
We proceed by induction on $\ted(\F, \G)$.

If $\ted(\F, \G) = 0$, then $\F_{\hld}=\G_{\hld}$ and the inequality holds trivially.  
Otherwise, let $\Hcal$ be an intermediate forest such that $\ted(\F, \Hcal) = 1$ and $\ted(\Hcal, \G) = \ted(\F, \G) - 1$.  
By the inductive hypothesis, $\ted(\Hcal_\hld, \G_\hld) \le c (1+\log n) \cdot \ted(\Hcal, \G)$.  

The edit from $\F$ to $\Hcal$ is a single node operation.  
If it is a relabeling, then \cref{obs:node_relabel} gives $\ted(\F_\hld, \Hcal_\hld) \le \Oh(\log n)$.  
If it is a deletion, then \cref{obs:node_delete} gives $\ted(\F_\hld, \Hcal_\hld) \le \Oh(\log n)$.  
If it is an insertion, the argument is symmetric to deletion, giving $\ted(\F_\hld, \Hcal_\hld) = \ted(\Hcal_\hld, \F_\hld) \le \Oh(\log n)$.  
In all cases, $\ted(\F_\hld, \Hcal_\hld) \le c (1+\log n)$ holds for a sufficiently large $c$.

By the triangle inequality,
\begin{multline*}
\ted(\F_\hld, \G_\hld) \le \ted(\F_\hld, \Hcal_\hld) + \ted(\Hcal_\hld, \G_\hld) 
\le c(1+\log n) + c(1+\log n) \cdot \ted(\Hcal, \G) \\ 
= c(1+\log n) \cdot \bigl(1 + \ted(\Hcal, \G)\bigr)
= c(1+\log n) \cdot \ted(\F, \G).
\end{multline*}
This completes the inductive argument.
\end{proof}

\subsection{String to Tree Alignment}
\label{subsec:approx proof}
Given forests $\F$ and $\G$ with their modified parenthesis representations $\paren(\F)$ and $\paren(\G)$, the main challenge in using a string alignment on these representations to approximate $\ted(\F, \G)$ is that an optimal string alignment may not respect the tree alignment constraints from Definition~\ref{def:ta}.

For an alignment $\A : \paren(\F) \onto \paren(\G)$, we say that a node $u$ in $\F$ (or $\G$) is \emph{tree-aligned} if both $o(u)$ and $c(u)$ are aligned under $\A$ to $o(v)$ and $c(v)$, respectively, for some node $v \in \G$. Otherwise, $u$ is \emph{misaligned}, meaning that it violates the tree alignment constraints.  

We distinguish three types of misaligned nodes. A node $u \in \F$ is \emph{partially deleted} if only one of $o(u)$ or $c(u)$ is aligned under $\A$, while the other twin is deleted. A node is a \emph{single-branch} misaligned node if $o(u) \sim_\A o(v_1)$ and $c(u) \sim_\A c(v_2)$ for distinct nodes $v_1, v_2 \in \G$ that lie on the same branch, i.e., $o(v_1) < o(v_2) < c(v_2) < c(v_1)$ or $o(v_2) < o(v_1) < c(v_1) < c(v_2)$. Finally, $u$ is a \emph{multi-branch} misaligned node if $o(u) \sim_\A o(v_1)$ and $c(u) \sim_\A c(v_2)$ for $v_1, v_2 \in \G$ that are on different branches, i.e., $o(v_1) < c(v_1) < o(v_2) < c(v_2)$. See Figure~\ref{fig:misaligned} for examples of each type.

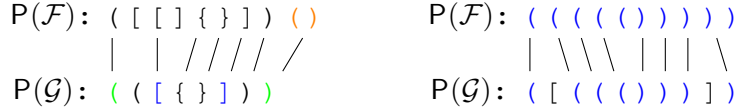
\begin{figure}
    \centering
\begin{tikzpicture}{\fontfamily{pcr}\selectfont 
    \node (F1) at (0, 0) {(};
    \node (F2) at (.3, 0) {[};
    \node (F3) at (.6, 0) {[};
    \node (F4) at (.9, 0) {]};
    \node (F5) at (1.2, 0) {\{};
    \node (F6) at (1.5, 0) {\}};
    \node (F7) at (1.8, 0) {]};
    \node (F8) at (2.1, 0) {)};
    \node (F9) at (2.4, 0) {\textcolor{orange}{(}};
    \node (F10) at (2.7, 0) {\textcolor{orange}{)}};
    \node (G1) at (0, -1) {\textcolor{green}{(}};
    \node (G2) at (.3, -1) {(};
    \node (G3) at (.6, -1) {\textcolor{blue}{[}};
    \node (G4) at (.9, -1) {\{};
    \node (G5) at (1.2, -1) {\}};
    \node (G6) at (1.5, -1) {\textcolor{blue}{]}};
    \node (G7) at (1.8, -1) {)};
    \node (G8) at (2.1, -1) {\textcolor{green}{)}};
    \node (F) at (-.8, 0) {$\str{\F}$:};
    \node (G) at (-.8, -1) {$\str{\G}$:};
    }
    \draw (F1) -- (G1);
    \draw (F3) -- (G3);
    \draw (F5) -- (G4);
    \draw (F6) -- (G5);
    \draw (F7) -- (G6);
    \draw (F8) -- (G7);
    \draw (F10) -- (G8);
\end{tikzpicture}\hspace{1cm}
\begin{tikzpicture}{[text color=red]\fontfamily{pcr}\selectfont\textcolor{blue}{
    \node (F1) at (0, 0) {(};
    \node (F2) at (.3, 0) {(};
    \node (F3) at (.6, 0) {(};
    \node (F4) at (.9, 0) {(};
    \node (F5) at (1.2, 0) {(};
    \node (F6) at (1.5, 0) {)};
    \node (F7) at (1.8, 0) {)};
    \node (F8) at (2.1, 0) {)};
    \node (F9) at (2.4, 0) {)};
    \node (F10) at (2.7, 0) {)};
    \node (G1) at (0, -1) {(};
    \node (G1b) at (.3, -1) {\textcolor{black}{[}};
    \node (G2) at (.6, -1) {(};
    \node (G3) at (.9, -1) {(};
    \node (G4) at (1.2, -1) {(};
    \node (G5) at (1.5, -1) {)};
    \node (G6) at (1.8, -1) {)};
    \node (G7) at (2.1, -1) {)};
    \node (G7b) at (2.4, -1) {\textcolor{black}{]}};
    \node (G8) at (2.7, -1) {)};}
    \node (F) at (-.8, 0) {$\str{\F}$:};
    \node (G) at (-.8, -1) {$\str{\G}$:};
    }
    \draw (F1) -- (G1);
    \draw (F2) -- (G2);
    \draw (F3) -- (G3);
    \draw (F4) -- (G4);
    \draw (F6) -- (G5);
    \draw (F7) -- (G6);
    \draw (F8) -- (G7);
    \draw (F9) -- (G8);
\end{tikzpicture}
    \caption{Example alignments between two parenthesis representations of forests $\F$ and $\G$. Parentheses with lines connecting them are aligned, and parentheses with no lines connecting them to another are deleted. (left) The orange pair of parentheses correspond to a partially deleted node. The blue pair of parentheses correspond to a single-branch misaligned node. The green pair of parentheses correspond to a multi-branch misaligned node.
    (right) The shown alignment contains a chain of nodes in $\F$ and $\G$ corresponding to the blue parentheses.}
    \label{fig:misaligned}
\end{figure}

A common structure in string alignments that violates the tree alignment constraints is a sequence of single-branch misaligned nodes where each node has one of its parentheses aligned to the previous node in the sequence and the other parenthesis aligned to the next node in the sequence. We formalize this notion as follows.

\begin{definition}[Chains]
\label{def:chain}
Let $\A: \paren(\F) \to \paren(\G)$ be an alignment. A \emph{chain} $C$ is a maximal sequence of at least three nodes alternating between $\F$ and $\G$, such that for each $i \in [0 \dd |C|-3]$, the node $C[i]$ is a proper ancestor of $C[i+2]$, and one of the following two conditions holds:
\begin{enumerate}
    \item For every $i \in [0 \dd |C|)$ with $C[i] \in \F$, we have $o(C[i]) \simeq_\A o(C[i+1])$ if $i < |C|-1$ and $c(C[i]) \simeq_\A c(C[i-1])$ if $i > 0$; in this case, $C$ is called an \emph{opening chain}.
    \item For every $i \in [0 \dd |C|)$ with $C[i] \in \F$, we have $c(C[i]) \simeq_\A c(C[i+1])$ if $i < |C|-1$ and $o(C[i]) \simeq_\A o(C[i-1])$ if $i > 0$; in this case, $C$ is called a \emph{closing chain}.
\end{enumerate}
\end{definition}

Note that in the above definition, we require that the aligned parentheses in each case are not only aligned under $\A$ but also match exactly, i.e., their labels are equal. Figure~\ref{fig:misaligned} illustrates a simple example of a chain. As a consequence of requiring exact matches, an important property of a chain is that the light subtree of every node in the chain is identical.

\begin{observation}
\label{obs:chain-rep}
Consider a chain $C$ with respect to an alignment $\A: \paren(\F) \to \paren(\G)$. 
All nodes in $C$ have identical heavy depth and light subtree. In particular, for any two nodes $C[i]$ and $C[j]$ in $C$:
\begin{enumerate*}[label=(\roman*)]
    \item if $C[i], C[j] \in \F$, then $\F'(C[i]) = \F'(C[j])$,
    \item if $C[i], C[j] \in \G$, then $\G'(C[i]) = \G'(C[j])$,
    \item if $C[i] \in \F$ and $C[j] \in \G$, then $\F'(C[i]) = \G'(C[j])$.
\end{enumerate*}
\end{observation}

\begin{proof}
We first show that $\lambda_{\hld}(C[i]) = \lambda_{\hld}(C[i+1])$ for each $i \in [0 \dd |C|-2]$.  
Since the chain alternates between nodes of $\F$ and $\G$, either $C[i] \in \F$ or $C[i+1] \in \F$.  
If $C[i] \in \F$, then by the definition of a chain we have $o(C[i]) \simeq_\A o(C[i+1])$ (for an opening chain) or $c(C[i]) \simeq_\A c(C[i+1])$ (for a closing chain). In either case, this implies $\lambda_{\hld}(C[i]) = \lambda_{\hld}(C[i+1])$.  
If $C[i+1] \in \F$, then similarly, $c(C[i+1]) \simeq_\A c(C[i])$ (opening chain) or $o(C[i+1]) \simeq_\A o(C[i])$ (closing chain), which also implies $\lambda_{\hld}(C[i]) = \lambda_{\hld}(C[i+1])$.  
By induction along the chain, this equality extends to any pair of nodes $C[i]$ and $C[j]$ in $C$.  
Finally, by Definition~\ref{def:parenF}, equality of modified labels implies that the light subtrees of $C[i]$ and $C[j]$ are identical.
\end{proof}

A simple consequence of \cref{obs:chain-rep} is that every chain includes nodes from a single heavy path in $\F$ and a single heavy path in $\G$.

\begin{observation}
\label{obs:chain-heavy}
Consider a chain $C$ with respect to an alignment $\A: \paren(\F) \to \paren(\G)$. 
The nodes of $C$ belong to exactly two heavy paths (one in $\F$ and one in $\G$), and each node in $C$ has a heavy child.
\end{observation}
\begin{proof}
By \cref{def:chain}, for each $i\in [0\dd |C|-3]$, the node $C[i]$ is an ancestor of $C[i+2]$.  
By \cref{obs:chain-rep}, these nodes share the same modified label and hence have equal heavy depths.  
Therefore, all nodes on the path from $C[i]$ to $C[i+2]$ have the same heavy depth, which implies that all intermediate edges are heavy and that $C[i]$ and $C[i+2]$ lie on the same heavy path.  

Since the nodes of $C$ alternate between $\F$ and $\G$, a simple induction along the chain shows that the nodes from $\F$ lie on a single heavy path and the nodes from $\G$ lie on a single heavy path.  

In particular, because $|C|\ge 3$, the path from $C[0]$ to $C[2]$ consists of heavy edges.  
Thus, $C[0]$ has a heavy child, which is represented by a sentinel node with label $\#$ in the light subtree of $C[0]$.  
The light subtrees of all nodes in $C$ are identical, so every node in $C$ must have a heavy child.
\end{proof}

Misaligned nodes can create a substantial gap between the string edit distance and the tree edit distance of $\paren(\F)$ and $\paren(\G)$. Chains introduce a certain flexibility in string alignments that can reduce the number of edits required, but this often comes at the cost of violating tree alignment constraints. Figure~\ref{fig:chaingap} illustrates a pair of forests whose parenthesis representations have string edit distance $\Oh(1)$ but tree edit distance $\Oh(\sqrt{n})$. By identifying chains and correcting the misaligned nodes within them, we can transform any string alignment $\A: \paren(\F) \to \paren(\G)$ into a valid tree alignment, incurring at most an additional multiplicative factor of $\Oh(\sqrt{n})$ in the number of edits.
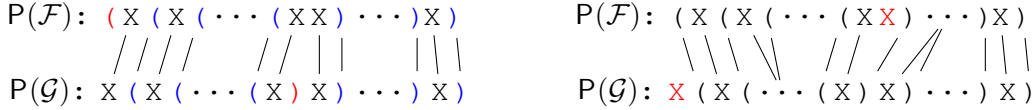
\begin{figure}
    \centering
\begin{tikzpicture}
{\fontfamily{pcr}\selectfont 
    \node(F1) at (0, 0) {\textcolor{red}{(}};
    \node (F2) at (.3, 0) {X};
    \node (F3) at (.6, 0) {\textcolor{blue}{(}};
    \node (F4) at (.9, 0) {X};
    \node (F5) at (1.2, 0) {\textcolor{blue}{(}};
    \node (F6) at (1.7, 0) {...};
    \node (F7) at (2.2, 0) {\textcolor{blue}{(}};
    \node (F8) at (2.5, 0) {X};
    \node (F9) at (2.8, 0) {X};
    \node (F10) at (3.1, 0) {\textcolor{blue}{)}};
    \node (F11) at (3.6, 0) {...};
    \node (F12) at (4.1, 0) {\textcolor{blue}{)}};
    \node (F13) at (4.3, 0) {X};
    \node (F14) at (4.6, 0) {\textcolor{blue}{)}};
    \node (G1) at (0, -1) {X};
    \node (G2) at (.3, -1) {\textcolor{blue}{(}};
    \node (G3) at (.6, -1) {X};
    \node (G4) at (.9, -1) {\textcolor{blue}{(}};
    \node (G5) at (1.4, -1) {...};
    \node (G6) at (1.9, -1) {\textcolor{blue}{(}};
    \node (G7) at (2.2, -1) {X};
    \node (G8) at (2.5, -1) {\textcolor{red}{)}};
    \node (G9) at (2.8, -1) {X};
    \node (G10) at (3.1, -1) {\textcolor{blue}{)}};
    \node (G11) at (3.6, -1) {...};
    \node (G12) at (4.1, -1) {\textcolor{blue}{)}};
    \node (G13) at (4.4, -1) {X};
    \node (G14) at (4.7, -1) {\textcolor{blue}{)}};
    \node (F) at (-.8, 0) {$\str{\F}$:};
    \node (G) at (-.8, -1) {$\str{\G}$:};
    }
    \draw (F2) -- (G1);
    \draw (F3) -- (G2);
    \draw (F4) -- (G3);
    \draw (F5) -- (G4);
    \draw (F7) -- (G6);
    \draw (F8) -- (G7);
    \draw (F9) -- (G9);
    \draw (F10) -- (G10);
    \draw (F12) -- (G12);
    \draw (F13) -- (G13);
    \draw (F14) -- (G14);
\end{tikzpicture}\hspace{1cm}
\begin{tikzpicture}
{\fontfamily{pcr}\selectfont 
    \node(F1) at (0, 0) {(};
    \node (F2) at (.3, 0) {X};
    \node (F3) at (.6, 0) {(};
    \node (F4) at (.9, 0) {X};
    \node (F5) at (1.2, 0) {(};
    \node (F6) at (1.7, 0) {...};
    \node (F7) at (2.2, 0) {(};
    \node (F8) at (2.5, 0) {X};
    \node (F9) at (2.8, 0) {\textcolor{red}{X}};
    \node (F10) at (3.1, 0) {)};
    \node (F11) at (3.6, 0) {...};
    \node (F12) at (4.1, 0) {)};
    \node (F13) at (4.3, 0) {X};
    \node (F14) at (4.6, 0) {)};
    \node (G1) at (0, -1) {\textcolor{red}{X}};
    \node (G2) at (.3, -1) {(};
    \node (G3) at (.6, -1) {X};
    \node (G4) at (.9, -1) {(};
    \node (G5) at (1.4, -1) {...};
    \node (G6) at (1.9, -1) {(};
    \node (G7) at (2.2, -1) {X};
    \node (G8) at (2.5, -1) {)};
    \node (G9) at (2.8, -1) {X};
    \node (G10) at (3.1, -1) {)};
    \node (G11) at (3.6, -1) {...};
    \node (G12) at (4.1, -1) {)};
    \node (G13) at (4.4, -1) {X};
    \node (G14) at (4.7, -1) {)};
    \node (F) at (-.8, 0) {$\str{\F}$:};
    \node (G) at (-.8, -1) {$\str{\G}$:};
    }
    \draw (F1) -- (G2);
    \draw (F2) -- (G3);
    \draw (F3) -- (G4);
    \draw (F4) -- (G5);
    \draw (F5) -- (G5);
    \draw (F7) -- (G6);
    \draw (F8) -- (G7);
    \draw (F10) -- (G8);
    \draw (F11) -- (G9);
    \draw (F11) -- (G10);
    \draw (F12) -- (G12);
    \draw (F13) -- (G13);
    \draw (F14) -- (G14);
\end{tikzpicture}
    \caption{Pictured above are a string alignment (left) and a tree alignment (right) for the parentheses representations of a forest $\F$ and $\G$. Parentheses with lines connecting them are aligned, and parentheses with no lines connecting them to another are deleted. $X$ represents a substring of length $\sqrt{n}/4$.  Using a chain, an optimal string alignment only requires two deletions of parentheses highlighted in red. On the other hand, an optimal tree alignment must delete two copies of substring $X$ to satisfy the constraints of Definition~\ref{def:ta}. The difference in the optimal string edit distance and tree edit distance is $\Oh(|X| - 2) = \Oh(\sqrt{n})$.}
    \label{fig:chaingap}
\end{figure}

\subsection{String to Tree Alignment Construction}\label{subsec:constr}
\newcommand{\M}{\mathcal{M}_{\A}}
In this section, we describe how to transform a string alignment $\A: \paren(\F) \to \paren(\G)$ into a valid tree alignment, at the cost of increasing the number of edits by a factor of $\Oh(\sqrt{n})$.  
As a first step, we replace all substitutions in $\A$ with deletions, which can at most double the cost of $\A$. This converts $\A$ into a deletion-only alignment, which we interpret via the underlying \emph{monotone matching}  
\[
\M = \{(x, y) \in [0\dd |\paren(\F)|) \times [0\dd |\paren(\G)|) : \paren(\F) \simeq_\A \paren(\G)\}.
\]

\subsubsection{Tree-Aligning a Misaligned Node}  
We begin by describing how to fix a single misaligned node $u$ in $\F$. 
In most cases, we simply remove from $\M$ all pairs containing $o(u)$ or $c(u)$; as a result, $\A$ is modified to delete $u$.
However, if $u$ belongs to a chain that continues with a node $v \in \G$, we sometimes match $u$ with $v$, and more generally, match the light subtrees $\F'(u)$ and $\G'(v)$. 
This requires adding $(o(u'), o(v'))$ and $(c(u'), c(v'))$ to $\M$ for every pair of corresponding nodes $u' \in \F'(u)$ and $v' \in \G'(v)$ (except for the special $\#$-labeled node). 
To maintain $\M$ as a monotone matching, we also remove existing pairs that conflict with newly added pairs.

We now formalize the two cases. For two nodes $v, v'$ in a forest $\F$, we define $w(v, v')\coloneqq |o(v) - o(v')| + |c(v) - c(v')|$. 
Note that $w(v, v') = 2|\F(v) \setminus \F(v')|$ holds if $v$ is an ancestor of $v'$.  

We refer to the following subroutine as $\mathsf{TreeAlign}(\A, u)$, which is defined for a single node $u$ and applied iteratively, as discussed in Subsection~\ref{subsubsec:allheavy}.

\medskip
\noindent$\mathsf{TreeAlign}(\A, u)$:
\begin{enumerate}
    \item If there exists a chain $C$ such that
    \begin{itemize}
        \item $C[k] = u$ for some $k \in [0 \dd |C| - 5]$ and
        \item $w(C[k], C[k+4]) < \sqrt{n}$,
    \end{itemize}
    then modify $\A$ to match the light subtree of $u$ with the light subtree of $v\coloneqq C[k+1]$.

    \item Otherwise, modify $\A$ to delete node $u$.
    
\end{enumerate}

In the above routine, we use two short-hand phrases ''match the light subtree of $u$ with the light subtree of $C[k+1]$'' and ``delete node $u$'' when discussing modifications to $\A$. We now rigorously define these two steps. Recall an alignment matches parentheses of strings $X$ and $y$ at indices $i, j$, respectively, when it contains pairs $(i, j), (i + 1, j + 1)$ and deletes a parenthesis at index $i$ in $X$ when it contains pairs $(i, j), (i+1, j)$ with a symmetric case for deletion in $Y$.

\textbf{Delete node $u$:} To ``\emph{delete a node $u$}'' for a node $u \in \F$, any matches of $o(u)$ and $c(u)$ are replaced with deletions in alignment $\A$. For example, if $o(u) \simeq_\A o(v)$ for $u \in \F, v \in G$ in a deletion-only alignment $\A: \paren(\F) \rightarrow \paren(\G)$, then we replace pairs $(o(u), o(v)), (o(u) + 1, o(v) + 1)$ with pairs $(o(u), o(v)), (o(u) + 1, o(v)), (o(u) + 1, o(v) + 1)$.  Note that this deletes both $o(u)$ and $o(v)$.  If instead $o(u)$ is not matched, we make no changes to $\A$. We do the analogous modifications of $\A$ for $c(u)$. 

\textbf{Match light subtrees of $u$ and $v$:} To ``\emph{match the light subtree of $u$ with the light subtree of $v$}'' for nodes $u \in \F, v \in G$ such that $u$ and $v$ are consecutive nodes in a chain $C$ of $\A$, $\A$ is modified to match all parentheses of $\F'(u)$ to those in $\G'(v)$ and delete any parentheses of $\F$ or $\G$ that no longer have a valid match in $\A$ according to the alignment definition, Definition~\ref{def:alignment}. This modification is much more involved than deleting a node. By \cref{obs:chain-heavy}, the node $u=C[k]$ an $v=C[k+1]$ have heavy children, which we denote $u'$ and $v'$, respectively. Every chain must either be an opening chain or a closing chain such that either $o(u) \simeq_\A o(v)$ or $c(u) \simeq_\A c(v)$, respectively.  Without loss of generality, we assume that $c(u) \simeq_\A c(v)$, and so pairs $(c(u), c(v)), (c(u) + 1, c(v) + 1)$ are in $\A$. 

We begin our modifications of $\A$ to match $\paren(\F)(c(u')\dd c(u)]$ to $\paren(\G)(c(v')\dd c(v)]$. Figure~\ref{fig:treealign-closing} provides a depiction of the changes we will make to $\A$ for these substrings, which we describe in detail as follows. In order to maintain alignment constraints as per Definition~\ref{def:alignment}, before we can add anything, we first remove any previous pairs of $\A$ containing an index corresponding to a parenthesis in $\paren(\F)(c(u')\dd c(u)]$ or $\paren(\G)(c(v')\dd c(v)]$. Let $(i_\F, i_\G) \in A$ be the rightmost pair in $\A$ such that $i_\F \leq c(u')$ and $i_\G \leq c(v')$. By the definition of alignments, it must be the case that at least one of $i_\F = c(u')$ or $i_\G = c(v')$ is true; without loss of generality, let $i_\F = c(u')$. We remove all pairs of $\A$ between pairs $(c(u'), i_\G)$ and $(c(u) + 1, c(v) + 1)$, exclusive. We insert a matching sequence $(c(u'), c(v')), (c(u') + 1, c(v') + 1), (c(u') + 2, c(v') + 2), \ldots, (c(u), c(v))$ into $\A$ at the position before pair $(c(u) + 1, c(v) + 1)$. We then insert a deletion sequence $(c(u'), i_\G + 1), (c(u'), i_\G + 2), \ldots, (c(u'), c(v') - 1)$ into $\A$ at the position after pair $(c(u'), i_\G)$. This completes the matching of $\paren(\F)(c(u')\dd c(u)]$ to $\paren(\G)(c(v')\dd c(v)]$ and deletion of any additional parentheses needed to maintain alignment constraints.

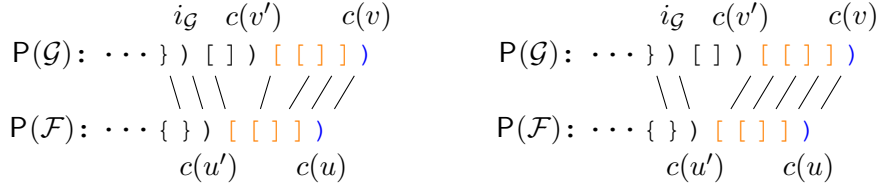
\begin{figure}
    \centering
\begin{tikzpicture}{\fontfamily{pcr}\selectfont 
    \node (F1) at (0, 0) {...};
    \node (F2) at (.5, 0) {\}};
    \node (F3) at (.8, 0) {)};
    \node (F4) at (1.1, 0) {[};
    \node (F5) at (1.4, 0) {]};
    \node (F6) at (1.7, 0) {)};
    \node (F7) at (2, 0) {\textcolor{orange}{[}};
    \node (F8) at (2.3, 0) {\textcolor{orange}{[}};    
    \node (F9) at (2.6, 0) {\textcolor{orange}{]}};
    \node (F10) at (2.9, 0) {\textcolor{orange}{]}};
    \node (F11) at (3.2, 0) {\textcolor{blue}{)}};
    \node (G1) at (0, -1) {...};
    \node (G2) at (.5, -1) {\{};
    \node (G3) at (.8, -1) {\}};
    \node (G4) at (1.1, -1) {)};
    \node (G5) at (1.4, -1) {\textcolor{orange}{[}};
    \node (G6) at (1.7, -1) {\textcolor{orange}{[}};
    \node (G7) at (2, -1) {\textcolor{orange}{]}};
    \node (G8) at (2.3, -1) {\textcolor{orange}{]}};
    \node (G9) at (2.6, -1) {\textcolor{blue}{)}};
    \node (F) at (-1, 0) {$\str{\G}$:};
    \node (G) at (-1, -1) {$\str{\F}$:};
    }
    \node (u) at (2.6, -1.5) {$c(u)$};
    \node (u') at (1.1, -1.5) {$c(u')$};
    \node (v') at (1.7, .5) {$c(v')$};
    \node (v) at (3.2, .5) {$c(v)$};
    \node (iG) at (.8, .5) {$i_\G$};
    \draw (F2) -- (G3);
    \draw (F3) -- (G4);
    \draw (F4) -- (G5);
    \draw (F7) -- (G6);
    \draw (F9) -- (G7);
    \draw (F10) -- (G8);
    \draw (F11) -- (G9);
\end{tikzpicture}\hspace{1cm}
\begin{tikzpicture}{\fontfamily{pcr}\selectfont 
    \node (F1) at (0, 0) {...};
    \node (F2) at (.5, 0) {\}};
    \node (F3) at (.8, 0) {)};
    \node (F4) at (1.1, 0) {[};
    \node (F5) at (1.4, 0) {]};
    \node (F6) at (1.7, 0) {)};
    \node (F7) at (2, 0) {\textcolor{orange}{[}};
    \node (F8) at (2.3, 0) {\textcolor{orange}{[}};    
    \node (F9) at (2.6, 0) {\textcolor{orange}{]}};
    \node (F10) at (2.9, 0) {\textcolor{orange}{]}};
    \node (F11) at (3.2, 0) {\textcolor{blue}{)}};
    \node (G1) at (0, -1) {...};
    \node (G2) at (.5, -1) {\{};
    \node (G3) at (.8, -1) {\}};
    \node (G4) at (1.1, -1) {)};
    \node (G5) at (1.4, -1) {\textcolor{orange}{[}};
    \node (G6) at (1.7, -1) {\textcolor{orange}{[}};
    \node (G7) at (2, -1) {\textcolor{orange}{]}};
    \node (G8) at (2.3, -1) {\textcolor{orange}{]}};
    \node (G9) at (2.6, -1) {\textcolor{blue}{)}};
    \node (F) at (-1, 0) {$\str{\G}$:};
    \node (G) at (-1, -1) {$\str{\F}$:};
    }
    \node (u) at (2.6, -1.5) {$c(u)$};
    \node (u') at (1.1, -1.5) {$c(u')$};
    \node (v') at (1.7, .5) {$c(v')$};
    \node (v) at (3.2, .5) {$c(v)$};
    \node (iG) at (.8, .5) {$i_\G$};
    \draw (F2) -- (G3);
    \draw (F3) -- (G4);
    \draw (F7) -- (G5);
    \draw (F8) -- (G6);
    \draw (F9) -- (G7);
    \draw (F10) -- (G8);
    \draw (F11) -- (G9);
\end{tikzpicture}
    \caption{Example of an alignment $\A$ (left) and the result of $\mathsf{TreeAlign}(\A, u)$ (right) for a node $u$ in Case 1 of the $\mathsf{TreeAlign}$ routine. Parentheses with lines connecting them are aligned, and parentheses with no lines connecting them to another are deleted. For a chain $C$, $u = C[k]$ corresponds to the rightmost (blue) parenthesis of $\paren(\F)$ and $v = C[k+1]$ corresponds to the rightmost (blue) parenthesis in $\paren(\G)$. After $\mathsf{TreeAlign}(\A, u)$, $o(u) \simeq_\A o(v)$ and the rest of the light subtrees of $u$ and $v$, depicted by the orange parentheses, are matched as well. The parentheses previously aligned with the light subtree of $u$, which correspond to range $[i_\G, c(v))$ are now deleted.}
    \label{fig:treealign-closing}
\end{figure}

We now finish matching the light subtrees of $u$ and $v$ by modifying $\A$ to match $\paren(\F)[o(u)\dd o(u'))$ and $\paren(\G)[o(v)\dd o(v'))$. See Figure~\ref{fig:treealign} for a depiction of these modifications to $\A$ that we describe formally as follows. Let $(o(u), i_\G)$ be the rightmost pair in $\A$ with left index $o(u)$ and let $(i_\F, o(v'))$ be the leftmost pair in $\A$ with right index $o(v')$. Observe that $o(v) \simeq_\A o(C[k+2])$ with $o(u) < o(C[k+2])$ by our earlier assumption of $C$.  Since $o(u) < o(C[k+2])$, $o(v') > o(v)$ and $(o(C[k+2], o(v)) \in \A$, by the monotonicity of alignments, $i_G \leq o(v)$ and $i_\F \geq o(u')$.  We remove all pairs of $\A$ between $(o(u), i_\G)$ and $(i_\F, o(v'))$, exclusive. We insert deletion sequence $(o(u), i_\G + 1),(o(u), i_\G + 2), \ldots, (o(u), o(v))$ into $\A$ at the position after pair $(o(u), i_\G)$. Next, we insert matching sequence $(o(u) + 1, o(v) + 1), (o(u) + 2, o(v) + 2), \ldots, (o(u'), o(v'))$ into $\A$  at the position after pair $(o(u), o(v))$. Finally, we insert deletion sequence $(o(u') + 1, o(v')), (o(u') + 2, o(v')), \ldots, (i_\F - 1, o(v'))$ into $\A$ at the position after pair $(o(u'), o(v'))$. The above steps modify $\A$ to match substrings $\paren(\F)[o(u)\dd  o(u'))$ to $\paren(\G)[o(v)\dd o(v'))$ while maintaining alignment constraints.
\begin{figure}
    \centering
\begin{tikzpicture}{\fontfamily{pcr}\selectfont 
    \node (F1) at (0, 0) {\textcolor{blue}{(}};
    \node (F2) at (.3, 0) {[};
    \node (F3) at (.6, 0) {[};
    \node (F4) at (.9, 0) {]};
    \node (F5) at (1.2, 0) {]};
    \node (F6) at (1.5, 0) {\textcolor{blue}{(}};
    \node (F7) at (1.8, 0) {\textcolor{orange}{[}};    
    \node (F8) at (2.1, 0) {\textcolor{orange}{]}};
    \node (F9) at (2.4, 0) {\{};
    \node (F10) at (2.9, 0) {...};
    \node (G1) at (0, -1) {\{};
    \node (G2) at (.3, -1) {\}};
    \node (G3) at (.6, -1) {\textcolor{blue}{(}};
    \node (G4) at (.9, -1) {\textcolor{orange}{[}};
    \node (G5) at (1.2, -1) {\textcolor{orange}{]}};
    \node (G6) at (1.5, -1) {(};
    \node (G7) at (1.8, -1) {\textcolor{blue}{(}};
    \node (G8) at (2.1, -1) {[};
    \node (G9) at (2.4, -1) {]};
    \node (G10) at (2.7, -1) {\{};
    \node (G11) at (3.2, -1) {...};
    \node (F) at (-.8, 0) {$\str{\G}$:};
    \node (G) at (-.8, -1) {$\str{\F}$:};
    }
    \node (u) at (.7, -1.5) {$o(u)$};
    \node (v) at (1.6, .5) {$o(v)$};
    \node (iG) at (-.1, .5) {$i_\G$};
    \node (iF) at (2.8, -1.5) {$i_\F$};
    \draw (F1) -- (G3);
    \draw (F3) -- (G4);
    \draw (F4) -- (G5);
    \draw (F6) -- (G7);
    \draw (F7) -- (G8);
    \draw (F8) -- (G9);
    \draw (F9) -- (G10);
\end{tikzpicture}\hspace{1cm}
\begin{tikzpicture}{\fontfamily{pcr}\selectfont 
    \node (F1) at (0, 0) {\textcolor{blue}{(}};
    \node (F2) at (.3, 0) {[};
    \node (F3) at (.6, 0) {[};
    \node (F4) at (.9, 0) {]};
    \node (F5) at (1.2, 0) {]};
    \node (F6) at (1.5, 0) {\textcolor{blue}{(}};
    \node (F7) at (1.8, 0) {\textcolor{orange}{[}};    
    \node (F8) at (2.1, 0) {\textcolor{orange}{]}};
    \node (F9) at (2.4, 0) {\{};
    \node (F10) at (2.9, 0) {...};
    \node (G1) at (0, -1) {\{};
    \node (G2) at (.3, -1) {\}};
    \node (G3) at (.6, -1) {\textcolor{blue}{(}};
    \node (G4) at (.9, -1) {\textcolor{orange}{[}};
    \node (G5) at (1.2, -1) {\textcolor{orange}{]}};
    \node (G6) at (1.5, -1) {(};
    \node (G7) at (1.8, -1) {\textcolor{blue}{(}};
    \node (G8) at (2.1, -1) {[};
    \node (G9) at (2.4, -1) {]};
    \node (G10) at (2.7, -1) {\{};
    \node (G11) at (3.2, -1) {...};
    \node (F) at (-.8, 0) {$\str{\G}$:};
    \node (G) at (-.8, -1) {$\str{\F}$:};
    }
    \node (u) at (.7, -1.5) {$o(u)$};
    \node (v) at (1.6, .5) {$o(v)$};
    \node (iG) at (-.1, .5) {$i_\G$};
    \node (iF) at (2.8, -1.5) {$i_\F$};
    \draw (F6) -- (G3);
    \draw (F7) -- (G4);
    \draw (F8) -- (G5);
    \draw (F9) -- (G10);
\end{tikzpicture}
    \caption{Example of an alignment $\A$ (left) and the result of $\mathsf{TreeAlign}(\A, u)$ (right) for a node $u$ in Case 1 of the $\mathsf{TreeAlign}$ routine. Parentheses with lines connecting them are aligned, and parentheses with no lines connecting them to another are deleted. For a chain $C$, $u = C[k]$ corresponds to the leftmost (blue) parenthesis of $\paren(\F)$ and $v = C[k+1]$ corresponds to the rightmost (blue) parenthesis in $\paren(\G)$. After $\mathsf{TreeAlign}(\A, u)$, $o(u) \simeq_\A o(v)$ and the rest of the light subtrees of $u$ and $v$, depicted by the orange parentheses, are matched as well. The parentheses previously aligned with the light subtrees of $u$ and $v$, which correspond to ranges $[o(u'), i_\F)$ and $[i_\G, o(v))$ are now deleted (where $u'$ is the heavy child of $u$, corresponding to the black opening parenthesis '\texttt{(}' in $\F$).}
    \label{fig:treealign}
\end{figure}
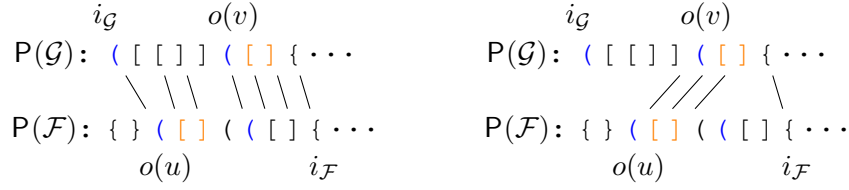

\subsubsection{Tree-Aligning All Nodes}
\label{subsubsec:allheavy}
We iteratively apply the $\mathsf{TreeAlign}$ subroutine to fix all misaligned nodes of $\A$ in $\F$.  
Let $\hvy = \{H_1, H_2, \ldots, H_\ell\}$ denote the set of all maximal heavy paths in $\F$, where each $H_i = (u^i_1, u^i_2, \ldots, u^i_{m_i})$ and the paths in $\hvy$ form a partition of the nodes of $\F$.  
By \cref{def:heavy-light}, any two nodes $u^i_{j_1}, u^i_{j_2} \in H_i$ share the same heavy depth.  
We order the heavy paths in $\hvy$ by increasing heavy depth.  
Starting with $H_1$, the path of smallest heavy depth, we process its nodes from top to bottom, invoking $\mathsf{TreeAlign}$ for each misaligned node.  
Formally, let $H_1 = (u^1_1, u^1_2, \ldots, u^1_{m_1})$, where $u^1_i$ is an ancestor of $u^1_j$ for all $i \le j$.  
For every $i$ from $1$ to $m_1$, we call $\mathsf{TreeAlign}(u^1_i, \A)$ if $u^1_i$ is misaligned.  
After processing $H_1$, we repeat the procedure for $H_2$, $H_3$, and so on, in order of increasing heavy depth, until all nodes have been processed.

\subsection{Tree Alignment Constraint Analysis}
\label{subsec:treeconstraints}

First, we show that for alignment $\A: \paren(\F) \onto \paren(\G)$ and node $u \in \F$, $\mathsf{TreeAlign}(\A, u)$ does in fact yield a sequence of pairs satisfying the alignment definition. 

\begin{lemma}
\label{lem:align-constraint}
    Given a deletion-only alignment $\A$ and a misaligned node $u \in \F$, sequence $\A$ after performing $\mathsf{TreeAlign}(\A, u)$ is an alignment. 
\end{lemma}

\begin{proof}
    By Definition~\ref{def:alignment}, we must make sure that $(0, 0), (|\paren(\F)|, |\paren(\G)| \in \A$ and $(x_{t_1}, y_{t+1}) \in {(x_t + 1, y_t + 1), (x_t + 1, y_t), (x_t, y_t + 1)}$ for all $t \in [0, m)$ where $\A = (x_t, y_t)_{t=0}^m$ after modifying $\A$. Note that in both cases of $\mathsf{TreeAlign}(\A, u)$, we never replace the first or last pair of $\A$, so we can focus on the third constraint as the first two are trivially true. It is easy to see that both cases of $\mathsf{TreeAlign}$ follow this constraint, but in the following, we show this explicitly for the sake of completion.
    If $u$ is handled in Case 2 of $\mathsf{TreeAlign}$, it is deleted. In this case pairs $(o(u), o(v)),(o(u) + 1, j+ 1)$ are replaced with $(o(u), o(v)),(o(u) + 1, o(v)), (o(u) + 1, o(v) + 1)$ or $(c(u), c(v)), (c(u) + 1, c(v) + 1)$ are replaced with $(c(u), c(v)), (c(u) + 1, c(v)), (c(u) + 1, c(v) + 1)$.  The only change to $\A$ are new pairs $(o(u) + 1, o(v))$ and $(c(u) + 1, c(v))$, which both satisfy the alignment constraint.

    If $u$ is handled in Case 1 of $\mathsf{TreeAlign}$, its light subtree is matched to the light subtree of a node $v$ where $C[k] = u, C[k+1] = v$ for some chain $C$. As in case 1, we denote $u'$ as the heavy child of $u$ and $v'$ as the heavy child of $v$. We assume without loss of generality that $c(u) \simeq_\A c(v)$ and that the rightmost pair $(i_\F, i_\G)$ in $\A$ satisfying $i_\F \leq c(u'), i_\G \leq c(v')$ has $i_\G = c(v')$ (see the previous ``\textbf{Match light subtrees}'' section for discussion of these assumptions). First, the $\mathsf{TreeAlign}$ routine replaces the sequence of pairs between $(c(u'), i_\G)$ and $(c(u) + 1, c(v) + 1)$ with the sequence $(c(u'), i_\G + 1), (c(u'), i_\G + 2), \ldots, (c(u'), c(v') - 1), (c(u'), c(v')), (c(u') - 1, c(v') - 1), (c(u') - 2, c(v') - 2), \ldots, (c(u), c(v))$. Thus, each index increments by at most one between every pair in the new alignment $\A$ after these modifications and so, the alignment constraint is satisfied. Second, the $\mathsf{TreeAlign}$ routine replaces the sequence of pairs between $(o(u), i_\G)$ and $(i_\F, o(v'))$ (where $(o(u), i_\G)$ is the first pair in $\A$ with left index $o(u)$ and $(i_\F, o(v'))$ is the first pair in $\A$ with right index $o(v')$) with sequence $(o(u), i_\G + 1),(o(u), i_\G + 2), \ldots, (o(u), o(v)), (o(u) + 1, o(v) + 1), (o(u) + 2, o(v) + 2), \ldots, (o(u'), o(v')), (o(u') + 1, o(v')), (o(u') + 2, o(v')), \ldots, (i_\F - 1, o(v'))$. Thus, again each index increments by at most one between every pair in the new alignment $\A$ after these modifications, and so $\A$ remains an alignment in either case of $\mathsf{TreeAlign}$.
\end{proof}

Recall that $\mathsf{TreeAlign}(\A)$ denotes the alignment $\A$ after performing $\mathsf{TreeAlign}(\A, u)$ on all misaligned nodes. 
Next, we show that in addition to $\mathsf{TreeAlign}(\A)$ satisfying the alignment requirements, it is also a tree alignment.  To this end, we show that there are no nodes of $\F$ misaligned by $\mathsf{TreeAlign}(\A)$. A useful observation about alignments is that if there is a misaligned node of $\G$, then there must be a misaligned node of $\F$. Therefore, if there are no misaligned nodes of $\F$, then all nodes of $\F$ and $\G$ are tree aligned by $\mathsf{TreeAlign}(\A)$. The proof of this observation is as follows.

\begin{observation}
\label{obs:F-only}
    Given an alignment $\A: \paren(\F) \to \paren(\G)$, if $u \in \G$ is misaligned by $\A$, then there must be some $v \in \F$ that is also misaligned.
\end{observation}

\begin{proof}
    Let $u \in \G$ be any misaligned node by $\A$.  By definition of misaligned nodes, at least one of $o(u)$ or $c(u)$ must be matched by $\A$. If $o(u) \simeq_\A o(v)$ for some $v \in \F$, by definition of misaligned nodes, $c(u)$ is either deleted or matched to a different parenthesis $c(v')$ for $v' \ne v$. Therefore, $c(v)$ does not match to $c(u)$ by $\A$, and so, $v$ is also misaligned. A symmetric argument can be shown for when $c(u) \simeq_\A c(v)$ and $o(u)$ is deleted or matched to a different parenthesis $c(v')$ for $v' \ne v$.
\end{proof}

As a precursor to the final lemma of the section, we show that when running $\mathsf{TreeAlign}$ on a misaligned node $u'$ in a heavy path $H$, any other node $u$ in $H$ previously aligned by $\mathsf{TreeAlign}$ remains unaffected by the modifications made during the alignment of $u'$.

\begin{lemma}
\label{lem:invariant}
Consider forests $\F$ and $\G$, a deletion-only alignment $\A: \F \onto \G$, and a heavy path $H$ in $\F$. 
If $\A$ tree-aligns a node $u \in H$ to a node $v \in \G$, then after performing $\mathsf{TreeAlign}(\A, u')$ for any other $u' \in H$, the alignment $\A$ still tree-aligns $u$ to $v$.
\end{lemma}
\begin{proof}
    Denote $H = (u_1, \ldots, u_m)$, with nodes listed from highest to lowest, and suppose $u = u_j$ and $u' = u_i$ for $i \neq j$.

    First, if $u_i$ is already tree-aligned by $\A$, then $\mathsf{TreeAlign}(\A, u_i)$ makes no modifications to $\A$.  
    If $u_i$ is misaligned but falls into Case~2 of $\mathsf{TreeAlign}(\A, u_i)$, the claim follows immediately, since $u_i$ is the only node of $\F$ affected by the updates to $\A$.

    If $u_i$ is misaligned and falls into Case~1, we examine the two steps performed by $\mathsf{TreeAlign}(\A, u_i)$ and their effect on a node $u_j \in H$ that is tree-aligned with $v \in \G$ under $\A$.  
    We will argue that these steps leave the alignment pairs for $u_j$ and $v$ unchanged; otherwise, the monotonicity of $\A$ would be violated.

    By the conditions of Case~1, $u_i$ belongs to a chain $C$ such that $C[k] = u_i$ for some $k \in [0 \dd |C| - 5]$.  
    By \cref{obs:chain-heavy}, $u_i$ has a heavy child $u_{i+1}$, and likewise $v_i \coloneqq C[k+1]$ has a heavy child $v_{i+1}$.  
    By symmetry (up to transposition), we may assume without loss of generality that the chain $C$ is closing, and therefore $c(u_i) \simeq_\A c(v_i)$.

    Our algorithm modifies $\A$ in two steps, which we analyze separately.  
    For the first step, let $(i_\F, i_\G) \in \A$ be the rightmost pair such that $i_\F \le c(u_{i+1})$ and $i_\G \le c(v_{i+1})$.  
    The first step of $\mathsf{TreeAlign}$ then modifies all pairs strictly between $(i_\F, i_\G)$ and $(c(u_i) + 1, c(v_i) + 1)$.  
    See Figure~\ref{fig:treealign-closing} for an illustration of these modifications.

    By definition, either $i_\F = c(u_{i+1})$ or $i_\G = c(v_{i+1})$.  
    In either case, we claim that $o(u_j) < i_\F$.  
    If $i_\F = c(u_{i+1})$, this is immediate since the opening parentheses corresponding to nodes in $H$ always appear before their closing parentheses in $\paren(\F)$.  
    If $i_\G = c(v_{i+1})$, the claim follows from the monotonicity of $\A$: since $(i_\F, c(v_{i+1}))$ and $(c(C[k+2]), c(C[k+3]))$ are in $\A$ with $c(C[k+3]) \le c(v_{i+1})$, it follows that $o(u_j) < c(C[k+2]) \le i_\F$.  
    Therefore, the position $o(u_j)$ lies to the left of the modified interval, and the first step of $\mathsf{TreeAlign}$ leaves $o(u_j)$ and its alignment untouched.

    If $j < i$, then $u_j$ is an ancestor of $u_i$ with $c(u_j) > c(u_i)$, so its closing parenthesis lies strictly to the right of all indices modified in the first step.  
    Hence, $c(u_j)$ is also unaffected in this case.  

    If $j > i$, then $c(u_j)$ could only be affected if $i_\F < c(u_j) < c(u_i)$.  
    However, since $c(u_j) \le c(u_{i+1})$, this would imply $i_\F < c(u_{i+1})$, and therefore $i_\G = c(v_{i+1})$.  
    Because $(i_\F, c(v_{i+1}))$, $(c(u_j), c(v))$, and $(c(u_i), c(v_i))$ all belong to $\A$, the monotonicity of $\A$ then enforces $c(u_{i+1}) < c(v) < c(u_i)$.  
    Moreover, since the substring $\paren(\G)(c(v_{i+1}) \dd c(v_i))$ is balanced (it represents the right portion of the light subtree of $v_i$) we must also have  $c(u_{i+1}) < o(v) < c(v) < c(u_i)$.
    This, however, yields a contradiction: monotonicity is violated between $(o(u_j), o(v))$ and $(i_\F, c(v_{i+1}))$, because $o(u_j) < i_\F$ while $o(v) > c(v_{i+1})$.  
    Therefore, the assumption that $i_\F < c(u_j) < c(u_i)$ cannot hold, and we conclude that $u_j$ is unaffected by the first step of $\mathsf{TreeAlign}$.

    Now, consider the second step of the procedure modifying $\A$.
    Let $(o(u_i), i_\G)$ be the rightmost pair in $\A$ with left index $o(u_i)$ and let $(i_\F, o(v_{i+1}))$ be the leftmost pair in $\A$ with right index $o(v_{i+1})$. The second step of the $\mathsf{TreeAlign}$ routine modifies all pairs between $(o(u_i), i_\G)$ and $(i_\F, o(v_{i+1}))$, exclusive (see Figure~\ref{fig:treealign} for example). 
    Since $(o(C[k+4]),o(C[k+3]))\in \A$ with $o(C[k+3])\ge o(v_{i+1})$, we have $i_\F \le o(C[k+4])$. Since $c(u_j) > o(C[k+4])$, this means that $c(u_j)$ is unaffected by the modifications of the first step of $\mathsf{TreeAlign}$.
    If $j < i$, then $o(u_j) < o(u_i)$, so $o(u_j)$ is also unaffected in this case.

    If $j > i$, on the other hand, then $o(u_j)$ could be affected if $o(u_{i}) < o(u_j) < i_\F$. 
    Since $(o(u_j),o(v)),(i_\F,o(v_{i+1}))\in \A$, the monotonicity of $\A$ implies that, in this case, $o(v) < o(v_{i+1})$. 
    At the same time, $u_j$ is a descendant of $u_i$, and thus $o(C[k+4]) < c(u_j) < c(C[k])$. 
    Since $o(C[k+4])\simeq_\A o(C[k+3])$, $c(u_j)\simeq_\A c(v)$, and $c(C[k])\simeq_\A c(C[k+1])$, the monotonicity of $\A$ implies $o(C[k+1]) < o(C[k+3]) < c(v) < c(C[k+1])$.
    By the non-crossing property of parentheses in $\str{\G}$, this means that $v$ is proper descendant of $v_{i}=C[k+1]$, i.e., $o(v_i) < o(v) < c(v) < c(v_i)$.
    Overall, we conclude that $o(v_i) < o(v) < o(v_{i+1})$. Since $\paren(\G)(o(v_i)\dd o(v_{i+1}))$ is balanced (it represents the left part of the light subtree of $v_i$), we must also have $o(v_i) < o(v)  < c(v) < o(v_{i+1})$, which contradicts $o(C[k+3]) < c(v)$ that we derived later.

    Overall, we conclude that $u_j$ is not modified by either step of $\mathsf{TreeAlign}(\A, u_i)$. 
\end{proof}

We now give an inductive argument to prove that all nodes of $\F$ are tree aligned. The main challenge of the proof comes from Case 1 of the $\mathsf{TreeAlign}(\A, u)$ routine, which affects the alignment of nodes beyond just $u$. The previous lemma already shows that no other node in the same heavy path as $u$ is affected, but we must also make sure any nodes not in this heavy path are also still tree-aligned after this routine. In the modifications of this case, we were careful to add deletions for all affected parentheses, so in fact even for nodes in other heavy paths, these modifications maintain that nodes of $\F$ are still tree-aligned. We discuss this formally as follows.

\begin{lemma}
\label{lem:tree-constraint}
    Given an alignment $\A$, $\mathsf{TreeAlign}(\A)$ is a tree alignment. 
\end{lemma}

\begin{proof}
    Let $\hvy= \{H_1, H_2, \ldots, H_\ell\}$ be the set of all maximal-length heavy paths in order from smallest node heavy depth to largest node heavy depth. We prove that all nodes in $\mathsf{TreeAlign}(\A)$ are tree-aligned by induction on $i$ where $i$ is the index of the current heavy path $H_i$ whose misaligned nodes are currently being processed by the $\mathsf{TreeAlign}$ routine.
    
    As the base case, we consider  $\mathsf{TreeAlign}(\A, u^1_i)$ for some $u^1_i$ in $H_1 = (u^1_1, u^1_2, \ldots, u^1_{m_1})$ where the nodes of $H_1$ are ordered from highest level to lowest level in $\F$ such that $u^1_i$ is the ancestor of all $u^1_j$ with $j > i$. We show that after performing $\mathsf{TreeAlign}(\A, u^1_i)$ for all $i$ from 1 to $m_1$, all nodes in $H_1$ are tree-aligned. First, if $u^1_i$ does not satisfy the conditions of case 1, $\A$ is modified to delete both $o(u^1_i)$ and $c(u^1_i)$.  Clearly, $u^1_i$ will be tree-aligned after this step and no other node alignment will be affected. Instead, if $u^1_i$ falls into case 1, then it is contained in a chain $C$ such that $C[k] = u^1_i$ and $\A$ is modified so that $u^1_i$ is tree-aligned to $v = C[k+1] \in \G$. By Lemma~\ref{lem:invariant}, all other nodes of $H$ that are already tree-aligned will remain tree-aligned, and so after all nodes of $H_1$ are passed to $\mathsf{TreeAlign}$, all nodes in $H_1$ will be tree-aligned.

    Now, we assume that all nodes in heavy paths up to path $H_{z-1}$ are tree-aligned, and we want to show that all nodes in heavy paths up to $H_z$ are tree-aligned after performing $\mathsf{TreeAlign}(\A, u^z_i)$ for all $i$ from 1 to $m_z$.  By the same argument as the base case, all misaligned nodes $u^z_i$ of $H_z$ will be tree-aligned after performing either case 1 or case 2 of $\mathsf{TreeAlign}(\A, u^z_i)$. We now must confirm that all nodes from $H_j$ for $j < z$ are also still tree-aligned after performing these steps. If $u^z_i$ falls into case 2 of the $\mathsf{TreeAlign}(\A, u^z_i)$ routine, it is deleted and no nodes of any prior paths are affected, so we only need to consider case 1.
    
    If $u^z_i$ falls into case 1, it is contained in a chain $C$ such that $C[k] = u^z_i$ and will be matched to $v = C[k+1] \in \G$. As in the base case, assume without loss of generality that $c(u^z_i) \simeq_\A c(v)$ and let $(i_\F, i_\G)$  denote the rightmost pair in $\A$ satisfying $i_\F \leq c(u^z_{i+1}), i_\G \leq c(v')$ where $v'$ is the heavy child of $v$. First, the $\mathsf{TreeAlign}$ routine replaces the sequence of pairs between $(i_\F, i_\G)$ and $(c(u^z_i) + 1, c(v) + 1)$ with a sequence of matches between the light subtrees $u^z_i$ and $v$ and then deletions of all other parentheses corresponding to indices of the removed pairs. 
    
    If $i_\F = c(u^z_{i+1})$, then only nodes $u' \in \F$ with $c(u^z_{i+1}) < o(u') < c(u') < c(u^z_i)$ are affected by these modifications. These are exactly the nodes in the light subtree of $u^z_i$ and are therefore matched and tree-aligned to corresponding nodes in the light subtree of $v$.

    If $i_\F < c(u^z_{i+1})$, then $i_G = c(v')$. In this case, there may exist some $u' \in \F$ that lie outside of the light subtree of $u^z_i$ and are affected by modifications to $\A$ $\mathsf{TreeAlign}(\A, u_i^z)$ as long as $u'$ is tree-aligned with a node in the light subtree of $v$ before modifications.  Note that for any such $u'$, $o(u^z_i) < o(u') < c(u') < c(u^z_i)$, and so $u'$ must either be in $H_z$ or a path $H_{z'}$ for $z' < z$. Since we have already dealt with nodes in $H_z$, we only concern ourselves with the latter case. By the induction hypothesis these nodes are already tree-aligned.  Therefore, both $o(u')$ and $c(u')$ must be deleted or aligned to parentheses in the light subtree of $v$ before $\mathsf{TreeAlign}$ modifies $\A$, and so, after modifications, these parentheses will both be deleted and $u'$ will still be tree-aligned.
        
    Now let $(o(u^z_i), i_\G)$ be the first pair in $\A$ with left index $o(u^z_i)$ and $(i_\F, o(v'))$ be the first pair in $\A$ with right index $o(v')$. Second, the $\mathsf{TreeAlign}$ routine replaces the sequence of pairs between $(o(u^z_i), i_\G)$ and $(i_\F, o(v'))$  with a matching of the parentheses corresponding to the light subtrees of $u^z_i$ and $v$ and then deletions of all other parentheses corresponding to removed pairs. By an analogous argument to the previous case, any node $u'$ affected by these modifications must be in a heavy path $H_{z'}$ for $z' < z$, and so by the induction hypothesis, after modifications, $u'$ is tree-aligned to a node in the light subtree of $v$ or both $o(u')$ and $o(v')$ are now deleted.
    
    By induction, all nodes in $\F$ will be tree-aligned in $\mathsf{TreeAlign}(\A)$.  Recall that $\A$ is modified to delete misaligned light nodes of $\F$, so trivially, there will be no misaligned light nodes in $\F$ in the final alignment. By Observation~\ref{obs:F-only}, all nodes in $\F$ and $\G$ are tree-aligned, and so $\mathsf{TreeAlign}(\A)$ is a tree alignment
\end{proof}

\subsection{Alignment Cost Analysis}
\label{subsec:costanalysis}

We now show that applying $\mathsf{TreeAlign}$ increases the number of edits in $\A$ by at most a $\Oh(\sqrt{n})$ multiplicative factor.  
The main challenge is to bound the additional deletions arising from case~1 of the $\mathsf{TreeAlign}$ routine.  
Every node that falls into this case lies in a chain, so we begin by analyzing the endpoints of chains.  

By definition, chains contain nodes along single branches in each forest, so they are finite, and the last node of every chain is a partially deleted node (Corollary~\ref{cor:chain-ending}).  
Consequently, we can charge the deletions introduced by $\mathsf{TreeAlign}$ to the deletions already present at the chain endpoints.  
This idea extends naturally to sequences that include multi-branch misaligned nodes as well.  
We formalize these sequences next using extended chains.

\begin{definition}[Extended chain]
\label{def:extended}
Given an alignment $\A: \paren(\F) \onto \paren(\G)$, an \emph{extended chain} $\C$ is an inclusion-wise maximal sequence of at least three nodes alternating between $\F$ and $\G$, such that one of the following two conditions holds:
\begin{enumerate}
    \item For every $i \in [0 \dd |\C|)$ with $\C[i] \in \F$, we have $o(\C[i]) \simeq_\A o(\C[i+1])$ if $i < |\C|-1$ and $c(\C[i]) \simeq_\A c(\C[i-1])$ if $i > 0$.
    \item For every $i \in [0 \dd |\C|)$ with $\C[i] \in \F$, we have $c(\C[i]) \simeq_\A c(\C[i+1])$ if $i < |\C|-1$ and $o(\C[i]) \simeq_\A o(\C[i-1])$ if $i > 0$.
\end{enumerate}
\end{definition}

Note the difference between the above definition and Definition~\ref{def:chain} is that we no longer require $C[i]$ to be an ancestor of $C[i+2]$, so the misaligned nodes of an extended chain do not need to be on a single branch of $\F$ or $\G$. This allows us to follow maximal-length sequences of aligned parentheses corresponding to misaligned nodes. Similarly to chains, we show that every extended chain ends in a partially deleted node. We now prove this formally for chains and then extended chains.

\begin{lemma}\label{lem:chain-cont}
Consider a chain $C$ with respect to an alignment $\A: \paren(\F) \onto \paren(\G)$.  
Let $u = C[i]$ for some $i \in [1 \dd |C|)$.  
If $\A$ matches $o(u)$ with $o(v)$ or $c(u)$ with $o(v)$ for some node $v$, then $v = C[i-1]$ or $v = C[i+1]$ (with $i \ne |C|-1$ in the latter case).
\end{lemma}

\begin{proof}
By symmetry between $\F$ and $\G$, assume without loss of generality that $u \in \F$ and $v \in \G$.  
Moreover, by symmetry with respect to transposition, assume $C$ is an opening chain.  
Then $c(u) \simeq_\A c(C[i-1])$, so $v = C[i-1]$ is the unique node such that $\A$ matches $c(u)$ with $c(v)$.  
If $i \ne |C|-1$, then $o(u) \simeq_\A o(C[i+1])$, and $v = C[i+1]$ is the unique node such that $\A$ matches $o(u)$ with $o(v)$.  

It remains to consider the case $o(u) \simeq_\A o(v)$ with $i = |C|-1$.  
By Definition~\ref{def:chain}, $|C| \ge 3$, so nodes $v' \coloneqq C[i-1]$ and $u' \coloneqq C[i-2]$ exist.  
Since $u'$ is an ancestor of $u$, $c(u) \simeq_\A c(v')$, and $o(u') \simeq_\A o(v')$, we have $o(u') < o(u) < c(u)$.  
Monotonicity of the matching implies $o(v') < o(v) < c(v')$.  
By the non-crossing property of parentheses in $\G$, we further get $o(v') < o(v) < c(v) < c(v')$, i.e., $v'$ is a proper ancestor of $v$.  

Since $o(u) \simeq_\A o(v)$ and $C$ is opening, the node $v$ could be used to extend $C$, contradicting its maximality.  
Hence, the claim follows.
\end{proof}

\begin{corollary}
\label{cor:chain-ending}
Given an alignment $\A: \paren(\F) \onto \paren(\G)$ and a chain $C$ in $\A$, the last node $C[|C|-1]$ is a partially deleted node.
\end{corollary}
\begin{proof}
    By Lemma~\ref{lem:chain-cont}, the last node $u = C[|C|-1]$ cannot have both of its parentheses aligned to any other node, so at least one of $o(u)$ or $c(u)$ is deleted by $\A$.  
    Hence, $u$ is a partially deleted node.  
    It remains to argue that $C$ is finite.  
    Suppose, for contradiction, that $C$ is infinite. Then some node appears twice in $C$, say $C[i] = C[j]\in \F$ with $i < j$.  
    But this violates the property from Definition~\ref{def:chain} that $C[i]$ is a (proper) ancestor of $C[j]$.
    Therefore, every chain is finite and must end at a partially deleted node.
\end{proof}

Now, we turn our attention back to extended chains. Every extended chain has a very predictable structure; each extended chain contains at most 1 multi-branch misaligned node and at most 2 chains. We now prove this formally.

\begin{lemma}
\label{lem:extended-structure}
Every extended chain $\C$ with respect to an alignment $\A: \paren(\F) \onto \paren(\G)$ consists of at most one multi-branch misaligned node and at most two chains.
\end{lemma}

\begin{proof}
Let $\C$ be an extended chain containing a multi-branch misaligned node $u_1 = \C[i]$ for some $i\in [2\dd |C|-2)$.  
Without loss of generality, assume $u_1 \in \F$, and let $v_1 = \C[i-1] \in \G$ and $v_2 = \C[i+1] \in \G$.  
Moreover, by symmetry (up to reversal of $\C$) assume without loss of generality that \cref{def:extended} is satisfied due to $o(u_1) \simeq_\A o(v_1)$ and $c(u_1) \simeq_\A c(v_2)$.  
Since $u_1$ is multi-branch misaligned, we have $o(v_1) < c(v_1) < o(v_2) < c(v_2)$.  

Consider the next node $u_2 = \C[i+2]$.  
By the alignment, $o(u_2) \simeq_\A o(v_2)$.  
Furthermore, because $o(u_1)\simeq_\A o(v_1)$ and $c(u_1)\simeq_\A c(v_2)$, and $o(v_1) < o(v_2) < c(v_2)$, we have $o(u_1) < o(u_2) < c(u_1)$.  
Hence, by Definition~\ref{def:chain}, $u_1, v_2, u_2$ form a chain.  

By Lemma~\ref{lem:chain-cont}, any node $\C[j]$ with $j\in (i\dd |\C|)$ also belongs to this chain, and $\C[|\C|-1]$ must be a partially deleted node.  
Symmetrically, the nodes $\C[i], \C[i-1], \ldots, \C[0]$ form a chain, with $\C[1]$ being a partially deleted node.  
This shows that $\C$ contains at most one multi-branch misaligned node and at most two chains.
\end{proof}

Since we know that extended chains only contain a single multi-branch misaligned node and can only continue as long as the chains contained in them, we know that eventually they must end in a partial deletion as chains do.
\begin{corollary}
\label{cor:extended-end}
    Given an alignment $\A: \paren(\F) \onto \paren(\G)$ and an extended chain $\C$ in $\A$, then $\C[|\C| - 1]$ is a partially-deleted node.
\end{corollary}

\begin{proof}
    By Lemma~\ref{lem:extended-structure}, every extended chain has at most 2 chains and at most 1 multi-branch misaligned node. By Corollary~\ref{cor:chain-ending}, chains must end in a partially deleted node.  Thus, extended chains are finite, and by Definition~\ref{def:extended}, the last node will have either $o(\C[|\C| - 1])$ or $c(\C[|\C| - 1])$ deleted.
\end{proof}

We now establish a few additional properties of chains.  
First, we observe that chains are well separated from tree-aligned nodes.

\begin{fact}\label{fct:non-aligned}
Consider a chain $C$ with respect to an alignment $\A: \paren(\F) \onto \paren(\G)$, and let $C[i], C[i+2] \in \F$.  
There is no tree-aligned node $u$ on the path between $C[i]$ and $C[i+2]$.
\end{fact}
\begin{proof}
Suppose, for contradiction, that there exists a tree-aligned node $u$ on the path from $C[i]$ to $C[i+2]$, and let $v \in \G$ be the node to which $u$ is tree-aligned, so that $o(u)\simeq_\A o(v)$ and $c(u)\simeq_\A c(v)$.  
Because $u$ lies strictly between $C[i]$ and $C[i+2]$ in $\F$, its parentheses interval is properly nested between those of $C[i]$ and $C[i+2]$, that is, $o(C[i]) < o(u) < o(C[i+2]) < c(C[i+2]) < c(u) < c(C[i])$.

By symmetry, it suffices to consider the case where $C$ is an opening chain; the closing case is symmetric. 
Then, by definition of an opening chain, we have $o(C[i])\simeq_\A o(C[i+1])$ and $c(C[{i+2}])\simeq_\A c(C[i+1])$.  
Combining these with $o(C[i]) < o(u) < c(C[i+2]) < c(u)$, the monotonicity of~$\A$ implies $o(C[i+1]) < o(v) < c(C[i+1]) < c(v)$.  
This configuration would force the intervals corresponding to $C[i+1]$ and $v$ to cross in $\G$, violating the non-crossing property of parentheses in $\str{\G}$.  
Hence, no tree-aligned node $u$ can lie on the path between $C[i]$ and $C[i+2]$.
\end{proof}

Next, we show that two chains sharing the same heavy path exhibit a highly constrained structure.  
Specifically, if $C$ and $C'$ are two chains that lie on the same heavy path, then between any two consecutive nodes of $C$ there can be at most one node of $C'$, and symmetrically, between any two consecutive nodes of $C'$ there can be at most one node of $C$.  
In particular, in the region where $C$ and $C'$ overlap, their nodes must alternate along the path.  
Figure~\ref{fig:overlapping-chains} illustrates this phenomenon: on the left, two chains interleave in the only possible monotone configuration, while on the right, attempting to modify the alignment of only one chain leads to a violation of monotonicity.

\begin{figure}[th]
    \centering
    \begin{tikzpicture}
{\fontfamily{pcr}\selectfont 
    \node(F1) at (0, 0) {\textcolor{red}{(}};
    \node (F2) at (.3, 0) {[};
    \node (F3) at (.6, 0) {\textcolor{blue}{(}};
    \node (F4) at (.9, 0) {[};
    \node (F5) at (1.2, 0) {\textcolor{blue}{(}};
    \node (F6) at (1.7, 0) {...};
    \node (F7) at (2.2, 0) {\textcolor{blue}{)}};
    \node (F8) at (2.5, 0) {]};
    \node (F9) at (2.8, 0) {\textcolor{blue}{)}};
    \node (G1) at (0, -1) {[};
    \node (G2) at (.3, -1) {\textcolor{blue}{(}};
    \node (G3) at (.6, -1) {[};
    \node (G4) at (.9, -1) {\textcolor{blue}{(}};
    \node (G5) at (1.4, -1) {...};
    \node (G6) at (1.9, -1) {\textcolor{blue}{)}};
    \node (G7) at (2.2, -1) {]};
    \node (G8) at (2.5, -1) {\textcolor{blue}{)}};
    \node (G9) at (2.8, -1) {\textcolor{red}{]}};
    \node (F) at (-.8, 0) {$\str{\F}$:};
    \node (G) at (-.8, -1) {$\str{\G}$:};
    }
    \draw (F2) -- (G1);
    \draw (F3) -- (G2);
    \draw (F4) -- (G3);
    \draw (F5) -- (G4);
    \draw (F7) -- (G6);
    \draw (F8) -- (G7);
    \draw (F9) -- (G8);
\end{tikzpicture}\hspace{1cm}
\begin{tikzpicture}
{\fontfamily{pcr}\selectfont 
    \node(F1) at (0, 0) {(};
    \node (F2) at (.3, 0) {[};
    \node (F3) at (.6, 0) {\textcolor{blue}{(}};
    \node (F4) at (.9, 0) {[};
    \node (F5) at (1.2, 0) {\textcolor{blue}{(}};
    \node (F6) at (1.7, 0) {...};
    \node (F7) at (2.2, 0) {\textcolor{blue}{)}};
    \node (F8) at (2.5, 0) {]};
    \node (F9) at (2.8, 0) {\textcolor{blue}{)}};
    \node (G1) at (0, -1) {[};
    \node (G2) at (.3, -1) {\textcolor{blue}{(}};
    \node (G3) at (.6, -1) {[};
    \node (G4) at (.9, -1) {\textcolor{blue}{(}};
    \node (G5) at (1.4, -1) {...};
    \node (G6) at (1.9, -1) {\textcolor{blue}{)}};
    \node (G7) at (2.2, -1) {]};
    \node (G8) at (2.5, -1) {\textcolor{blue}{)}};
    \node (G9) at (2.8, -1) {]};
    \node (F) at (-.8, 0) {$\str{\F}$:};
    \node (G) at (-.8, -1) {$\str{\G}$:};
    }
    \draw [red] (F2) -- (G3);
    \draw (F3) -- (G2);
    \draw [red] (F4) -- (G5);
    \draw (F5) -- (G4);
    \draw [red] (F6) -- (G7);
    \draw (F7) -- (G6);
    \draw [red] (F8) -- (G9);
    \draw (F9) -- (G8);
\end{tikzpicture}
    \caption{(\textbf{Left}) Two overlapping chains $C$ and $C'$, where $C$ includes the blue parentheses and $C'$ includes the black brackets.  
    This is the only possible alignment structure for two chains when $C'[i']$ lies on the path from $C[i]$ to $C[i+2]$.  
    (\textbf{Right}) An example of two overlapping chains whose structure violates monotonicity.}
    \label{fig:overlapping-chains}
\end{figure}

\begin{lemma}
\label{lem:chain-overlap}
Consider an alignment $\A: \paren(\F) \onto \paren(\G)$, and let $C[i]$ and $C'[i']$ be nodes lying on the same heavy path of $\F$ for some $0 \le i \le |C| - 3$ and $0 \le i' \le |C'| - 3$.
Then, $C[i]$ is an ancestor of $C'[i']$ if and only if $C[i+2]$ is an ancestor of $C'[i'+2]$.
\end{lemma}
\begin{proof}
    By \cref{obs:chain-heavy}, all four nodes lie on the same heavy path.
    Further, \cref{def:chain} implies that $C[i+2]$ is a descendant of $C[i]$ and $C'[i'+2]$ is a descendant of $C'[i']$.
    Assume for contradiction that $C[i+2]$ is an ancestor of $C'[i'+2]$, but $C[i]$ is not an ancestor of $C'[i']$.  
    In this case, the four nodes must appear along the path in the order  
    $C'[i']$, $C[i]$, $C[i+2]$, $C'[i'+2]$.

    By symmetry, assume without loss of generality that $C$ is an opening chain.  
    Then, $o(C[i]) \simeq_\A o(C[i+1])$ and $c(C[i+2]) \simeq_\A c(C[i+1])$.  
    If $C'$ is also an opening chain, then $o(C'[i']) \simeq_\A o(C'[i'+1])$ and $c(C'[i'+2]) \simeq_\A c(C'[i'+1])$.  
    Because $o(C'[i']) < o(C[i]) < c(C'[i'+2]) \le c(C[i+2])$, the monotonicity of $\A$ implies  
    $o(C'[i'+1]) < o(C[i+1]) < c(C'[i'+1]) \le c(C[i+1])$,  
    which violates the non-crossing property of the parentheses in $\str{\G}$.

    If instead $C'$ is a closing chain, then $c(C'[i']) \simeq_\A c(C'[i'+1])$ and $o(C'[i'+2]) \simeq_\A o(C'[i'+1])$.  
    Since $o(C[i]) < o(C'[i'+2]) < c(C[i+2]) < c(C'[i'])$, monotonicity implies  
    $o(C[i+1]) < o(C'[i'+1]) < c(C[i+1]) < c(C'[i'+1])$,  
    which again contradicts the non-crossing property of $\str{\G}$.

    Hence, the assumption leads to a contradiction, proving that if $C[i+2]$ is an ancestor of $C'[i'+2]$, then $C[i]$ must also be an ancestor of $C'[i']$.  
    By symmetry, exchanging the roles of $C$ and $C'$ gives the converse.  
    Thus, $C[i]$ is an ancestor of $C'[i']$ if and only if $C[i+2]$ is an ancestor of $C'[i'+2]$.
\end{proof}

With the above lemma, we can now analyze the case when two nodes $C[i]$ and $C'[i']$ on the same heavy path fall into distinct branches of the $\mathsf{TreeAlign}$ procedure due to $w(C[i],C[i+4]) < \sqrt{n} \le w(C'[i'],C'[i'+4])$.

\begin{lemma}
\label{lem:big-width-misaligned-strong}
Consider chains $C$ and $C'$ with respect to an alignment $\A: \paren(\F) \onto \paren(\G)$.  
Suppose that for some $i \in [0\dd |C|-5]$ and $i' \in [0\dd |C'|-5]$, the nodes $C[i], C'[i'], C[i+4], C'[i'+4]$ appear in this order along a single heavy path of $\F$.  
If $w(C[i],C[i+4]) < w(C'[i'],C'[i'+4])$, then there exists a node $v$ on the path between $C[i+4]$ and $C'[i'+4]$ such that:
\begin{enumerate}
    \item $v$ is either deleted, partially deleted, multi-branch misaligned, or among the first two nodes of $\F$ in the chain in belongs to, and
    \item $w(C'[i'],v) \le w(C[i],C[i+4])$.
\end{enumerate}
\end{lemma}
\begin{proof}
    Let $v$ be the highest node on the path from $C[i+4]$ to $C'[i'+4]$ that is not of the form $C_v[i_v+4]$ for any chain $C_v$ and index $i_v \in [0\dd |C_v|-5]$; if no such node exists, set $v = C'[i'+4]$.

    By construction, every node $u$ on the path from $C[i+4]$ (inclusive) to $v$ (exclusive) must be of the form $C_u[i_u+4]$ for some chain $C_u$ and index $i_u$.  
    By \cref{lem:chain-overlap}, the corresponding nodes $C_u[i_u]$ appear in the same relative order along the path from $C[i]$ (inclusive) to $C'[i']$ (exclusive).  
    Moreover, for each such $u$, we have $\F'(C_u[i_u]) = \F'(C_u[i_u+4])$, implying that the light subtrees rooted between $C[i]$ and $C'[i']$ are in total at least as large as the light subtrees rooted between $C[i+4]$ and $v$.
    Thus, $w(C[i+4],v) \le w(C[i],C'[i'])$, and consequently,
    \[
    w(C'[i'],v) = w(C'[i'],C[i+4]) + w(C[i+4],v)
        \le  w(C'[i'],C[i+4]) + w(C[i],C'[i'])
        = w(C[i],C[i+4]).
    \]
    Since $w(C'[i'],C'[i'+4])>w(C[i],C[i+4])$, it follows that $v \neq C'[i'+4]$.

    If $v$ is deleted, partially deleted, or multi-branch misaligned, or if $v$ is among the first two nodes of $\F$ in some chain, the lemma holds immediately.  
    By \cref{fct:non-aligned}, $v$ cannot be tree-aligned.  
    The only remaining possibility is that $v$ itself is of the form $C_v[i_v+4]$ for some chain $C_v$ and index $i_v \in [0\dd |C_v|-5]$, contradicting the definition of $v$.  
    Hence, such a node $v$ as described must exist.
\end{proof}
Finally, we are ready to prove the main lemma of this section.  
Let $\A: \paren(\F) \to \paren(\G)$ be an optimal alignment, and let $\A'$ denote $\mathsf{TreeAlign}(\A)$.  
Let $D_\A$ be the deleted parentheses of $\A$, and let $D_{\A'}$ be the deleted parentheses produced by $\mathsf{TreeAlign}(\A, u)$.  
For a misaligned node $u$, let $\C(u)$ denote the extended chain containing $u$.  
For an extended chain $\C(u)$, let $d_{\C(u)}$ denote the deleted parenthesis corresponding to $\C[|\C|-1]$.  
We define a mapping $\gamma: D_{\A'} \onto D_\A$ and show that the preimage $\gamma^{-1}(d)$ of any $d \in D_\A$ has size $\Oh(\sqrt{n})$.

To construct $\gamma$, we examine the deletions caused by $\mathsf{TreeAlign}(\A, u)$ for each heavy node $u$ in $\F$.  
If $u$ is tree-aligned, no deletions occur, so we only consider misaligned nodes.  
Suppose $u$ lies on a heavy path $H$, and let $\C$ be the extended chain containing $u$.

\paragraph{Case 1:}
If $u$ falls into Case~1, then $u = C[k]$ for some chain $C$ and integer $k \in [0\dd |C|-5]$, with $w(C[k], C[k+4]) < \sqrt{n}$.  
The routine $\mathsf{TreeAlign}$ then modifies $\A$ so that the nodes in the subtree $\F'(u)$ are aligned to those in $\G'(v)$, where $v = C[k+1]$, deleting any parentheses whose alignments would violate monotonicity after this modification.  
Let $\paren(\F)[a]$ be one of the parentheses deleted in this step, and let $u_a \in H$ denote the heavy node in $H$ whose light subtree contains the node corresponding to $\paren(\F)[a]$.  
By the characterization from \cref{subsec:treeconstraints}, the position $a$ corresponds to a node in $\F(C[k]) \setminus \F(C[k+4])$, and hence $u_a$ lies on the path from $C[k]$ to $C[k+4]$.

To define $\gamma$ for such $a$, we further subdivide this case to ensure that no more than $\Oh(\sqrt{n})$ deletions of $\A'$ map to a single deletion of $\A$.

Before discussing the subcases, note that if $u_a$ is a node $C_a[k_a]$ in some chain $C_a$ with $k_a \in [0\dd |C_a|-5]$ and $w(C_a[k_a], C_a[k_a+4]) < \sqrt{n}$, then $u_a$ itself falls into Case~1 of $\mathsf{TreeAlign}$.  
In that case, the light subtree of $u_a$ (including $\paren(\F)[a]$) will later be matched perfectly by $\mathsf{TreeAlign}(\A, u_a)$.  
Therefore, we do not define $\gamma(a)$ at this point; if $a \in D_{\A'}$ is ultimately deleted, then we will instead define $\gamma(a)$ as we analyze the final call to $\mathsf{TreeAlign}(\A, u')$ that deletes $\paren(\F)[a]$.

\begin{enumerate}[label=(\alph*)]
    \item If $\paren(\F)[a]$ is already deleted in $\A$, set $\gamma(a) = a$.

    \item If $u_a$ is partially deleted or multi-branch misaligned, set $\gamma(a) = d_{\C(u_a)}$.

    \item If $u_a = C_a[k_a]$ for some chain $C_a$ and $k_a \ge |C_a| - 4$, set $\gamma(a) = d_{\C(u_a)}$.

    \item If $u_a = C_a[k_a]$ for some chain $C_a$ and $w(C_a[k_a], C_a[k_a+4]) \ge \sqrt{n}$, then we apply \cref{lem:big-width-misaligned-strong} to obtain a node $u'$ on the path from $C[k+4]$ to $C_a[k_a+4]$ such that $u'$ is deleted, partially deleted, multi-branch misaligned, or among the first two nodes of $\F$ in its chain.  
    Moreover, $u'$ satisfies $w(C_a[k_a], u') = w(C[k], C[k+4]) < \sqrt{n}$.  
    We then set $\gamma(a) = o(u')\in D_\A$ if $u'$ is deleted, and $\gamma(a) = d_{\C(u')}$ otherwise.
\end{enumerate}

\paragraph{Case 2:}
If $u$ falls into Case 2, then $u$ is either a partially deleted node, a multi-branch misaligned node, or a single-branch misaligned node whose chain $C$ does not satisfy the width condition required for Case 1 in $\mathsf{TreeAlign}$.  
In this case, we set $\gamma(o(u)) = \gamma(c(u)) = d_{\C(u)}$.
We now proceed to analyze the above mapping in the proof of the main lemma of this section.

\begin{lemma}
\label{lem:underapprox}
 For all forests $\F$ and $\G$ of size at most $n$, we have \[\ted(\F, \G) \leq \Oh(\sqrt{n}) \cdot \ed(\paren(\F), \paren(\G)).\]
\end{lemma}

\begin{proof}
We consider each of the cases in the definition of the mapping $\gamma$ and show that, for any $d \in D_{\A}$, the number of parentheses $a$ such that $\gamma(a) = d$ is at most $\Oh(\sqrt{n})$.

\vspace{1mm}
\noindent\textbf{Case 1(a):}  
Here $a = d$, so trivially only a single element is mapped to $d$ from this case. 
\vspace{1mm}

\noindent\textbf{Case 1(b):}  
If $a$ is mapped to $d$ in this case due to a partially deleted or multi-branch misaligned node $u_a$, then $d$ must belong to the same extended chain as $u_a$.  
By Lemma~\ref{lem:extended-structure}, each extended chain contains at most one multi-branch misaligned node and at most two partially deleted nodes (by Definition~\ref{def:extended}).  
Since $w(C[k], C[k+4]) < \sqrt{n}$, fewer than $\sqrt{n}$ parentheses lie in the light subtree of $u_a$, and hence at most $\Oh(\sqrt{n})$ parentheses $a$ map to $d$ in this case.
\vspace{1mm}

\noindent\textbf{Case 1(c):}  
Here $u_a$ is one of the last two nodes of $C_a$ in $\F$.  
As in Case~1(b), since $w(C[k], C[k+4]) < \sqrt{n}$, at most $\Oh(\sqrt{n})$ deletions are mapped to $d_{\C(u_a)}$ via $u_a$.  
By Lemma~\ref{lem:extended-structure}, there are only $\Oh(1)$ such nodes per extended chain, and hence $\Oh(\sqrt{n})$ total deletions mapped to $d_{\C(u_a)}$.
\vspace{1mm}

\noindent\textbf{Case 1(d):}  
In this case, there exists a deleted or misaligned node $u'$ on the path from $C_a[k_a]$ to $C_a[k_a+4]$ such that $w(C_a[k_a], u') < \sqrt{n}$.  
There are fewer than $\sqrt{n}$ deletions attributed to $u'$, as they all correspond to nodes in the light subtree $\F'(C_a[k_a])$, which lies entirely within the subtree $\F(C_a[k_a]) \setminus \F(u')$ of size less than $\frac{1}{2}\sqrt{n}$.  
If $u'$ is deleted, these deletions are mapped directly to $u'$.  
Otherwise, $u'$ is partially deleted, multi-branch misaligned, or among the first two nodes in $\F$ of its chain.  
By Lemma~\ref{lem:extended-structure}, there are $\Oh(1)$ such nodes $u'$ per extended chain $\C(u')$, and each contributes at most $\Oh(\sqrt{n})$ deletions mapped to $d_{\C(u')}$.
\vspace{1mm}

\noindent\textbf{Case 2:}  
For $u$ in Case~2, $u$ must either be partially deleted, multi-branch misaligned, or a single-branch misaligned node whose chain $C$ does not satisfy the requirements of Case~1.  
The deletions of $o(u)$ and $c(u)$ are assigned to $d_{\C(u)}$.  
By Definition~\ref{def:extended} and Lemma~\ref{lem:extended-structure}, only $\Oh(1)$ multi-branch or partially deleted nodes can be mapped to $d_{\C}$ from this case.  
If $u$ is the $k$th node in a chain $C$, then either $k \ge |C| - 4$, $w(C[k], C[k+2]) \ge \sqrt{n}$, or $w(C[k+2], C[k+4]) \ge \sqrt{n}$.  
The first case implies that $u$ is the last or second-to-last node of $\F$ in $C$, so it can occur at most twice per chain overlapping with $\C$.  
In the remaining two cases, since each occurrence of $w(\cdot,\cdot) \ge \sqrt{n}$ corresponds to a distinct $\sqrt{n}$-sized interval, it can occur only $\Oh(n / \sqrt{n}) = \Oh(\sqrt{n})$ times per chain.  
Thus, at most $\Oh(\sqrt{n})$ nodes are mapped to $d_{\C(u)}$ from Case~2.

\vspace{1mm}
In all cases, at most $\Oh(\sqrt{n})$ parentheses are mapped to any $d \in D_{\A}$.  
Hence $|D_{\A'}| \le \Oh(\sqrt{n}) \cdot |D_{\A}|$, and therefore
\[
    \ted(\F, \G) \le |D_{\A'}| \le \Oh(\sqrt{n}) \cdot |D_{\A}| = \Oh(\sqrt{n}) \cdot \ed(\paren(\F), \paren(\G)).\qedhere
\]
\end{proof}

\subsection{Approximate Dynamic Tree Edit Distance}
\label{subsec:dynTreeSec}
We now show that our novel tree to string reduction can be maintained dynamically and used in conjunction with the dynamic string edit distance algorithm of \cite{KMS2023} to obtain a dynamic approximation for tree edit distance with sub-polynomial update time. Our reduction relies on maintaining a heavy-light decomposition on the input trees as dynamic edits occur on the tree. Fortunately, we have already shown how to dynamically maintain our heavy-light decomposition in our dynamic Dyck edit distance algorithm; observe that the relationship between input trees and their parentheses representation is precisely the same relationship between a parenthesis string and its reduction tree from Section~\ref{sec:fastapprox}. The only difference is that now dynamic edits come directly to the tree, which correspond to two edits on the parentheses string. With this in mind, we can directly use our earlier result that states that we can maintain the heavy-light paths of our decomposition in poly-logarithmic time.  The main challenge of this section then is showing how we update the subtree embeddings in the labels of our reduction strings according to Definition~\ref{def:parenF}. Recall that in our reduction strings, a node's label contains the subtree of all its light children. For any dynamic insertion or deletion, we therefore need to reflect the dynamic edit in the labels of $\Oh(\log n)$ ancestor nodes, one per light node. Additionally, for any nodes that change from heavy to light or light to heavy we need to adjust its parents label accordingly.

In order to reflect the dynamic edit in light node ancestor labels, we first need a way to compute the set of light node ancestors of the dynamic edit.  We show that using a binary tree built on top of the parenthesis representation, we can do so.

\begin{lemma}
\label{lem:lightancestorset}
    For a forest $\F$, a data structure can be dynamically maintained to compute the light node ancestor set of node $v$ for all $v \in \F$ in time $\Oh(\log^2 |\F|)$ with update time $\Oh(\log^2 |\F|)$ when a node of $\F$ changes from heavy to light or light to heavy.  The data structure can be initialized in time $\tOh(|\F|)$.
\end{lemma}

\begin{proof}
    We begin by computing a heavy-light decomposition of $\F$. Recall that by Observation~\ref{obs:paren-anc}, ancestors of node $v$ are precisely those nodes $u$ with $o(u) < o(v) < c(v) < c(u)$. So, we show how to find the set of light nodes satisfying these conditions. To do so, we utilize a binary tree $B$ built on top of the parenthesis representation of $\F$. 

    The leaves of $B$ partition $\str{\F}$ into individual characters so the $i$th leaf corresponds to the $i$th character in $\str{\F}$ and each internal node $v$ corresponds to the concatenation of the strings in its children, which we denote as $S_v$.  Each node $v$ of $B$ stores the number of unmatched light node opening parentheses $U^O_v$ and unmatched light node closing parentheses $U^C_v$ in $S_v$. Note here that parenthesis $o(u)$ ($c(u)$) for a node $u$ is unmatched in a substring if $c(u)$ ($o(u)$) is not contained in that substring. Observe that given node $v$ with left child $u$ and right child $w$, $U^O_v = \max\{U^O_u - U^C_w, 0\} + U^O_w$ and $U^C_v = \max\{U^C_w - U^O_u, 0\} + U^C_u$. It is easy to see that constructing $B$ bottom-up can be done in $\Oh(|\F|)$-time with this fact.  Furthermore, if a single parentheses is deleted from $\F$ or changed from light to heavy or heavy to light, updating the root-to-leaf path of the updated parenthesis can be done in time logarithmic of the height of tree, which can be maintained using a self-balancing tree, e.g., an AVL tree, while only affecting an additional factor of $\Oh(\log |\F|)$ nodes per update. Therefore, when a node $u$ changes from heavy to light, we can rebuild the root-to-leaf path to the leaves corresponding to $o(u)$ and $c(u)$ in $\Oh(\log^2 |\F|)$ time.

    Now, we discuss how to use $B$ to find the set of light node ancestors of a node $v \in \F$. We start at the leaf corresponding to $o(v)$ and traverse upwards through parent-child links. As we traverse upwards, if we are at a node $u$, we will keep track of the number of unmatched closing parentheses to the left of $o(v)$ in $S_u$, which we denote as $L_u$. If we are currently at node $w \in B$ moving to parent node $u$, we may compute $L_u$ from $L_w$ as follows.  If $w$ is the left child of $u$ and $w'$ is the right child of $u$, then $L_u = L_w$.  If $w$ is the right child and $w'$ is the left child of $u$, $L_u = \max\{L_w - U^O_{w'}, 0\}$.

    We use the above values $L_u$ to determine when there is a new unmatched light node parenthesis $o(v')$ in $S_u$ to the left of $o(v)$, which implies that $o(v') < o(v) < c(v) < c(v')$, i.e., $v'$ is a light node ancestor of $v$. If $w$, the previous node we moved to $u$ from is the left child of $u$, we know that there is no new such parenthesis $o(v')$ to the left of $u$ in $S_u$ that was not already present in $S_w$, so we do nothing.  However, if $w$ is the right child of $u$ and $w'$ denotes the left child of $u$, we know the number of new light node parentheses to the left of $o(v)$ in $S_u$ is $U^O_{w'} - L_w$.  If this value is positive, we mark the left child $w'$ and continue traversing to the root. 
    
    For all marked nodes $u$, we traverse back down the tree from $u$ to find the indices of the light node ancestors of $v$. Note that all parentheses in strings corresponding to marked nodes lie to the left of $o(v)$ by our marking algorithm. Therefore, we just need to determine the exact leaf nodes corresponding to opening parentheses whose match is to the right of $o(v)$. Let $w$ be the left child of $u$ and let $w'$ be the right child of $u$.  If $U^O_{w'}$ is greater than 0, we traverse right. If $U^O_{w} - U^C_{w'} > 0$, we traverse left. In either case, we know that there is an unmatched opening parenthesis corresponding to a light node ancestor of $v$. These two cases are not exclusive, we may do one or both traversals.  When we reach the leaf nodes, we may take the index stored there and add its corresponding node in $\F$ to the set of light node ancestors.  A single ancestor node requires traversing up and down the tree once each where the height of the tree is $\Oh(\log |\F|)$, and since there are $\Oh(\log |\F|)$ such light node ancestors, we may compute the entire set in time $\Oh(\log^2 |\F|)$.
\end{proof}

To dynamically maintain our subtree embeddings efficiently, we utilize a dynamic strings data structure of \cite{DBLP:conf/soda/GawrychowskiKKL18}, which allows us to store and perform splitting and merging operations on a collection of strings. Additionally, each string in the structure is assigned a unique integer identifier (see \cite{DBLP:conf/soda/GawrychowskiKKL18} for details on how this is done).

We now state some necessary definitions from \cite{DBLP:conf/soda/GawrychowskiKKL18}. As mentioned, we will dynamically maintain this data structure on a collection of strings $\X$ that will be initially created during a pre-processing step. We define the following operations for the data structure:

\begin{theorem}[\cite{DBLP:conf/soda/GawrychowskiKKL18} Thm. 6.12]\label{thm:dynstr}
    A collection $\X$ of persistent strings of total length $n$ can be dynamically maintained subject to the following operations:
    \begin{description}
        \item[$\add(X)$] for $X \in \Sigma^+$ results in $\X = \X \cup \{X\}$ and takes $\Oh(\log n + |X|)$ time.
        \item[$\con(X_1, X_2)$] for $X_1, X_2 \in \X$ results in $\X := \X \cup \{X_1\cdot X_2\}$ and takes $\Oh(\log n)$ time.
        \item[$\spl(X, k)$] for $X \in \X, k \in [|X|]$ results in $\X := \X \cup \{X[0\dd k), X[k \dd |X|)\}$ and takes $\Oh(\log n)$ time.
        \item[$\lcp(X_1,X_2)$] for $X_1,X_2\in \X$ computes the length of longest common prefix of $X_1$ and $X_2$  and takes $\Oh(1)$ time.
    \end{description}
    The algorithms are Las-Vegas randomized and the running time bounds hold with high probability.
\end{theorem}

We now discuss how to maintain our label embeddings formally in the proof of the main theorem of this section.

\begin{restatable}{theorem}{DynamicTree}
\label{thm:dyn-tree}
   Given integers $2\le b \le n$, there exists a randomized dynamic algorithm, that maintains an $\Oh(b\log_b n \log^2n \cdot \sqrt{n})$-approximation of $\ted(\F, \G)$ (correctly with high probability against an oblivious adversary) for (initially empty) forests $\F, \G$ of size at most $n$ undergoing edits. The expected update time of the algorithm is $\Oh(b^2 (\log n)^{\Oh(\log_b n)})$ per edit.
\end{restatable}

\begin{proof}
    Let $n = |\F|$. Without loss of generality let the dynamic update occur in $\F$.  We will utilize a dynamic stings data structure instance $\X$ from Theorem~\ref{thm:dynstr} to store the parentheses representations of light subtrees for all nodes in $\F$ and $\G$, and we use the unique integer identifier assigned to each string in $\X$ as the label of each node in $\paren(\F)$.
    
    First, we show how to maintain $\str{\F}$ in $\X$ so that we may quickly find the parenthesis representation of the light subtree of any node in $\F$. If $u \in \F$ is deleted, then we may use $\spl(\str{\F}, o(u))$ and $\spl(\str{\F}, o(u) + 1)$ followed by $\con(\str{F}[0 \dd o(u)-1], \str{\F}[o(u) + 1 \dd))$ to remove $\str{F}[o(u)]$ from $\str{\F}$ in $\X$.  We may then remove $\str{F}[c(u)]$ similarly.  If $u \in \F$ is inserted with label $\lambda(u)$, then we may similarly perform $\spl(\str{\F}, o(u)), \spl(\str{\F}, o(u)+ 1)$ and then $\con(\str{F}[0 \dd o(u) - 1], \lambda(u))$, $\con(\str{F}[0 \dd o(u) - 1] \cdot \lambda(u), \str{F}[o(u) \dd])$ to insert $\lambda(u)$ at $\str{F}[o(u)]$.  We may insert $\lambda(u)$ at $\str{F}[c(u)]$, similarly. By Theorem~\ref{thm:dynstr}, these steps take time $\Oh(\log n)$ since we are using $\Oh(1)$ $\spl$ and $\con$ operations.

    We now show how to maintain $\paren(\F) \in \X$. For a dynamically deleted node $u \in \F$, we may remove it from $\paren(\F)$ analogously to deletions from $\str{\F}$ described above. 
    For a dynamically inserted node $u \in \F$, we need to find a unique integer identifier for its light subtree that we may use as its label in $\paren(\F)$. To find the parenthesis representation of the light subtree of $u$, we just take its subtree in $\paren(F)$ and replace its heavy child $v$ with a dummy character \#. We perform $\spl(\str{F}, o(u))$ and then on the resulting string $\str{F}[o(u) \dd)$, perform $\spl(\str{F}[o(u) \dd], o(v))$ and $\spl(\str{F}[o(u)\dd n], c(v) + 1)$. We obtain $\str{F}[o(u)\dd o(v))$ and $\str{F}[c(u)\dd c(v))$, which are the two strings corresponding to the light subtree of $u$.  We then can concatenate the two substrings together with an additional \# between them to obtain the parenthesis representation of the light subtree of $u$, which may done with two $\con$ operations. If $u$ is a light node, we also concatenate its heavy depth as well as per Definition~\ref{def:parenF}. We let the unique integer identifier for the string corresponding to the light subtree of $u$ represent the label of $u$ in $\paren(\F)$ in $\X$.  We then insert the resulting integer identifier at positions $o(u)$ and $c(u)$ in $\paren(\F)$ analogously to insertions into $\str{F}$.  Note that the above steps again take time $\Oh(\log n)$.

    Now, a dynamic edit affects not only the edited node itself in $\paren(\F)$ but also the labels of other nodes.  
    The proofs of \cref{obs:node_delete,obs:node_relabel} characterize precisely which nodes have their labels affected: these are the parents or grandparents of light ancestors of the edited node $u$.  
    By \cref{lem:lightancestorset}, we can efficiently obtain the set of light ancestors of $v$ and recompute the light-subtree labels of their parents and grandparents, in a manner analogous to the dynamic insertion procedure described above.

    Furthermore, some nodes may change from heavy to light and light to heavy during a dynamic edit.  By Lemma~\ref{lem:heavylight}, we may maintain the heavy-light decomposition of $\F$ and $\G$ in $\Oh(\log^5)$-time.  For each such node $u$ that changes from heavy to light or light to heavy, we update the data structure of Lemma~\ref{lem:lightancestorset} accordingly in $\Oh(\log^2 n)$-time by Lemma~\ref{lem:lightancestorset} and reconstruct the light subtree label of its parent and add the resulting label to $\paren(\F)$ in $\X$ as done above. In total, this can be done in $\Oh(\log^7 n)$ time, and so each update to $\paren(\F)$ can be maintained in poly-logarithmic time.

    Finally, we use the dynamic string edit distance algorithm of Theorem~\ref{thm:dyned} to compute $\ed(\paren(\F), \paren(\G))$, which is our approximation of $\ted(\paren(\F), \paren(\G)).$ The approximation factor follows from Theorem~\ref{thm:dyned}, Lemma~\ref{lem:underapprox}, Lemma~\ref{lem:overapprox}, and the update time is dominated by Theorem~\ref{thm:dyned}.
\end{proof}

\section{\boldmath $k$ versus $\Oh(k^2 \log n)$ Approximation for Tree Edit Distance}
\label{sec:kTreeApprox}

We begin by giving a new static approximation algorithm for tree edit distance. Given forests $\F, \G$ that we wish to approximate the distance of, the following approximation algorithm maintains a collection of \emph{pieces} of $\F$ and iteratively checks if each piece has a corresponding occurrence in $\str{\G}$.  A \emph{piece} of $\str{\F}$ may be either a \emph{forest} $\str{\F}[i\dd j)$ which is a balanced parenthesis string, or a \emph{context} $(\str{\F}[i_1\dd j_1), \str{\F}[i_2\dd j_2))$ with $i_1 < j_1 < i_2 < j_2$ which is a pair of substrings such that $\str{\F}[i_1 \dd j_2)$ is a tree and $\str{\F}[j_1 \dd i_2)$ is itself a forest. A context may be thought of as a subtree of a node $w \in \F$ corresponding to indices $o(w) = i_1$ and $c(w) = j_2 - 1$ excluding, for some descendant $u$ of $w$, a contiguous set of children of $u$ and their subtrees. We call $\str{\F}[i_1 \dd j_1)$ the \emph{left} part and $\str{\F}[i_2 \dd j_2)$ the \emph{right} part. See Figure~\ref{fig:context} for an example. 
For a given forest piece $\str{\F}[i_1\dd j_1]$, we define $|\piece| := j - i$, and for a given context piece  $(\str{\F}[i_1\dd j_1), \str{\F}[i_2\dd j_2))$, we define $|\piece| := j_1 - i_1 + j_2 - i_2$.

In each iteration of the approximation algorithm, a single piece of the collection is checked for a corresponding occurrence in $\str{\G}$ up to a potential left or right shift in index of up to size $k$, and if an occurrence is found successfully, we remove that piece from the collection.  If no occurrence is found, it is partitioned into a constant number of new pieces of size smaller by a constant factor, which are added back into the collection. At any step, if a considered piece of our collection is of size $\Oh(k)$, we may remove it from the collection and move on to the next piece. We will argue that only $k \log |\F|$ iterations of this procedure is enough to distinguish whether $\ted(\F, \G) \leq O(k)$ or $\ted(\F, \G) \geq  \Oh(k^2 \log |\F|)$.   Since we only consider $\Oh(k \log |\F|)$ pieces in total, we know that ignoring pieces of size $\Oh(k)$ only incurs potentially $\Oh(k^2 \log |\F|)$ total edits. If $\ted(\F, \G) \leq k$, we will show that all $k$ edits will fall into pieces of size at most $\Oh(k)$ after $\Oh(k \log |\F|)$ steps, and so at the end of the algorithm, the collection of pieces will be empty in this case. Furthermore, the set of pairs of matched pieces found by the algorithm can be complemented by a set of deletions for any unmatched pieces to form a valid tree alignment of distance $\Oh(k^2 \log n)$, which we discuss in the proof of Theorem~\ref{thm:kApproxTree} below. If the collection of pieces is non-empty, this implies instead that $\ted(\F, \G) \geq k^2 \log |\F|$. 

\SetKwFunction{hm}{HasMatch}

An important subroutine in our algorithm will be finding corresponding pieces in $\str{\G}$ for each piece of $\str{\F}$. We refer to this subroutine as $\hm$ in the following pseudocode. For a given forest piece $\piece = \str{F}[i^{\F} \dd j^{\F})$ of $\F$, $\hm(\piece)$ returns \textsf{true} if there exists $i^{\G} \dd j^{\G}$ such that $ \str{\F}[i^{\F}\dd j^{\F}) =  \str{\G}[i^{\G} \dd j^{\G})$ and $|i_\F - i_\G| \leq k$, otherwise it returns \textsf{false}. For a given context piece $\piece_\F = (\str{\F}[i_1^{\F}\dd j_1^{\F})\dd \str{\F}[i_2^{\F} \dd j_2^{\F}))$, $\hm(\piece_\F)$ with $j_1^{\F} - i_1^{\F} \geq 2k$ and $j_2^{\F} - i_2^{\F} \geq 2k$ returns \textsf{true} if there exists indices $i_1^{\G}, i_2^{\G}, j_1^{\G}, j_2^{\G}$ such that $\piece_\G = (\str{\G}[i_1^{\G}\dd j_1^{\G}), \str{\G}[i_2^{\G} \dd j_2^{\G}))$ is a context with $\piece_\F = \piece_\G$ and $|i_1^\F - i_1^\G|, |i_2^\F - i_2^\G| \leq k$, otherwise it returns \textsf{false}.  For a context piece with only one of either $j_1^{\F} - i_1^{\F} \geq 2k$ or $j_2^{\F} - i_2^{\F} \geq 2k$, i.e. only the left or right part of $\piece_\F$ has size at least $2k$, we return \textsf{true} if the large part has a match in $\G$ up to a shift of size $k$ (and we ignore the smaller part since it does not affect the approximation factor).

In the following algorithm, we assume $M = O(1)$ and $C < 1$ is a constant.

\begin{algorithm}
\caption{$O(k^2 \log n)$-approximate Tree Edit Distance}
\label{alg:kApprox-Tree}
$S \leftarrow \str{\F}$ \;
\For{$i = 0$ \KwSty{to} $(1+C) k \log_{1/C} |\F|$}{
    $\piece \leftarrow S[0]$\;
    \uIf{$|\piece| \leq 4k$}{
        Remove $\piece$ from $S$\;
    }
    \uElseIf{$\hm(\piece)$}{
        Remove $\piece$ from $S$\;
    }
    \uElse{
        Partition $\piece$ into at most $M$ pieces $\piece_1, \piece_2, \ldots, \piece_M$ of size at most $C|\piece|$\;
        Add $\piece_1, \piece_2, \ldots, \piece_M$ to $S$\;
    }
}
\eIf{$|S| = 0$}
{
    Return \textsf{Yes}\;
}
{
    Return \textsf{No}\;
}
\end{algorithm}

\begin{theorem}
\label{thm:kApproxTree}
    Given forests $\F$ and $\G$ such that either $\ted(\F, \G) \leq k$ or $\ted(\F, \G) \geq k^2 \log |\F|$, Algorithm~\ref{alg:kApprox-Tree} outputs \textsf{Yes} if $\ted(\F, \G) \leq k$ and \textsf{No} if $\ted(\F, \G) \geq k^2 \log |\F|$.
\end{theorem}

\begin{proof}

First, we consider when $\ted(\F, \G) \leq k$ and show that Algorithm~\ref{alg:kApprox-Tree} always outputs \textsf{YES}. Let $n = |\F|$. Consider a tree in which each node of the tree corresponds to one of the pieces processed by Algorithm~\ref{alg:kApprox-Tree}. Any internal node corresponding to a piece $\piece$ is the parent of at most $M$ children
whose corresponding pieces $\piece_1, \piece_2, \ldots, \piece_M$ 
partition $\piece$ and each have size at most $C|\piece|$.  Any leaf node corresponds to a piece $\piece$ such that either $|\piece| \leq 4k$ or $\hm(\piece)$ is \textsf{true}. Note that each parenthesis in $\str{\F}$ appears in the piece of at most one node per level in the tree. Since $\ted(\F, \G) \leq k$, there at most $k$ parentheses undergoing edits in either $\F$ or $\G$, and so, at each level of the tree, at most $k$ pieces will be unmatched and potentially be partitioned further. Furthermore, the size of a piece at level $i$ is at most $\frac{n}{C^i}$, so there are at most $\log_{1/C} \frac{n}{k}$ levels before all pieces are smaller than $4k$ and removed from $S$. In total, we have at most $k \cdot \log_{1/C} (\frac{n}{k})$ internal nodes and therefore, a total of at most $(1+C)\cdot k  \log_{1/C} (\frac{n}{k})$ that nodes in the entire tree. In other words, after at most $(1+C)\cdot k  \log_{1/C} (\frac{n}{k})$ pieces are processed, Algorithm~\ref{alg:kApprox-Tree} will terminate with $S$ empty when $\ted(\F, \G) \leq k$. Since our algorithm runs for exactly $(1+C) k \cdot \log_{1/C} n$ iterations, the algorithm will return \textsf{Yes}.

Now, we consider when Algorithm~\ref{alg:kApprox-Tree} outputs \textsf{Yes}, and we show that in this case, $\ted(\F, \G) \leq \Oh(k^2 \log n)$. Even better, we will explicitly construct a sequence of pairs $\A$, which will be a tree alignment between $\str{\F}$ and $\str{\G}$ with cost $\Oh(k^2 \log n)$. 

Observe that when Algorithm~\ref{alg:kApprox-Tree} outputs \textsf{Yes}, $S$ must be empty, and so every parenthesis of $\F$ must have been part of a removed piece $\piece_\F$ such that either $|\piece_\F| \leq 2k$ or $\hm(\piece_\F)$ is \textsf{true}.  If $\hm(\piece_\F)$ is \textsf{true}, there must be some $\piece_\G$ that is a subforest or context of $\G$ such that $\piece_G = \piece_\F$ with starting index or indices matching that of $\piece_\F$ up to a shift of $k$ in either direction. We add a sequence of pairs to match $\piece_\F$ and $\piece_\G$, excluding the first and last $2k$ parentheses and their matches.

Formally, if $\piece_\F$ is a forest $\str{\F}[i^\F \dd i^\F]$ with match $\str{\G}[i^\G\dd j^\G]$, we add $(i^\F + 2k, i^\G + 2k), (i^\F +2k + 1, i^\G +2k+1), (i^\F + 2k + 2, i^\G + +2k 2), \ldots, (j^\F - 2k, j^\G -2k), (j^\F -2k, j^\G -2k)$ to $\A$. Note by Definition~\ref{def:ta}, we must careful that all nodes $u \in \F$ have both $o(u)$ and $c(u)$ either deleted or both aligned to some $o(v)$ and $c(v)$ for $v \in \G$. Since $\piece_\F$ is a forest, there are no nodes $u \in \F$ such that $i \leq o(u) < j < c(u)$, i.e., all parentheses in $\piece_\F$ have their corresponding twin in $\piece_\F$ as well. Thus, the only issue we must be careful about is if the opening parenthesis $o(u)$ of a node $u \in \F$ is in the first $2k$ parentheses of $\piece_\F$ and therefore excluded from the added matching sequence, but the corresponding twin $c(u)$ is contained in the matching sequence. In this case, we may replace the pairs matching $c(u)$ with two deletions.  Specifically, $(c(u), c(v)), (c(u) + 1, c(v)+1)$ is replaced with $(c(u), c(v)), (c(u) + 1, c(v)), (c(u) + 1, c(v) + 1)$, which deletes $c(u)$ and $c(v)$. Later, we will add deletions for $o(u)$ and $o(v)$ after adding all matching sequences to satisfy Definition~\ref{def:ta}.  We make the analogous modifications for the symmetric case when $c(u)$ is in the last $2k$ parentheses of $\piece_\F$ and excluded from the matching sequence but $o(u)$ is not excluded.

If $\piece_\F$ is a context $(\str{\F}[i_1^{\F}\dd j_1^{\F}], \str{\F}[i_2^{\F}, j_2^{\F}])$ with match $(\str{G}[i_1^{\G}\dd j_1^{\G}], \str{G}[i_2^{\G}\dd j_2^{\G}])$, similarly we add $(i_1^\F + 2k, i_1^\G + 2k), (i_1^{\F} + 2k +1, i_1^{\G} + 2k + 1), \ldots (j_1^\F -2k + 1, j_1^\G -2k + 1)$ to $\A$ if $j_1^\F - i_1^\F \leq 2k$ and $(i_2^\F +2k, i_2^\G + 2k), (i_2^\F + 2k + 1, i_2^\G + 2k + 1), \ldots, (j_2^\F -2k + 1, j_2^\G -2k + 1)$. We additionally modify this sequence to add deletions for any matched parentheses whose twins are not matched as in the forest case above. Furthermore, recall that if the left or right part of $\piece_\F$ has size less than $2k$, we do not find a match for it in the $\hm$ routine.  In this case, we add a matching sequence for only the part of $\piece_\F$ which has size at least $2k$.

Observe that after our above steps, any left index $x$ of a pair $(x, y) \in \A$ with $0 \leq x \leq |\F|$ is added to $\A$ at most once since when a piece of $\F$ has a match in $\G$, it is removed from $S$ for the remainder of Algorithm~\ref{alg:kApprox-Tree}. Any parentheses contained in a matched piece will never be matched in a future piece.
 Furthermore, any right index $y$ of a pair $(x, y) \in \A$ is similarly added to $\A$ at most once.
 Consider a second pair $(\piece_2^\F, \piece_2^\G)$ of matched pieces handled by the algorithm. First, since $\piece_1^\F, \piece_2^\F$ are both matched, they must have no overlap in $\str{\F}$, and second, according to the definition of $\hm$, since we only allow shifts of up to $k$, their corresponding pieces in $\G$ can only overlap in at most $2k$ characters. Since we are careful to skip the first and last $2k$ indices in our matching sequences added to $\A$, these right indices will never be added twice. 
 
 To finish our alignment $\A$, if $(0, 0)$ or $(|\F|, |\G|)$ is not already included in $\A$, we add these pairs to the beginning or end of $\A$, respectively. Finally, we add deletions for any remaining unmatched indices of $\F$ and $\G$. Formally, for any adjacent pairs $(x_1, y_1), (x_2, y_2)$ in $\A$ with $x_2 - x_1 > 1$ ($y_2 - y_1 > 1$) we add sequence $(x_1 + 1, y_1, x_1 + 2, y_1), \ldots, (x_1, y_2)$ ($(x_1, y_1+1, x_1, y_1 + 2), \ldots, (x_1, y_2)$) to $\A$ between these pairs. Deletions of indices in $\F$ correspond exactly to the pieces $\piece$ with $|\piece| \leq k$ and the $2k$ beginning and ending indices per matched piece. Since there are $(1+C)k \log_{1/C} n$ such pieces, we have at most $\Oh(k^2 \log_{1/C} n)$ deletions in $\F$. By choosing $\F$ and $\G$ such that $|\F| \geq |\G|$, there similarly must be at most $\G - \F + \Oh(k^2 \log_{1/C} n) \leq \Oh(k^2 \log_{1/C} n)$ deletions of indices in $\G$. Thus, $\A$ is a tree alignment with cost $\Oh(k^2 \log_{1/C} n)$.
\end{proof}

\subsection{$\tOh(k)$-time implementation}

We now discuss how Algorithm~\ref{alg:kApprox-Tree} can be implemented in $\tOh(k)$-time in the dynamic setting. The main two challenges we consider are: 1) the partitioning of a piece of forest $\F$ into $O(1)$ pieces with size at most a constant fraction of the original piece size 2) the $\hm$ routine.

To perform some useful basic operations on trees in the dynamic setting, we utilize the data structure of \cite{Navarro2014}. This dynamic tree data structure supports poly-logarithmic dynamic node insertion and deletion and can be used to answer a number of common query types that we list in the following. We assume, as in \cite{Navarro2014}, that forest $\F$ is stored in its parenthesis string form $\str{\F}$ and dynamic updates occur on this string.

\begin{theorem}[Theorem 2 and Table 1 of \cite{Navarro2014}]
\label{thm:navarro}
    Given an ordered tree $T$ with $n$ nodes the following operations can be carried out with $2n + \Oh(n \log \log n / \log n)$ bits:
    \begin{itemize}
        \item $\mathsf{insert}(i, j)$, insert a node at indices $i$ and $j$ in $\str{\F}$ in time $\Oh(\log n/ \log \log n)$.
        \item $\mathsf{delete}(v)$, delete a node $v$, in time $\Oh(\log n/ \log \log n)$.
        \item $\mathsf{depth}(v)$, return the depth of node $v$ in time $\Oh(\log n / \log \log n)$.
        \item $\laq(v, d)$, \emph{level ancestor queries}, which returns the ancestor $u$ of $v$ with $depth(u) = depth(v) - d$ in time $\Oh(\log n)$.
        \item $\lca(u, v)$, \emph{lowest common ancestor},  which returns the lowest common ancestor of nodes $u$ and $v$ in time $\Oh(\log n/ \log \log n)$.
        \item $\mathsf{parent}(v)$, which returns the parent of node $v$ in time $\Oh(\log n / \log \log n)$.
    \end{itemize}
\end{theorem}

We will create an instance of the above data structure for each tree of forests $\F$ and $\G$ and use the insert and delete operations to maintain these instances when a dynamic update occurs in $\F$ and $\G$ at the cost of only $\Oh(\log n / \log \log n)$-time per update.

In addition to the above operations, for a forest $\F$ and node $u$, we will often want to know $|\F(u)|$, the size of the subtree rooted at $u$ (in the parenthesis representation of $\F$). By Observation~\ref{obs:paren-anc}, it is easy to see that $\str{F}[o(u) \dd c(u)] = |\F(u)|$, and so we assume we have constant-time access to these sizes.

\subsubsection{Partitioning Forest and Context Pieces}

In order to partition a forest piece $\piece$ of $S$, we find a central node $u$ contained in a tree $T \in \piece$ such that the following four sets all have size at most $n/2$:
\begin{itemize}
    \item  siblings to the left of $u$ and their subtrees,
    \item siblings to the right of $u$ and their subtrees,
    \item the subtree of of $u$,
    \item all remaining nodes.
\end{itemize}

We must be careful with the fourth set, the remaining nodes, as they may not be either a forest nor context. To handle this, we further separate the remaining nodes into three more additional sets:
\begin{itemize}
    \item the forest to the left of $T$,
    \item the forest to the right of $T$,
    \item the context formed of the remaining nodes of $T$ not yet in a piece.
\end{itemize}
We will later show that each part of the decomposition has size at most $|\piece|/2$, and show that with some additional work, a similar decomposition works for context pieces as well. We utilize the data structure of \cite{Navarro2014} to answer two additional query types to help us find our key node $u$.

\begin{lemma}
\label{lem:findHighestAnc}
    Given a forest $\F$ with $n = |\F|$, a subforest $\str{\F}[i \dd j)$, and node $v$ with $i \leq o(v) < c(v) < j$, the highest ancestor $u$ of $v$ satisfying $i \leq o(u) < c(u) < j$ can be found in time $\Oh(\log^2 n / \log \log n)$.
\end{lemma}

\begin{proof}
    To find $u$ from node $v$, we utilize a binary search on the path from the root of the tree $T$ containing $v$.  First, using the data structure of Theorem~\ref{thm:navarro}, we get the depth $d$ of $v$ and perform a level ancestor query $\laq(v, d/2)$.  Let $v'$ be the returned node of this query and let $v''$ denote its parent if it exists.
    
    If $o(v') < i$, $v'$ is outside of the subforest $\str{F}[i\dd j]$ and so, we continue the search recursively moving back towards $v$, i.e., we would next query $\laq(v, 3d/4)$ in order to find an ancestor in the target subforest. Otherwise, we continue the search recursively moving towards the root of $T$. In this case, we would next query $\laq(v, d/4)$. If we always traverse towards the root and eventually reach the root itself without stopping in one of our earlier cases, we return the root as the highest ancestor of $v$ in $\str{F}[i, j]$. Otherwise, we must have found our target node on the root to $v$ path and will return that node instead from one of the above cases.

    Since we halve our search range every iteration, our algorithm will finish in $\Oh(\log n)$ iterations and by Theorem~\ref{thm:navarro}, each iteration takes time $\Oh(\log n/ \log \log n)$ for a total time of $\Oh(\log^2 n / \log \log n)$.
\end{proof}

\begin{lemma}
\label{lem:findMidAncestor}
    Given a forest $\F$ with $n = |\F|$, subforest $\str{F}[i \dd j)$, a threshold $m$, and node $v$ with $i \leq o(v) < c(v) < j$, the highest ancestor $u$ of $v$ satisfying $|\F(u)| \leq m$ and $i \leq o(u) < c(u) < j$ can be found in time $\Oh(\log^2 n / \log \log n)$ if it exists.
\end{lemma}

\begin{proof}
    To find $u$ from node $v$, we utilize a binary search on the path from the root of the tree $T$ containing $v$.  First, using the data structure of Theorem~\ref{thm:navarro}, we get the depth $d$ of $v$ and perform a level ancestor query $\laq(v, d/2)$.  Let $v'$ be the returned node of this query and let $v''$ denote its parent if it exists.
    
    If $o(v') < i$, $v'$ is outside of the subforest $\str{F}[i\dd j]$ and so, we continue the search recursively moving back towards $v$, i.e., we would next query $\laq(v, 3d/4)$ in order to find an ancestor in the target subforest. Similarly, if $|\F(v')| > m$, we continue the search recursively towards $v$ since the size of the subtree of $v'$ is too large. Instead, if $|\F(v')| \leq m$ and $o(v'') < i$, then we return $v'$, the highest ancestor of $v$ in $\str{F}[i, j)$. If $|\F(v')| \leq m$ and $|\F(v'')| > m$,  we return $v'$ as it satisfies the lemma requirements. Finally, $|\F(v')| \leq m$ and $|\F(v'')| \leq m$, we continue the search recursively moving towards the root of $T$. In this case, we would next query $\laq(v, d/4)$. If we always traverse towards the root and eventually reach the root itself without stopping in one of our earlier cases, we return the root as the highest ancestor of $v$ in $\str{F}[i\dd j]$. Otherwise, we must have found our target node on the root to $v$ path and will return that node instead from one of the above cases.  If we finish our search and do not find $u$, we return nothing.

    Since we halve our search range every iteration, our algorithm will finish in $\Oh(\log n)$ iterations and by Theorem~\ref{thm:navarro}, each iteration takes time $\Oh(\log n/ \log \log n)$ for a total time of $\Oh(\log^2 n / \log \log n)$.
\end{proof}

We now give the formal piece partition statement.

\begin{lemma}
\label{lem:partitionpiece}
    Given forest $\F$ and either forest piece $\piece = \str{F}[i \dd j)$ or context piece $\piece = (\str{F}[i_1 \dd j_1), \str{F}[i_2 \dd j_2))$, we may partition $\piece$ into at most 8 pieces each of size at most $\lceil |\piece|/2 \rceil$ in time $\Oh(\log^2 n / \log \log n)$.
\end{lemma}

\begin{proof}
    We begin with the treatment for forest pieces $\piece = \str{F}[i \dd j)$. Let $m = \lceil \frac{j - i}{2} \rceil$. First, we check index $m + i$. We assume $\str{F}[m + i]$ is an opening parenthesis, otherwise the following steps may be done symmetrically in the other direction. Let $v$ be the node of $\F$ corresponding to $\str{F}[m + i]$.  Since $v$ is the node corresponding to the midpoint of $\piece$, by Lemmas~\ref{lem:findHighestAnc} and~\ref{lem:findMidAncestor}, we may find ancestors $w$, the highest ancestor of $v$ in $\str{F}[i \dd j)$ and $u$, the highest ancestor of $v$ in $\str{F}[i  \dd j)$ such that $|\F(u)| \leq m$ if it exists, in time $\Oh(\log^2 n / \log \log n)$. Additionally, by Theorem~\ref{thm:navarro}, we may find the parent $u'$ of $u$ if it exists in time $\Oh(\log n / \log \log n)$.  We now list the six parts we partition $\piece$ into and then discuss their size. We utilize the following observation that since $u$ is an ancestor of $v$, we have that $o(u) \leq o(v) = m < c(v) \leq c(v)$.

    \begin{enumerate}
        \item If $u$ exists, the first part is $\str{F}[o(u) \dd c(u)]$, the subtree rooted at node $u$. By Lemma~\ref{lem:findMidAncestor}, this part must have size at most $m$.
        \item If $u$ exists and $w \ne u$, the second part is $\str{F}(o(u') \dd o(u))$, the forest composed of subtrees rooted at siblings of $u$ to the left of $u$. Since $w \ne u$ and $w$ is the highest ancestor of $v$ satisfying $i \leq o(w)$, then we have that $i \leq o(w) \leq o(u')$. Furthermore, we know that $o(u) \leq o(v) = m + i$, and so the size of this part is bounded by $o(u) - o(u') \leq (m + i) - i$ = m.
        \item If $u$ exists and $w \ne u$, the third part is $\str{F}(c(u) \dd c(u'))$, the forest composed of subtrees rooted at siblings of $u$ to the right of $u$. By a symmetric argument as to the previous part, the size of this part is bounded by $m$.
        \item If $u$ exists and $w \ne u$, the fourth part is $(\str{F}[o(w)\dd o(u')], \str{F}[c(u') \dd c(w)])$, the context rooted at $w$ excluding the subtree of $u'$.  By Lemma~\ref{lem:findMidAncestor}, $|\F(u')| > m$ and so, this part can have size at most $j - i - m \leq m$.
        \item The fifth part is $\str{F}[i \dd o(w))$, the forest composed of subtrees rooted at siblings of $w$ to the left of $w$.  By the same reasoning as the fourth part, the size of this part is bounded by $m$.
        \item The sixth part is $\str{F}(c(w)\dd j)$, the forest composed of subtrees rooted at siblings of $w$ to the right of $w$.  By the same reasoning as the fourth part, the size of this part is bounded by $m$.
    \end{enumerate}

    The above partitions $\str{F}[i \dd j)$ into six pieces of size at most $m$, which is half of the original forest piece, in time $\Oh(\log^2 n / \log \log n)$.  See Figure~\ref{fig:forest-part} for an example of this partition scheme.

    Now we show how to partition a context piece $\piece = (\str{F}[i_1 \dd j_1), \str{F}[i_2 \dd j_2))$. Let $m = (j_1 - i_1 + j_2 - i_2)/2$.  We begin by finding node $v$, the parent of the subforest $\str{F}[j_1 \dd i_2)$, which we may find using Theorem~\ref{thm:navarro} by querying the parent of the node corresponding to $\str{F}[j_1]$. We will have two cases for partitioning contexts. 
    
    First, if $|\F(v)| - (i_2 - j) \leq m$ then we repeat the same partitioning as for forest pieces using nodes $w$ and $u$ from Lemmas~\ref{lem:findHighestAnc} and~\ref{lem:findMidAncestor} in relation to $v$.  The main change is that the first part, subforest $\str{F}[o(u) \dd c(u)]$ is replaced with context $(\str{F}[o(u) \dd j_1), \str{F}[i_2 \dd c(u)])$ as we must still exclude the missing subforest $\str{F}[j_1 \dd i_2)$.  Also, we note that by the definition of contexts, parts five and six will be empty since a context contains nodes only in a single tree, and so $\piece$ will not contain any siblings of the root of this tree. Therefore, we will have 4 pieces total of size $m$.

    If instead $|\F(v)| - (i_2 - j_1) > m$, we cannot repeat the same partitioning as already the subtree of $v$ is too large to be its own part. Observe that only one of $|(o(v) \dd j_1)| > m$ or $|[i_2 \dd c(v))| > m$ may be true. Without loss of generality, let us assume the former, $|(o(v) \dd j_1)| > m$.  In this case, we set $|[i_2 \dd c(v))|$ to be its own forest part of size at most $m$. We additionally set $(\str{F}[i_1 \dd o(v)], \str{F}[c(v)\dd j_2))$ to be a context part of size at most $(j_1 - i_1 + j_2 - i_2) - (|\F[v]| - (i_2 - j_1)) \leq 2m - m = m$. Finally, with the remaining forest $\str{F}(o(v) \dd j_1)$, we partition it into 6 pieces in the same way we did in the above case for forest pieces. Therefore, we get 8 pieces total of size at most $m$ in time $\Oh(\log^2 n/ \log \log n)$.
\end{proof}

\subsubsection{The \textsf{HasMatch} Routine}

We now show how, given a forest or context piece $\piece$ in $\F$, we may find a matching piece in $\G$ if one exists up to a $k$-index shift.

We primarily will rely on the data structure of \cite{KK22}, which supports updates and \emph{internal pattern matching queries} ($\ipm$) in $\Oh(\log^{1+o(1)} n)$ time each. Given a pattern $P$ and a text $T$, $\ipm(P, T)$ reports all occurrences of $P$ in $T$. For $|T| < 2|P|$, all occurrences are reported as a single arithmetic progression of starting positions. Since $\hm$ is only interested in finding matches up to a $2k$-shift in either direction, in order to limit our search, we use the $\spl$ operation so that we may cut out exactly the set of indices we need to perform an $\ipm$ query on.

Finding matches in $\G$ for forest pieces of $\F$ is straightforward applications of the previous theorem. Finding matches in $\G$ for context pieces of $\F$ requires additional work due to being two separate parts, which we cannot find a match for with a single IPM query.  Instead, we use IPM queries for each part of a context piece, which return two separate arithmetic progressions for the starting indices of matches for each part. Using $\mathsf{depth}$ queries from the data structure of Theorem~\ref{thm:navarro}, we may rewrite each progression as an arithmetic progressions of the depths of the starting indices of matches for each part.  Then, we find where these depths are equal between each progression and check if at such an equal depth these matches correspond to a single context piece in $\G$. We give the statement formally as follows.

\begin{lemma}
\label{lem:hm}
    Given forests $\F, \G$, and a piece $\piece$ of $\F$, $\hm(\F)$ can be computed in time $\Oh(\log (|\F| + |\G|))$.
\end{lemma}

\begin{proof}
    We utilize an instance $\X$ of the dynamic strings data structure of \cite{KK22}.
    We first consider when a piece $\piece$ that we pass to the $\hm$ routine is a forest piece $\str{\F}[i \dd j)$.  In order to limit our search in $\G$ to just the starting index range $[i - 2k \dd i + 2k]$, we first perform two $\spl$ operations on $\str{\G}$ to retrieve $\str{\G}[i-2k\dd j+2k)$, the substring of $\G$ that may contain a match for $\piece$ with a shift of up to $2k$.  
    Finally, we query $\ipm(\piece, \str{\G}[i-2k \dd j + 2k))$ using $\X$ and, if there is a match, output $\hm(\piece) = \mathsf{true}$. By Theorems~\ref{thm:dynstr} and~\cite{KK22}, this takes time $\Oh(\log (|\F| + |\G|))$ with high probability. 

    Now, if a piece $\piece$ that we pass to the $\hm$ routine is a context $\langle\str{\F}[i_1 \dd j_1), \str{\F}[i_2 \dd j_2)\rangle$, we need to check the left and right parts of the context separately and then find if any matches for each part correspond to a single context\footnote{As per the description of $\hm$, if either part has size at most $4k$, we ignore it and treat the remaining part as a forest piece.}.
    Now, using the same steps as for forest pieces, we may find an arithmetic progression for all occurrences of each part in $\G$ separately with  a shift of size up to $k$ using the dynamic strings data structure $\X$ to perform $\spl$ and $\ipm$ operations. We may have multiple matches for the left and right parts separately, and we must check if any pair of these matches forms a single context, i.e., the unmatched opening parentheses in the left part correspond to the same nodes as the unmatched closing parentheses in the right part. 

    Since the string $\str{\F}[i_1 \dd j_1)$ is of length more than $4k$, its occurrences within $\str{\G}[i_1-2k \dd j_1+2k)$ overlap,
    and their starting positions form an arithmetic progression $i_1+\delta_{1,1},i_1+\delta_{1,2},\ldots, i_1+\delta_{1,s}$.
    Moreover, since $\str{\F}[i_1 \dd j_1)$ does not contain any unmatched closing parenthesis, the same must be true about
    $\str{\G}[i_1+\delta_{1,1}\dd j_1+\delta_{1,s})$.
    Consequently, all potential matches for the root of the context (the opening parenthesis at position $\str{F}[i_1]$)
    must lie on a single root-to-leaf path of~$\G$.
    Furthermore, the depths of these potential matches also form an arithmetic progression (whose difference is the number of unmatched opening parentheses in  $\str{G}[i_1+\delta_{1,1}\dd i_1+\delta_{1,2})=\ldots = \str{G}[i_{1}+\delta_{1,s-1}\dd i_1+\delta_{1,s})$.
    A symmetric argument shows that the possible matches for the root of the context based on the exact occurrences of the right part $\str{\F}[i_2 \dd j_2)$
    within $\str{\G}[i_2-2k \dd j_2+2k)$ also lie on a single root-to-leaf path at depths constituting an arithmetic progression.

     Let $\{a_1\cdot x + b_1 : x \in \{0,\ldots, c_1\}\}$ be the arithmetic progression of the depths for the set of returned left part matches and let $\{a_2\cdot x + b_2 :  x \in \{0,\ldots, c_2\}\}$ be the arithmetic progression of the depths for the right part matches, which we may obtain using $\mathsf{depth}$ queries in $\Oh(\log |\G|)$-time with the data structure of Theorem~\ref{thm:navarro}. 
     We intersect the two arithmetic progressions and, if the intersection is empty, return $\mathsf{false}$.
     If the intersection is non-empty, we identify its minimum element and retrieve the corresponding occurrences $\str{G}[i'_1\dd j'_1)$ and $\str{G}[i'_2\dd j'_2)$ of $\str{\F}[i_1 \dd j_1)$
     and $\str{\F}[i_2 \dd j_2)$, respectively.
     We return $\mathsf{true}$ if and only if $\str{G}[i'_1\dd j'_1)$ is balanced.
     
     In this case, $\langle \str{G}[i'_1\dd j'_1), \str{G}[i'_2\dd j'_2) \rangle$ is indeed a valid context: since $\str{G}[i'_1]$ and $\str{G}[j'_2-1]$ are opening and closing, respectively, and correspond to nodes at the same level, the unmatched parentheses from $\str{G}[i'_1\dd j'_1)$ and $\str{G}[i'_2\dd j'_2)$ are all matched within $\str{G}[i'_1\dd j'_2)$.

     It remains to prove that we have not lost a valid match by greedily taking the minimum depth.
     For this, suppose that there was valid match $\langle \str{G}[i''_1\dd j''_1), \str{G}[i''_2\dd j''_2) \rangle$ at a larger depth.
     In this case, $\str{G}[j''_1\dd i''_2)$ is balanced, and so is $\str{G}[i''_1\dd j''_2)$.
     Then $\str{G}[i'_1\dd j'_1)$ is sandwiched between the two, and hence also balanced: any hypothetical unmatched parenthesis within $\str{G}[i'_1\dd j'_1)$ must be matched within $\str{G}[i''_1\dd j''_2)$ but, since $\langle \str{G}[i''_1\dd j''_1), \str{G}[i''_2\dd j''_2) \rangle$ is a context, the unmatched parentheses of $\str{G}[i''_1\dd j'_1)$ are matched within $\str{G}[i'_2\dd j''_2)$, and vice versa.
\end{proof}

\subsection{Dynamic Approximation Algorithm}

The results above can be interpreted as a dynamic data structure supporting the following two operations:
\begin{description}
    \item[Updates] that modify  $\F$ or $\G$ using a given edit operation, supported in $\Oh(\log^{1+o(1)} n)$ time. These updates simply maintain the dynamic strings data structure \cite{KK22} and the dynamic forest data structure of \cref{thm:navarro}.
    \item[Queries] that, given an integer $k$, return \textsf{Yes} if $\ted(\F,\G) \le k$ and only if $\ted(\F,\G) \le  C k^2 \log n$ (for some explicit constant $C$).
    The query procedure works in $\Oh(k\log^{3+o(1)} n)$ time.
\end{description}

We now show how to apply these two procedures in order to maintain an approximate tree edit distance $\apted(\F,\G)$ that, all times, satisfies $\ted(\F,\G)\le \apted(\F,\G) \le \Oh(\ted(\F,\G)\log n)$,
and takes $\Oh(\log^{3+o(1)}n)$ worst-case time per update.

We always ask queries in succession for subsequent thresholds $k=1,2,4,\ldots$ until the answer is $\textsf{Yes}$. 
The threshold $\kappa$ for which this condition happens satisfies $\kappa/2 \le \ted(\F,\G) \le C \kappa^2 \log n$, and it is computed in time at most $\Oh(\kappa \log^{3+o(1)} n)$. 
We call this an \emph{extended query procedure}.

Our approximation algorithm maintains the aforementioned dynamic data structure, $\mathcal{D}$, which stores forests $\F$ and $\G$.
The updates are not immediately applied to $\mathcal{D}$, though, but stored in a buffer $B$ (implemented as a simple queue).
We occasionally run the extended query procedure on $\mathcal{D}$ in small $\Oh(\log^{3+o(1)}n)$-time steps so that a value $\kappa$ is guaranteed to be returned in no more than $\lceil 0.1\kappa \rceil$ steps.
Furthermore, we store the result $\kappa$ of the last completed execution of the extended query procedure, and the value $b$ representing the number of updates (applied to $\mathcal{D}$ or still in the buffer) since that execution has been launched.

Initially, $\mathcal{D}$ stores two empty forests, the buffer $B$ is empty, no extended query procedure is underway, $\kappa=0$ (which is the tree edit distance of empty strings), and $b=0$.
At each update, we perform the following steps:
\begin{enumerate}
    \item Buffer the incoming update and increment $b$.
    \item If there is no extended query procedure underway:
    \begin{itemize}
        \item Apply the first $\min(|B|,2)$ updates from the buffer.
        \item If the buffer is now empty, launch a new  extended query procedure (without running it yet).
    \end{itemize}
    \item If there is an extended query procedure underway (possibly just launched):
    \begin{itemize}
        \item Run the extended query procedure for one step.
        \item If the procedure terminates, update $\kappa$ and set $b := |B|$.
    \end{itemize}
    \item Return $C\cdot \kappa^2 \log n + b$.
\end{enumerate}

\begin{lemma}
    The returned value $\apted(\F,\G)$ satisfies $\ted(\F,\G)\le \apted(\F,\G) \le \Oh(\ted(\F,\G)^2 \log n)$.
\end{lemma}
\begin{proof}
    The lower bound for $\apted(\F,\G)$ is easy: the forests $\F^*$ and $\G^*$ from the moment when we launched the previous extended query procedure
    satisfy $\ted(\F^*,\G^*) \le C \kappa^2 \log n$ as well as $\ted(\F,\G) \le \ted(\F^*,\G^*)+b$.
    Combining these two inequalities, we get $\ted(\F,G) \le C \kappa^2 \log n + b = \apted(\F,\G)$.

    As for the upper bound, let us first note that 
    $\ted(\F,\G) \ge \ted(\F^*,\G^*)-b \ge 0.5\kappa-b$.
    When the extended query procedure terminates with value $\kappa$, it has been running for at most $\lceil{0.1 \kappa}\rceil$ steps, and thus the buffer contains at most $0.1\kappa$ unprocessed updates, that is, we have $b\le 0.1\kappa$ once the procedure terminates. These updates are applied in the double speed, and thus $b \le 0.2\kappa$ once the next extended procedure is launched.

    Consequently, if the update algorithm returns $\apted(\F,\G)$ when no extended query algorithm is underway or one has already been running for $s \le 0.05\kappa$ steps, then $b\le 0.25\kappa$, and hence
    $\ted(\F,\G) \ge 0.5\kappa - 0.25\kappa = 0.25\kappa$.
    Thus, $\apted(\F,\G) = C \kappa^2 \log n + b \le C \kappa^2 \log n + 0.25\kappa \le 16 C \ted(\F,\G)^2 \log n + \ted(\F,\G).$

    Otherwise, an extended query procedure is underway running for $s \ge 0.05\kappa$ steps already, and then $b \le 0.2\kappa + s \le 5s$.
    Moreover, we know that a value $\kappa'\ge 10s$ will be returned by the procedure underway, and the instance $\F',\G'$ from its launch time satisfies $\ted(\F',\G') \ge \kappa'/2 = 5s$.
    Since $s$ updates have been requested since the launch, we have $\ted(\F,\G) \ge \ted(\F',\G')-s \ge 4s$.
    Thus, $\apted(\F,\G) = C \kappa^2 \log n + b \le 400 C s^2 \log n + 5s \le 25 C \ted(\F,\G)^2 \log n + 1.25 \ted(\F,\G).$ 
\end{proof}

In order to derive \cref{thm:dyn-tree-kk2-simple}, it suffices to observe that the procedure that we run at each update takes $\Oh(\log^{3+o(1)} n)$ time, dominating by the single step of the extended query algorithm.

\section{Further Results on Approximating Dynamic Dyck Edit Distance}\label{sec:approx}
In this section, we present two algorithms that each return good approximate solutions for different regimes of Dyck edit distance.  The first algorithm provides an approximation factor of $\Oh(f(\epsilon)\log |X|)$ in $\tOh_\epsilon(|X|^{1+\epsilon}/\ded(X))$ update time where $f(\epsilon)$ is a function dependent only on a fixed parameter $\epsilon$. The second algorithm gives an $\Oh(\log |X|)$-approximations in $\tOh(\ded(X) + 1)$ amortized update time. Thus, the first algorithm is preferable for $\ded(X) = \Omega(\sqrt{n})$ while the second algorithm is preferable for $\ded = \Oh(\sqrt{n})$.  These algorithms can be combined to create a logarithmic-approximation algorithm for dynamic Dyck edit distance for any input string. 

In this section, we will often consider approximate Dyck edit distance with deletions only.  For a parenthesis string $X$, $\dedd(X)$ refers to the \emph{deletion-only} Dyck edit distance.
To see why considering only deletions is still satisfactory for our purposes, observe that any substitution in a minimum sequence of edits occurs so that two parentheses can be matched together, and instead, we may replace the substitution by deleting the formerly matched pair of parentheses. 
Since we already lose a factor of $\Oh(\log n)$ in our approximation algorithms, we may handle doubling the number of edits needed by disallowing substitutions without increasing our asymptotic approximation factor.
As for insertions, any minimal sequence of edits with insertions and deletions has a corresponding sequence of edits with deletions only with the same number of total edits.
This is because whenever a parenthesis is inserted into $X$ to be matched with a pre-existing parenthesis in $X$, we can instead forego the insertion and simply delete the existing parenthesis in $X$ to satisfy the requirements of a $\Dyck$ string.

\subsection{Large Dyck Edit Distance}\label{subsec:largeed}

We first consider the case when the Dyck edit distance of the input string $X\in \Sigma^n$ is large, that is $k = \dedd(X)$ satisfies $k\ge \sqrt{n}$.
Since $k$ is so large and a single edit changes $\dedd(X)$ by at most $\pm 1$, at a constant-factor loss in the approximation ratio,
it suffices to update our approximation of $k=\dedd(X)$ every $\Theta(k)$ updates. 
As a result, we can afford to use a static approximation algorithm because its large cost amortizes over the $\Theta(k)$ updates.

The first tool we use is an $\Oh(\log n)$-approximate reduction from Dyck edit distance to string edit distance from \cite{KSSODA2023}.
This technique is a large part of several of our later algorithms.
We give some necessary definitions to understand this reduction.  
A string $Y$ is called an \emph{LR-string} if it can be decomposed into a prefix $\Lp(Y)$ consisting of opening parentheses and a suffix $\Rp(Y)$ consisting of closing parentheses, i.e. $Y = \Lp(Y) \cdot \Rp(Y)$. 
A \emph{maximal LR-segment} of a string $X$ is an LR-string that is a substring of $X$ and cannot be extended to a larger LR-string. 
Given a string $X$, its maximal LR-strings determine a unique partition of $X$ into disjoint substrings, the \emph{LR-decomposition} of $X$.

We observe that finding the deletion-only Dyck edit distance of an LR-string $X$ is equivalent to finding the minimum number of edits needed to transform the maximal-length sequence of opening parentheses in $X$ into the transpose of the maximal-length sequence of closing parentheses in $X$.  

\begin{observation}
\label{obs:LRed}
    Every LR-string $X$ satisfies $\dedd(X) = \edd(\Lp(X), T(\Rp(X)))$.
\end{observation}
\begin{proof}
    Let $X'$ be a longest subsequence of $X$ such that $X' \in \Dyck$. 
    Since $X$ is an LR-string, $X'$ must also be an LR-string and, furthermore, since $X' \in \Dyck$, it must be that $\Lp(X') = T(\Rp(X'))$ is a common subsequence of $\Lp(X)$ and $T(\Rp(X))$.
    Therefore, $\dedd(X)= |X|-|X'|=|\Lp(X)|+|T(\Rp(X))|-2|\Lp(X')| \le  \edd(\Lp(X), T(\Rp(X)))$.
    Similarly, let $Z$ be a longest common subsequence of $\Lp(X)$ and $T(\Rp(X))$. 
    Note that $Z\cdot T(Z)$ belongs to $\Dyck$ and is a subsequence of $X$.
    Consequently, $\edd(\Lp(X), T(\Rp(X))=|\Lp(X)|+|T(\Rp(X))|-2|Z|=|X|-|Z\cdot T(Z)| \le \dedd(X)$.
\end{proof}

We now state the reduction.
\begin{theorem}[Theorem 7.1, \cite{KSSODA2023}]
\label{thm:dycktoed}
    There is a deterministic algorithm that, given a string $X$, in $\tOh(|X|)$ time outputs a collection $\C(X)$ of LR-strings such that:
    \begin{enumerate}
        \item $\sum_{Y \in \C(X)} |Y| = |X|$, and
        \item $\dedd(X) \leq \sum_{Y \in \C(X)} \dedd(Y) \leq \left(3 + 2\lg (\dedd(X))\right) \dedd(X)$.
    \end{enumerate}
\end{theorem}

The preceding theorem, for a parenthesis string $X$, outputs a collection $\C(X)$ of LR-strings which may be used to find a logarithmic-factor approximation of the Dyck edit distance of $X$.
We will then use string edit distance algorithms to find the Dyck edit distance of each such LR-string according to Observation~\ref{obs:LRed}. 
The particular string edit distance approach we utilize will be as fast as possible while making sure not to add any super-logarithmic factors to the approximation ratio. 

\begin{theorem}[{\cite[Theorem 1.1]{ANFOCS20}}]\label{thm:consted}
    There exists a function $f:\mathbb{R}_+\to \mathbb{R}_+$ and a randomized algorithm that, given two strings $X,Y\in \Sigma^{\le n}$ and a parameter $\epsilon\in \mathbb{R}_+$, computes an $f(\epsilon)$-factor approximation of $\edd(X,Y)$ in $\Oh_\epsilon(n^{1+\epsilon})$ time correctly with high probability.\footnote{\label{ftn:fep}As discussed in \cite{ANFOCS20}, the value $f(\epsilon)$ is doubly exponential in $1/\epsilon$ and independent of any other variables. The subscript $\epsilon$ in $\tOh_\epsilon$ denotes hidden $\epsilon$ factors.}
\end{theorem}

With the above techniques, the main result of this subsection follows.

\largeDed*

\begin{proof}
    The algorithm proceeds in epochs.
    Each epoch starts with finding the collection $\C(X)$ using the algorithm of Theorem~\ref{thm:dycktoed}. 
    Recall that each string $Y \in \C(X)$ is an LR-string and $\dedd(Y) = \edd(\Lp(Y), T(\Rp(Y)))$ holds by Observation~\ref{obs:LRed}.
    For each $Y \in \C(X)$, the second step is to approximate $\edd(\Lp(Y), T(\Rp(Y)))$ using the algorithm of Theorem~\ref{thm:consted}.
    Let $a_Y$ represent the approximated edit distance found for each $Y \in \C(X)$; we are guaranteed that $\edd(\Lp(Y), T(\Rp(Y))) \le a_Y \le f(\epsilon) \edd(\Lp(Y), T(\Rp(Y)))$. 
    We use $a := 2\sum_{Y\in \C(X)} a_Y$ as the approximation of $\dedd(X)$ for the duration of the epoch, which lasts for $m:=\lfloor a / (4f(\epsilon)(3+2\lg a))\rfloor$ dynamic edits.
    At the beginning of the epoch,
    \begin{align*}
        a = 2\sum_{Y \in \C(X)} a_Y 
        & \geq 2\sum_{Y \in \C(X)} \edd(\Lp(Y), T(\Rp(Y)))\\
        & =2\sum_{Y \in \C(X)} \dedd(Y)\\
        &\ge 2\cdot \dedd(X) \ge 2\cdot \dedd(X) \\
        & \text{and}\\
        a= 2\sum_{Y \in \C(X)} a_Y 
        &\leq 2\sum_{Y \in \C(X)} f(\epsilon) \edd(\Lp(Y), T(\Rp(Y))) \\
        &= 2f(\epsilon) \sum_{Y \in \C(X)} \dedd(Y) \\
        &\leq 2f(\epsilon)(3 + 2\lg (\dedd(X))) \dedd(X)\\
        &\leq 2f(\epsilon)(3 + 2\lg a) \dedd(X) \\
        &\leq 4f(\epsilon)(3 + 2 \lg a)\dedd(X),
    \end{align*}
    where, in both chains, the first inequality follows from Theorem~\ref{thm:consted} and the second inequality follows from Theorem~\ref{thm:dycktoed}.

    Throughout the epoch, the value $\dedd(X)$ may change by at most $\pm m$, so we still have 
    $a \ge \frac{a}{2}+m \ge \dedd(X)-m+m = \dedd(X)$
    and $a = 2a - a \le 8f(\epsilon)(3 + 2\lg a)(\dedd(X)+m)-a \le 8f(\epsilon)(3 + 2\lg a)\dedd(X)+8f(\epsilon)(3 + 2\lg a)m-a \le 8f(\epsilon)(3 + 2\lg a)\dedd(X)$, so the value $a$ remains an $\Oh(f(\epsilon)\log n)$-approximation of $\dedd(X)$.

    It remains to analyze the running time. The application of Theorem~\ref{thm:dycktoed} takes $\tOh(|X|)$ time.
    The application of Theorem~\ref{thm:consted} takes $\Oh_\epsilon(|Y|^{1+\epsilon})$ time for each $Y\in \C(X)$.
    Across all $Y\in \C(X)$, this cost is proportional to
    \[
        \sum_{Y \in \C(X)} |Y|^{1+\epsilon} \leq \sum_{Y \in \C(X)} |Y|\cdot |X|^{\epsilon} = |X|^{1+\epsilon}.
    \]  Furthermore, summing up each approximation for the edit distance of all strings in the LR-decomposition of $X$ takes at most $\Oh(|X|)$ time. 
    Consequently, the total cost of each epoch is $\Oh_\epsilon(|X|^{1+\epsilon})$. 
    Across $m=\Omega(\dedd(X)/\log |X|)$ updates within the epoch, this yields $\tOh(|X|^{1+\epsilon}/\dedd(X))$ amortized time per update.
\end{proof}
\subsection{Small Dyck Edit Distance}\label{subsec:smalled}

To complement the previous dynamic approximate Dyck edit distance algorithm, we consider the case when the Dyck edit distance of an input string is small.  Specifically, given an input string $X$ with $|X| = n$, let $k = \dedd(X)$; we aim to match our approximation ratio and update time from the large edit distance case of the previous section when $k = \Oh(\sqrt{n})$ instead. To this end, we pre-process $X$ such that all neighboring parentheses which may be matched together are removed.  We denote the remaining string as $\hat{X}$, which has the same Dyck edit distance as $X$ but is composed of $\Oh(k)$ maximal LR-segments instead of potentially $\Oh(n)$ segments. We are careful to maintain $\hat{X}$ as $X$ is updated, potentially removing any new pairs of neighboring matching parentheses or even reintroducing any previously removed parentheses if their neighbor no longer matches after an update. 
We again use the Dyck to string edit distance reduction to find a collection of LR-segments that we use to approximate the Dyck edit distance of $X$. Then, we introduce a string edit distance algorithm of \cite{Cole2002} to compute the edit distance of each segment in time $\Oh(k^2)$ each, which we then sum to find our approximation for $\dedd(X)$. 
In each step of the preceding process, we will rely on the dynamic strings data structure (\cite{DBLP:conf/soda/GawrychowskiKKL18}) to build a balanced binary tree on top of our input string $X$ which allows us to efficiently construct, store, and query information of $\hat{X}$.

In addition to the listed operations in Theorem~\ref{thm:dynstr}, we augment the data structure with two additional operations to efficiently maintain $\hat{X}$ and the Dyck to string edit distance reduction.  First, for any string $X \in \X$, the transpose of $X$, $T(X)$, may be queried and will also be stored in $\X$. Second, for any string $X \in \X$, we will answer $\lcpt(X)$ queries (recall $\lcpt(X)$ is the longest prefix of $X$ composed of only opening or only closing parentheses) in $\Oh(\log n)$ time.

\begin{lemma}
\label{lem:augdynstr}
    A collection $\X$ of parenthesis strings of total length $n$ can be dynamically maintained subject to the following queries in addition to the operations of \cref{thm:dynstr}:
    \begin{description}
        \item[$\addT(X)$] for $X\in \X$ results in $\X=\X\cup\{T(X)\}$ and takes $\Oh(1)$ time,
        \item[$\lcpt(X)$] for $X\in \X$ computes the length of the longest monotone prefix of $X$ (consisting of opening or closing parentheses only) and takes $\Oh(\log n)$ time.
    \end{description}
\end{lemma}

\begin{proof}
    Internally, we use \cref{thm:dynstr} to maintain two collections of strings $\X_T= \X \cup \{T(X) : X\in \X\}$ and $\X_{\out}=\{\out(X) : X \in \X_T\}$, where $\out(X)$ is the \emph{outline} of $X$, defined as a binary of string of length $|X|$ such that $\out(X)[i] = 0$ if $X[i]$ is an opening parenthesis and $\out(X)[i] = 1$ if $X[i]$ is a closing parenthesis for all $i \in [|X|]$.

    Let us first describe how to maintain $\X_T$.
    For every operation $\add(X)$ on $\X$, we compute $T(X)$ in $\Oh(|X|)$ time and use both $\add(X)$ and  $\add(T(X))$ on $\X_T$. 
    For every operation $\spl(X,k)$ on $\X$, we perform both $\spl(X,k)$ and $\spl(T(X),|X|-k)$ on $\X_T$ so that, besides $X[0 \dd k)$ and $X[k \dd |X|)$, also $T(X[0\dd k))=T(X)[|X|-k\dd |X|)$ and $T(X[k\dd |X|))=T(X)[0\dd |X|-k)$ are also added to $\X_T$.
    For every operation $\con(X_1,X_2)$ on $\X$, we perform both $\con(X_1,X_2)$ and $\con(T(X_2),T(X_1))$ on $\X_T$ so that, besides $X_1\cdot X_2$, also $T(X_1\cdot X_2)=T(X_2)\cdot T(X_1)$ is added to $\X_T$.
    The $\lcp(X_1,X_2)$ operation on $\X$ is simply performed on $\X_T$.
    It is easy to see that, with the above implementations, the asymptotic update and query times match those of \cref{thm:dynstr}.
    Moreover, the $\addT(X)$ operation on $\X$ does not require any update on $\X_T$: since we already have $T(X)\in \X_T$, it suffices to return the handle to $T(X)$. 
    
    Next, we focus on maintaining $\X_\out$. 
    When we perform $\add(X)$ on $\X_T$, we also straightforwardly compute $\out(X)$ in $\Oh(n)$-time and perform $\add(\out(X))$ on $\X_\out$. When we perform $\spl(X, k)$ or $\con(X_1, X_2)$ on $\X_T$, we also perform $\spl(\out(X), k)$ or $\con(\out(X_1), \out(X_2))$ on $\X_\out$, respectively. 
    Again, it is easy to see that, with the above implementations, the asymptotic update and query times match those of \cref{thm:dynstr}.
    It remains to implement the $\lcpt(X)$ query on $\X$. 
    For this, we observe that \[\lcpt(X) = \lcp(\out(X), \out(X)[1\dd|X|)) + 1.\]
    Consequently, we may perform $\spl(\out(X), 1)$ in $\Oh(\log n)$-time and then using the above $\lcp$ query to find $\lcpt(X)$ in $\Oh(\log n)$ total time. 
    To make sure that $\X_O$ does not grow too much, we remove the outputs of $\spl(\out(X), 1)$ as soon as the $\lcp$ query is completed.
\end{proof}

\cite{DBLP:conf/focs/Charalampopoulos20} and \cite{Cole2002} show that an oracle that answers $\lcp$ queries in logarithmic time is sufficient to efficiently compute the string edit distance between two strings. These papers utilize a slightly modified implementation of the classic dynamic programming solution of \cite{LV1989} for string edit distance that only relies on $\lcp$ queries. We may utilize the dynamic strings data structure to answer these queries and solve an instance of string edit distance in quadratic runtime with respect to the edit distance between the two strings. Since in this section we are considering when the edit distance is at most $\Oh(\sqrt{n})$, we can afford a quadratic blow-up (and some additional log factors). 

\begin{lemma}[\cite{DBLP:conf/focs/Charalampopoulos20} Lemma 6.2, \cite{Cole2002}, Section 5]
\label{lem:ked}
    Given strings $X_1, X_2$, $k \in \mathbf{Z}^+$, and an oracle that answers $\lcp$ queries of any pair of substrings of $X_1, X_2$ in $\Oh(\log n)$ time, $\min\{\edd(X_1, X_2), k\}$ can be computed in $\tOh(k^2)$ time.
\end{lemma}

\begin{proof}
    As noted in \cite{DBLP:conf/focs/Charalampopoulos20}, the algorithm of Section 5 \cite{Cole2002} can be utilized to find $\edd(X_1, X_2)$ with $\Oh(k^2)$ $\lcp$ queries. Therefore, if each $\lcp$ query may be answered in $\Oh(\log n)$ time, we may compute $\edd(X_1, X_2)$ in total time $\Oh(k^2 \log n)$.
\end{proof}

As a result of the preceding lemma, it turns out that the main challenge of this section is efficiently maintaining the dynamic reduction of Dyck to string edit distance, rather than solving the string edit distance instances themselves. In order to limit how many instances of string edit distance we must compute, we first pre-process $X$ to find and store $\hat{X}$ in a collection $\X$ of \cref{lem:augdynstr}.

\begin{definition}
\label{def:preproc}
    Given a parenthesis string $X$, we define $\hat{X}$ to be the parenthesis string resulting from removing every neighboring pair of an opening parenthesis followed by a matching closing parenthesis until no more removals can be performed.
\end{definition} 

An important detail about $\hat{X}$ is that it contains no balanced substrings since we remove all neighboring matching parentheses from $X$. Any sequence of opening parentheses followed by closing parentheses in $\hat{X}$ requires at least 1 edit to balance. Therefore, $\C(\hat{X})$, which is a collection of LR-segments by Theorem~\ref{thm:dycktoed}, only contains at most $\Oh(k)$ segments. This allows us to limit the number of string edit distance instances that we need to solve to $\Oh(k)$ each time we recompute our Dyck edit distance approximation via the reduction to string edit distance. 

\begin{lemma}[\cite{DBLP:conf/pods/BackursO16} Claim 35]
\label{lem:kpeaks}
    Given a parenthesis string $X$ with $\dedd(X) \leq k$, there are $\Oh(k)$ maximal LR-segments in $\hat{X}$.
\end{lemma}

An important step of constructing $\C(\hat{X})$ is finding sequences of opening parentheses followed by equal-length sequences of closing parentheses. We construct $\out(\hat{X})$ to quickly find such neighboring sequences since we do not care about the type of parentheses in this step, only whether they are opening or closing. 

\begin{lemma}
    \label{clm:hatX}
        For a parenthesis string $X$ with $n = |X|$, the string $\hat{X}$ can be maintained with $\tOh(n)$ pre-processing time and $\Oh(\log^2 n)$ update time when $X$ undergoes dynamic insertions, deletions, and substitutions.
    \end{lemma}

\begin{proof}
        First, we describe how to construct and maintain $\hat{X}$. 
        We are going to build a binary tree where the leaves of the tree partition $X$ into individual characters such that each leaf contains a single parenthesis of $X$.
         Each parent of a leaf node will concatenate the characters in its children, removing any pairs of matching parenthesis. We will continue this process level by level: for any internal node, we concatenate the substrings in its children and combine them, removing any new matching neighboring parentheses. We set the height of the binary tree such that the root node will contain $\hat{X}$ since we made sure to remove matching neighboring parentheses at each level of the tree.  
        
        To build the above binary tree, we use a dynamically maintained collection of strings. First, we let $\X$ be an empty collection of strings. Let $L = \Oh(\log n)$ be the number of levels in the binary tree we wish to build. We denote each node of the binary tree as $v_I^\ell$ where $\ell \in [L]$ is the level of the node and $I \subseteq [0 \dd |X| - 1]$ is the interval such that $v_I^\ell$ contains a pointer to $\hat{X_I} \in \X$ where $X_I = X[I]$ and $I$ is the union of the intervals of the children of $v_I^\ell$.
        Handling each leaf node $v_{I}^L$ is straightforward since $I$ contains just a single index and so $X_I = \hat{X_I}$ since there can be no pairs of parentheses in any of the intervals at the leaf level. We sort the leaves from left-to-right in increasing index such that the leftmost leaf is $v_{\{0\}}^L$ and the rightmost leaf is $v_{\{|X| - 1\}}^L$. We additionally perform $\add(X[i])$
        to $\X$ and store a pointer to $X[i]$
        in leaf $v_{\{i\}}^L$ for all $i \in [|X| - 1]$. Now we show how to continue building the binary tree bottom-up using $\X$. We assume for any level of the tree, the intervals of the nodes at that level form a non-overlapping partition of $[|X|]$ and are in sorted order. It is easy to see this will be the case inductively since the leaf level already partitions $[|X|]$ in sorted order, and we defined the interval of a node to be the union of the intervals of its children. For a given node $v_{I}^\ell$ with $0 \leq \ell \leq L$, we wish to add $\hat{X_I}$
        to $\X$ and store pointers to these substrings at this node. We may assume its children $v_{I_1}^{\ell + 1}$ and $v_{I_2}^{\ell+1}$ have pointers to $\hat{X_{I_1}} \in \X$ and $\hat{X_{I_2}} \in \X$, respectively and $I = I_1 \cup I_2$ as per our construction. To construct $\hat{X_I}$, we want to concatenate $\hat{X_{I_1}}, \hat{X_{I_2}}$ and remove any neighboring matching parentheses. By Definition~\ref{def:preproc}, there must not be any neighboring matching parentheses within either $\hat{X_{I_1}}$ or $\hat{X_{I_2}}$; we only need to check parentheses that are matches between the suffix of $\hat{X_{I_1}}$ and prefix of $\hat{X_{I_2}}$. To do so, we take $m = \min(\lcp(T(\hat{X_{I_1}}), \hat{X_{I_2}}), \lcpt(\hat{X_{I_2}}))$, which tells us exactly how long the sequence of matching neighboring parentheses between the two substrings is. Note that we only want to find the number of opening parentheses in the suffix of $\hat{X_{I_1}}$ that match with the closing parentheses in the prefix of $\hat{X_{I_2}}$, and we must be careful not to include any closing parentheses of $\hat{X_{I_1}}$ or any opening parentheses of $\hat{X_{I_2}}$ by making sure the largest value we can return is $\lcpt(\hat{X_{I_2}})$.
        We then perform $\spl(\hat{X_{I_1}}, |\hat{X_{I_1}}| - m)$ and $\spl(\hat{X_{I_2}}, m)$, and finally we concatenate the resulting substrings $\con(\hat{X_{I_1}}[0 \dd |\hat{X_{I_1}}| - m]), \hat{X_{I_2}}[m \dd])$ to add $\hat{X_I}$ to $\X$ and store a pointer to $\hat{X_I}$ in node $v_I^\ell$.
        We will continue concatenating substrings of each node's intervals in this way level by level until we reach the root and obtain $\hat{X}$.

        First, we note that there are $\Oh(\log n)$ levels of the binary tree, and therefore $\Oh(n)$ total nodes in the binary tree.  Each node requires 1 $\lcp$ query, 1 $\lcpt$ query, 2 $\spl$ operations, and 1 $\con$ operation to construct the substring for its interval. Therefore, each node takes time $\Oh(\log n)$ to process by Theorem~\ref{thm:dynstr} and Lemma~\ref{lem:augdynstr}, and so the total pre-processing time to build $\hat{X}$ is $\Oh(n \log n)$.

        When $X$ gets updated, we need only update nodes whose corresponding substring intervals are affected by the update. When the height of the binary tree is $\Oh(\log n)$, there are only $\Oh(\log n)$ such nodes. 
        Since we used a constant number of $\spl$, $\add$, and $\con$ operations each taking time $\Oh(\log n)$, reconstructing our affected paths may take time $\Oh(\log^2 n)$. Since deletions and insertions may cause the tree to become unbalanced, instead of a standard binary tree we may use a self-balancing tree such as an AVL tree instead to maintain the height of $\Oh(\log n)$ while still affecting $\Oh(\log n)$ nodes per update.
    \end{proof}

We now have all the ingredients to prove the main theorem of this subsection.

\smallDed*

\begin{proof}
    We will dynamically maintain a collection of strings $\X$ according to Theorem~\ref{thm:dynstr}, Lemma~\ref{lem:augdynstr} and
    construct $\hat{X}$ according to Lemma~\ref{clm:hatX}. Using these tools, we build our reduction to string edit distance in time $\tOh(k)$ where $k = \dedd(X)$. Normally the reduction takes $\tOh(n)$ time according to Theorem~\ref{thm:dycktoed}, but we show how to reduce computation time utilizing our dynamic strings collection in the following claim.

    \begin{claim}
    \label{clm:kreduction}
        Given a parenthesis string $X$ belonging to a dynamic collection $\X$ of \cref{lem:augdynstr}, one can implement the Dyck-to-string reduction of \cref{thm:dycktoed} in time $\tOh(s)$, where $s$ is the number of maximal LR-segments of $X$,
        with each string in $\C(X)$ reported as a handle to an element of $\X$.
    \end{claim}

    \begin{proof}
        We first recall the algorithm of \cite{KSSODA2023} for a parenthesis string $X$, which constructs $\C(X)$ satisfying Theorem~\ref{thm:dycktoed} as follows:  
        \begin{enumerate}
            \item\label{st:1} Construct the LR-decomposition $\D$ of $X$.
            \item\label{st:2} If the first (last) segment of $\D$ only contains closing (opening) parentheses, add this segment to $\C(X)$.
            \item\label{st:3} For each remaining segment $Y$ in $\D$, let $Y'$ be the maximal-length substring of $Y$ composed of a sequence of opening parentheses followed by an equal-length sequence of closing parentheses. Add $Y'$ to $\C(X)$.
            \item\label{st:4} Let $X'$ be the remainder of the string of $X$ after removing all such substrings $Y'$ from the previous step. If $X'$ is not empty, repeat steps \ref{st:1}--\ref{st:3} on $X'$.
        \end{enumerate}

        Note that in step \ref{st:3}, we fully remove either all the opening parentheses or all the closing parentheses in each segment, and therefore, we reduce the number of segments by at least half in each round of the algorithm. Thus, steps \ref{st:1}--\ref{st:3} may be repeated at most $\Oh(\log n)$ times. Unfortunately, without utilizing our dynamic string collection, each such repetition of steps \ref{st:1}--\ref{st:3} may take time $\Oh(|X|)$ to construct the LR-decomposition of $X$ and handle each segment. We now explain how to leverage $\X$ to avoid this linear-time blowup. 
    
        Since the LR-decomposition of $X$ consists of $\Oh(s)$ segments handled in steps \ref{st:2}--\ref{st:3} and we repeat steps \ref{st:1}--\ref{st:3} at most $\Oh(\log n)$ times, then the total reduction will take $\tOh(s)$ time if each segment can be processed in $\tOh(1)$ time.

        Recall that via Lemma~\ref{lem:augdynstr}, we may find the length of maximal prefixes of strings in $\X$ composed of only opening or closing parentheses in $\Oh(\log n)$-time.  Let $\ell_1 = \lcpt(X)$, we then perform $\spl(X, \ell_1)$ to isolate the first maximal sequence of opening or closing parentheses in $X$. If $X[0]$ is a closing parenthesis, then we add $X[0 \dd \ell_1)$ to $\C(X)$ as per step \ref{st:2}. If $X[0]$ is an opening parenthesis, we find the length of the following sequence of closing parentheses $\ell_2 = \lcpt(X(\ell_1\dd])$.  
        Define $m = \min(\ell_1, \ell_2)$, we then add $X[\ell_1 - m \dd \ell_1 + m)$ to $\C(X)$, which is the maximal substring of the first LR-segment in $X$ composed of an equal number of opening parentheses and following closing parentheses. If $m = \ell_1$, then there are more closing parentheses in this segment, and so, we use $\spl$ to add the remaining fragment $Y_1 = X(\ell_1 + m \dd \ell_1 + \ell_2]$ to dynamic string collection $\X$. If $m = \ell_2$, we instead use $\spl$ to add the remaining fragment $Y_1 = X[0 \dd \ell_1 - m)$ to $\X$. This marks the end of the processing for the first LR-segment, and we now repeat the same steps on $X(\ell_1 + \ell_2 \dd]$ to handle the remaining segments.  When we find the new remaining fragment $Y_2$ for the second segment, we perform $\con(Y_1, Y_2)$ and continue as before.  We process each segment, querying $\X$ to determine the length of the maximal sequence of opening and closing parentheses in each segment of the LR-decomposition.  Then, we find the remaining fragment and concatenate it to all previous fragments. After processing all segments, we will have concatenated all remaining fragments of LR-segments not added to $\C(X)$, and we may recurse on the resulting concatenation as per step \ref{st:4}.

        Each segment of the LR-decomposition requires a constant number of $\spl$, $\con$ operations and $\lcp$, $\lcpt$ queries each of which takes at most $\Oh(\log n)$ time by Theorem~\ref{thm:dynstr} and Lemma \ref{lem:augdynstr}.  There are $\Oh(s)$ segments per LR-decomposition and $\Oh(\log n)$ levels of recursion.  Therefore, in total computing $\C(X)$ takes time $\tOh(s)$.
    \end{proof}

    Since the runtime of \cref{clm:kreduction} depends on the number of segments in the LR-decomposition of $X$, instead of applying the reduction directly to $X$, we apply the reduction to $\hat{X}$, which has $\Oh(k)$ segments in its LR-decomposition by \cref{lem:kpeaks}.
    This is possible feasible because $\ded(X)=\ded(\hat{X})$ and, by Lemma~\ref{clm:hatX}, $\hat{X}$ can be maintained in $\X$.
    Once we have our set of LR-strings $\C(\hat{X})$, we may solve the corresponding string edit distance instances according to Observation~\ref{obs:LRed}. We will utilize Lemma~\ref{lem:ked} to solve all instances in total time $\tOh(k^2)$. Note that one condition of Lemma~\ref{lem:ked} is that we may answer $\lcp$ queries in $\Oh(\log n)$ time. Note that Claim~\ref{clm:kreduction} ensures that  $Y \in \X$ for all $Y \in \C(\hat{X})$. Therefore, for any substrings $Y_1, Y_2$ of $Y$ we wish to query, we only need a constant number of $\spl$ operations on $\X$ to add $Y_1$ and $Y_2$ to $\X$ and then perform our $\lcp$ query. By Theorem~\ref{thm:dynstr}, each $\spl$ takes $\Oh(\log n)$ time and each $\lcp$ query takes $\Oh(1)$, so we indeed satisfy the conditions of Lemma~\ref{lem:ked}. Lemma~\ref{lem:ked} gives an exact solution for string edit distance\footnote{If we do not know $k$, we may perform a binary search for the value of $k$ at the cost of an additional $\Oh(\log k)$ factor.}, and so by summing the string edit distance of each $Y \in \C(\hat{X})$, we obtain an $\Oh(\log n)$-approximation for $\dedd(\hat{X})$, and thereby also $\ded(X)$, by Theorem~\ref{thm:dycktoed}.

    Our algorithm proceeds in epochs.  Each epoch starts by computing $\hat{X}$ according to Lemma~\ref{clm:hatX} and the collection $\C(\hat{X})$ using the algorithm of Claim~\ref{clm:kreduction}. For each $Y \in \C(\hat{X})$, the next step is to compute $\dedd(Y) = \edd(\Lp(Y), T(\Rp(Y)))$ using the exact algorithm of Lemma~\ref{lem:ked}. Let $a_Y$ denote $\dedd(Y)$ for each $Y \in \C(\hat{X})$. We now discuss the approximation factor in a similar analysis to the Proof of Theorem~\ref{thm:largeded}. We use $a = 2\sum_{Y \in \C(Y)} a_Y$ as our approximation of $\dedd(X)$ for the duration of the epoch, which lasts for $m:=\lfloor a / (4(3+2\lg a))\rfloor$ dynamic edits.
    At the beginning of the epoch,
    \begin{align*}
        a = 2\sum_{Y \in \C(\hat{X})} a_Y 
        &= 2\sum_{Y \in \C(\hat{X})} \edd(\Lp(Y), T(\Rp(Y)))\\
        & =2\sum_{Y \in \C(\hat{X})} \dedd(Y)\\
        &\ge 2\cdot \dedd(X) \ge 2\cdot \dedd(X) \\
        & \text{and}\\
        a= 2\sum_{Y \in \C(\hat{X})} a_Y 
        &= 2\sum_{Y \in \C(\hat{X})} \edd(\Lp(Y), T(\Rp(Y))) \\
        &= 2\sum_{Y \in \C(\hat{X})} \dedd(Y) \\
        &\leq 2(3 + 2\lg (\dedd(X))) \dedd(X)\\
        &\leq 2(3 + 2\lg a) \dedd(X) \\
        &\leq 4(3 + 2 \lg a)\dedd(X),
    \end{align*}
    
    Throughout the epoch, the value $\dedd(X)$ may change by at most $\pm m$, so we still have 
    $a \ge \frac{a}{2}+m \ge \dedd(X)-m+m = \dedd(X)$
    and $a = 2a - a \le 8(3 + 2\lg a)(\dedd(X)+m)-a \le 8(3 + 2\lg a)\dedd(X)+8(3 + 2\lg a)m-a \le 8(3 + 2\lg a)\dedd(X)$, so the value $a$ remains an $\Oh(\log n)$-approximation of $\dedd(X)$.

    For each dynamic edit to $X$, we must update $\hat{X}$ and $\C(\hat{X})$, which takes time $\tOh(\log n + k)$ according to Lemma~\ref{clm:hatX} and Claim~\ref{clm:kreduction}. We only recompute $a$ using the $\Oh(k^2)$ string edit distance algorithm of Lemma~\ref{lem:ked} once per epoch that contains $m = \Omega(k/ \log n)$ edits; this yields $\tOh(\log n + k + k^2 \log n/k) = \tOh(1 + k)$ amortized time per update.
\end{proof}

\subsection{Combined Algorithm}
We may combine the algorithms of Theorem~\ref{thm:largeded} and Theorem~\ref{thm:smallded} and transition between the two routines according to our current approximation of Dyck edit distance for the parenthesis string.
\combinedDed*

\section{Exact Dynamic $k$-Dyck Edit Distance}
\label{sec:exact}

In this section, we give an algorithm for exact dynamic Dyck edit distance by keeping track of a unique decomposition of parentheses strings in trapezoidal segments and the remaining clusters \cite{DBLP:conf/soda/FriedGKKPS22}. The algorithm supports dynamic updates in time $\Oh((1 + \ded(X))^5)$. To build this decomposition, we consider the graph of the heights of all the parentheses in a parenthesis string, as shown in Figure~\ref{fig:parenthesis-tree}a, where recall the height of an index is the difference in the number of preceding closing parentheses from the number of opening parentheses. We then use this graph to find natural trapezoids, defined formally in the following definition, and build a tree structure on the resulting components, which we now state.
\begin{definition}[{Trapezoid, \cite[Definition 3.2]{DBLP:conf/soda/FriedGKKPS22}}]
Given a parenthesis string $X$, a \emph{trapezoid} $(a, b, c, d)$ is a quadruple such that 
\begin{enumerate}
    \item $h(a)=h(d)$,
    \item $b-a = h(b)-h(a)=h(c)-h(d)=d-c$, and
    \item $h(i) \geq h(b)=h(c)$ for all $i\in [b\dd c]$.
\end{enumerate}

\begin{figure}
\begin{center}
    \hfill\begin{tikzpicture}[scale=0.35,every node/.style={scale=0.75}]
\draw[->] (0,0)--(0,11) node[pos=0.95,right] {{$H(i)$}}; 
\draw[->] (0,0)--(29,0) node[pos=0.98,below] {{$i$}};

\draw (0,0)--(6,6)--(8,4)--(14,10)--(19,5)--(20,6)--(21,5)--(22,6)--(27,1);

\draw[very thick,green] (1,1)--(4,4);
\draw[very thick,green] (24,4)--(27,1);

\draw[very thick,red] (4,4)--(6,6)--(8,4);
\draw[very thick,yellow] (9,5)--(14,10)--(19,5);

\draw[thin,black] (4,4)--(24,4);
\draw[thin,black] (1,1)--(27,1);
\draw[thin,black] (9,5)--(19,5);
\draw[thin,black] (0,0)--(27,1);

\node at (1.5,0.5) {$C_1$};
\node at (13.5,2.5) {\textcolor{black}{$T_1$}};
\node at (13.5,6.5) {\textcolor{black}{$T_3$}};
\node at (6,5) {\textcolor{black}{$T_2$}};
\node at (15,4.5) {$C_2$};

\end{tikzpicture}\hfill\begin{tikzpicture}[xscale=0.35,yscale=-0.35,every node/.style={scale=0.75}]

\node[circle,draw=black, fill=white, inner sep=1pt,minimum size=7pt] (a) at (0,0) {\textcolor{black}{$C_1$}};
\node[circle,draw=black, fill=white, inner sep=1pt,minimum size=7pt] (b) at (0,-2) {\textcolor{black}{$T_1$}};
\node[circle,draw=black, fill=white, inner sep=1pt,minimum size=7pt] (c) at (0,-4) {\textcolor{black}{$C_2$}};
\node[circle,draw=black, fill=white, inner sep=1pt,minimum size=7pt] (d) at (-2,-6) {\textcolor{black}{$T_2$}};
\node[circle,draw=black, fill=white, inner sep=1pt,minimum size=7pt] (e) at (2,-6) {\textcolor{black}{$T_3$}};

\draw (a)--(b);
\draw (b)--(c);
\draw (c)--(d);
\draw (c)--(e);

\end{tikzpicture}\hfill\;
\end{center}
\caption{Example of a trapezoid and cluster decomposition (left) and corresponding decomposition (right) from \cite{DBLP:conf/soda/FriedGKKPS22}. Note that colored edges depict tall maximal trapezoids while black edges depict clusters.}
\label{fig:decompTree}
\end{figure}
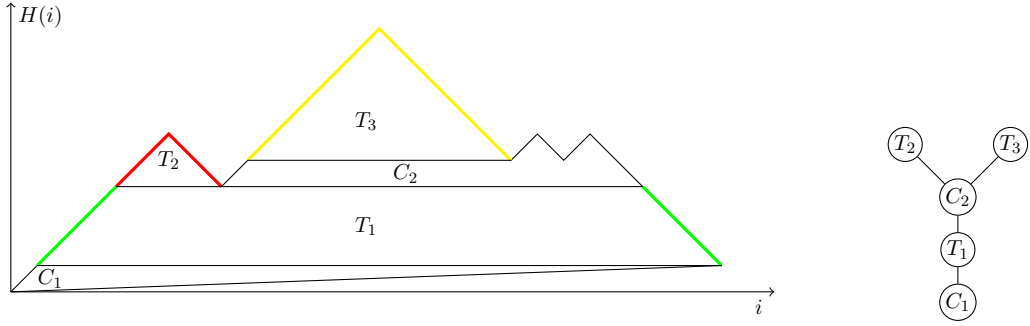

A trapezoid is \emph{maximal} if it cannot be extended, i.e., neither $(a - 1, b, c, d + 1)$ nor $(a, b + 1, c - 1, d)$ is a trapezoid.  A trapezoid is \emph{tall} if $b - a \geq 2k$, where $k = \ded(X)$\footnote{Since $k$ is initially unknown, we may use an exponential search for $k$ by running the algorithm on $k$ set to $1, 2, 4, 8, \ldots$ until it succeeds with the nearest power of $2$ larger than $\ded(X)$.  This only adds a factor of $\Oh(\log k)$ to the pre-processing of the data structure. When dynamic edits occur, we may use our computed value of $\ded(X) + 2$ before the edit as our upper bound since Dyck edit distance can only change by 2 in either direction per dynamic edit.}.%\tknote{We need to fix some uppre bound $k$ on $\ded(X)$.}
\end{definition}

\begin{definition}[{Cluster, \cite[Definition 3.5]{DBLP:conf/soda/FriedGKKPS22}}]
    Consider a cycle on $n+1$ vertices $[0\dd n]$ with edges $\{(i, i + 1) | i \in [0\dd n)\} \cup \{(n, 0)\}$.  Given a parenthesis string $X$, for each tall maximal trapezoid $(a, b, c, d)$ of $X$, remove all edges $(i, i + 1)$ for $i \in [a \dd b) \cup [c \dd d)$, remove vertices $(a\dd b)\cup (c\dd d)$, and add edges $(a,d)$, $(b,c)$. A \emph{cluster} is a connected component of the resulting graph.
\end{definition}

The set of tall maximal trapezoids and clusters partition the indices of $X$, and there exists a natural tree, the \emph{decomposition tree} such that each node of the tree corresponds to a tall maximal trapezoid or a cluster. The cluster containing indices 0 and $n$ is the root of the tree.  Given a cluster $C$ and trapezoid $T = (a, b, c, d)$, $C$ is the parent of $T$ if $a, d \in C$, and $T$ is the parent of $C$ if $b, c \in C$. See Figure~\ref{fig:decompTree} for an example of such a decomposition.

The key component of the trapezoid and cluster decomposition is that the Dyck edit distance of $\hat{X}$ (recall $\hat{X}$ is the parenthesis string obtained by removing all neighboring pairs of matching parentheses from $X$ iteratively until no such pairs remain and has the same Dyck edit distance as $X$) may be computed efficiently when $k$ is small using this tree.  There are at most $\Oh(k)$ trapezoids and $\Oh(k^2)$ indices of $\hat{X}$ in clusters. Each trapezoid may be processed in time $\tOh(k^3)$ with access to an oracle that answers $\lcp$ queries in $\Oh(1)$-time, and each cluster processed in time $\Oh(k^5)$. The main difficulty in the dynamic setting is constructing the trapezoid and cluster decomposition and maintaining the decomposition tree efficiently. The static setting requires a linear time construction, but we may utilize the LR-decomposition, the partition of $\hat{X}$ into maximal-length substrings of opening parentheses followed by closing parentheses, first defined in Section~\ref{subsec:largeed} to compute this in polynomial time dependent on $k$. Similarly to Section 3, we utilize the dynamic strings data structure to build the LR-decomposition, and we augment the data structure to store additional information about $X$ that we wish to access later.

\begin{lemma}
\label{lem:LRdecomp}
    Given a parenthesis string $X$ with $k = \ded(X)$ and a dynamically maintained collection of strings $X$ with $\hat{X} \in \X$, the LR-decomposition $\D$ of $\hat{X}$ can be computed in time $\Oh(k \log n)$.  Furthermore, $\D$ may constructed such that given a segment $S \in \D$, $|\Lp(S)|, |\Rp(S)|$ and for any index $i \in [|X| - 1]$, $h(i)$ may be found in $\Oh(\log k)$ time. 
\end{lemma}

\begin{proof}
    Let $\ell_1 = \lcpt(\hat{X})$, where recall $\lcpt(\hat{X})$ is the maximal-length prefix of $\hat{X}$ composed of only opening or only closing parentheses. Next, we perform $\spl(\hat{X}, \ell_1 - 1)$ to separate this prefix from $\hat{X}$. If $\hat{X}[0]$ is a closing parenthesis, we add $S_1 = \hat{X}[0 \dd \ell_1)$ to $\D$. Otherwise, we let $\ell_2 = \lcpt(\hat{X}[m\dd])$ and perform $\spl(\hat{X}, \ell_1 + \ell_2 - 1)$ and then add $S_1 = \hat{X}[0 \dd \ell_1 + \ell_2)$ to $\D$. We continue repeating the same steps for the remaining string $X' = X[\ell_1 + \ell_2 \dd]$ until $X'$ is empty after splitting the last segment.  By Lemma~\ref{lem:kpeaks}, there are $\Oh(k)$ segments in $\Oh(\hat{X})$. Each segment may be added to $\D$ using a constant number of $\spl$ and $\lcpt$ queries.  By Theorem~\ref{thm:dynstr}, each $\spl$ takes time $\Oh(\log n)$ and by Lemma~\ref{lem:augdynstr}, each $\lcpt$ query takes time $\Oh(\log n)$. Therefore, in total, constructing $\D$ takes time $\Oh(k \log n)$.

    Now, during construction of $\D$, when finding a segment $S_i \in \D$, we observe that the values $\ell_1$ and $\ell_2$ correspond to lengths $|\Lp(S_i)|, |\Rp(S_i)|$, respectively. Furthermore, we may easily store the heights of the first index $\alpha^\Lp_i$ in $\Lp(S_i)$ and the first index $\alpha^\Rp_i$ in $\Rp(S_i)$ by taking $h(\alpha^\Lp_i) = \sum_{j = 0}^{i-1} |\Lp(S_j)| - |\Rp(S_j)|$ and $\alpha^\Rp_i = |\Lp(S_i)| + \sum_{j = 0}^{i-1} |\Lp(S_j)| - |\Rp(S_j)|$, respectively. Then, for any index $i \in [|X|]$, we may first use a binary search through $\D$ to identify which segment $S_i \in \D$ that $i$ belongs to.  Then, if $X[i] \in \Lp(S_i)$, we find that $h(i) = h(\alpha^\Lp_i) + i - \alpha^\Lp_i$. Similarly, if $X[i] \in \Rp(S_i)$, then $h(i) = h(\alpha^\Rp_i) + \alpha^\Rp_i - i$. Overall, this takes $\Oh(\log k)$ time per height query due to the binary search required.
\end{proof}

We now give the construction of the set of tall maximal trapezoids of $\hat{X}$.

\begin{lemma}
\label{lem:trapconstr}
    Given a parenthesis string $X$ and a dynamically maintained collection of strings $\X$ with $\hat{X} \in \X$, the set of tall maximal trapezoids of $\hat{X}$ can be computed in time $\Oh(k^2 \log n)$.
\end{lemma}

\begin{proof}
    First, we compute the LR-decomposition $\D$ of $\hat{X}$ in time $\Oh(k\log n)$ according to Lemma~\ref{lem:LRdecomp}.

    We assume as per their construction that the segments $S_1, S_2, \ldots, S_m$ of $\D$ are sorted according to their starting indices, e.g. $X[0] \in S_1$. For each segment $S_i$, let $\alpha^{\Lp}_i$ and $\beta^{\Lp}_i$ be the first and last index of $\Lp(S_i)$, respectively, and define $\alpha^{\Rp}_i$ and $\beta^{\Rp}_i$ similarly for $\Rp(S_i)$. For each segment $S_i$, we observe that any index $j$ of $\Lp(S_i)$ will be part of a trapezoid with some index in $\Rp(S_j)$ for some $j \geq i$ if $S_j$ is the nearest segment to the right of $S_i$ such that it contains an index of height exactly 1 more than $h(j)$. Therefore, we iterate through $\D$ from left-to-right to match every opening parenthesis in a segment of $\D$ to the closing parenthesis of a later (or the same) segment with the correct corresponding height. 
    
    We consider an index unmatched if it has not yet been included in a trapezoid by our construction algorithm. When considering two segments $S_i, S_j$ with $i \le j$, we first take the highest height of its unmatched opening parentheses and the lowest height of closing parentheses in $S_j$.  We note that as we are considering segments left-to-right, $S_j$ must be able to match all unmatched opening parentheses with heights between the highest height of $S_i$ and the lowest height of $S_j$. Therefore, we take this difference as the height of our trapezoid, taking care to not exceed the number of unmatched opening parentheses in $S_i$. If the computed trapezoid height is larger than $2k$, we add it to the set of tall trapezoids, otherwise we move onto the next segment $\Rp(S_{j + 1})$ to consider matching with $\Lp(S_i)$.  In either case, we update how many opening parentheses have been matched so far in $\Lp(S_i)$. The following is our construction algorithm, which we do for each $S_i$ for $i$ from 1 to $m$:

    \begin{enumerate}
        \item $matched \leftarrow 0$, $\T \leftarrow \emptyset$
        \item For each $S_j$ for $j$ from $i$ to $m$:
            \begin{enumerate}

                \item Let $d \leftarrow h(\beta^\Lp_i - matched) - h(\beta^\Rp_j)$, $d' \leftarrow |\Lp(S_i)| - matched$
                \item If $d' \geq d \geq 2k$, then $\T \leftarrow \T \cup (\beta_i^\Lp - matched - d, \beta^\Lp_i - matched, \beta^\Rp_j - d, \beta^\Rp_j)$
                \item If $d > d' \geq 2k$,
                then $\T \leftarrow \T \cup (\alpha^\Lp_i, \beta^\Lp_i - matched, \beta^\Rp_j - d, \beta^\Rp_j - d + d')$
                \item If $h(\beta^{\Lp}_i- matched) \geq h(\beta^\Rp_j)$, then set $matched \leftarrow \min(matched + h(\beta^{\Lp}_i- matched) - h(\beta^\Rp_j) + 1, |\Lp(S_i)|)$.
            \end{enumerate}
    \end{enumerate}

    Note that there are $m \leq k$ segments, so there are $k^2$ repetitions of step 2 for our construction of $\T$. By Lemma~\ref{lem:LRdecomp}, we may answer height queries in time $\Oh(\log n)$, and therefore, the total time of construction for $\T$ is $\Oh(k^2 \log n)$.
\end{proof}

\begin{lemma}
\label{lem:clusterconstr}
    Given a parenthesis string $X$, its LR-decomposition $\D$, and its corresponding set of tall maximal trapezoids $\T$, the set of clusters of $\hat{X}$ and the decomposition tree of $X$ can be computed in time $\tOh(k^2)$.
\end{lemma}

\begin{proof}

We describe the procedure for computing each cluster.  Recall that each cluster is a connected component of the graph of parenthesis heights in $X$ with each edge of tall maximal trapezoids $(a, b, c, d)$ removed and edges $(a, d)$ and $(b, c)$ inserted. We assume the segments in $\D$ are in increasing sorted order according to the starting indices of the segments. Let $S_i$ be the first segment of $\D$ with a substring contained in a trapezoid $T = (a, b, c, d) \in \T$. We add all previous segments $S_j, j < i$ to our first cluster $C_1$ as well as the fragment of $S_i$ before index $a$. Then, we find the segment $S_j$ containing index $d$, and continue in a similar manner. Let the $S_{i'}$ be the next closest segment $i' \geq j$ contained in a tall maximal trapezoid $T' = (a', b', c', d') \ne T$. We add the fragment of $S_j$ after index $d$ and before trapezoid $T'$ to cluster $C_1$ as well as all segments between $S_j$ and $S_{i'}$. Additionally, we add the fragment of $S_{i'}$ before index $a'$. We then repeat these steps until we reach the end of $\hat{X}$. $C_1$ is the parent of any trapezoids traversed in this step.

We continue in a similar manner to construct the remaining connected components. To find the next cluster $C_m$, we find the first index $u$ of $X$ not contained in a previous cluster $C_\ell$, $\ell < m$, and find the segment $S_i \in \D$ containing $u$ (for example, for $C_2$ this index will be $b$ of trapezoid $T$ from our construction of $C_1$). We add every index in the remainder of $S_i$ and every following segment of $\D$ until we reach a trapezoid $T = (a, b, c, d)$.  If $a$ is the first index of $T$ that we reach, then we continue adding indices starting at $d$ as we did in the first step and set $T$ as the child of $C_m$.  If instead we reached index $c$ (e.g., in the case we started our connected component at index $b$ earlier), we have found the entire connected component, set $T$ as the parent of $C_m$, and move onto the next cluster.

There are at most $\Oh(k)$ clusters we need compute, and for each, we may have to check each segment of $\D$, where $|\D| = \Oh(k)$ by Lemma~\ref{lem:kpeaks}.  We may sort the set of trapezoids are sorted according to their starting indices in $\Oh(k \log k)$-time, and so we may use a binary search whenever we need to find the nearest trapezoid after an a given index of $\hat{X}$.  Therefore, the total time of construction will be $\tOh(k^2)$,

\end{proof}

\exactK*

\begin{proof}
The subroutine processing a tall trapezoid is given $\ded_{\le k}(X[i\dd j))$ for all $i,j\in [b-2k\dd b]\cup [c-2k\dd c]$, and the output is $\ded_{\le k}(X[i\dd j))$ for $i,j\in [a\dd a+2k]\cup [d-2k\dd d]$.
For $(i,j)\in [a\dd b]^2\cup [c\dd d]^2$, we have $\ded_{\le k}(X[i\dd j))=\lceil{(j-i)/2}\rceil$ because all parentheses in $X[i\dd j)$ are all opening or all closing. 
Furthermore, \cite[Lemma 3.4]{DBLP:conf/soda/FriedGKKPS22} claims that $(i,j)\in ([a\dd b]\times [c\dd d])\setminus ([b-2k\dd b]\times [c\dd c+2k])$ and $v\in [0\dd k]$, we have $\ded(X[i\dd j))\le v$ if and only if $L_v[(i+j)-(b+c)]\ge j$, where $L_v$ is computed by the following algorithm, which assumes that all out-of-bounds are initialized entries have value $-\infty$. We also denote $\delta_+ = \max(\delta,0)$ and $\delta_- =  \min(\delta,0)$.

{\footnotesize
\begin{algorithm}
\caption{Processing a trapezoid $(a,b,c,d)$}
\For{$v:=0$ \KwSty{to} $k$}{
    \For{$\delta:=-2v$ \KwSty{to} $2v$}{
        $L'_v[\delta]:=\min(d-2k+\delta_+,\linebreak\text{ }\qquad\quad\qquad \max(L_{v-1}[\delta-2]+2,L_{v-1}[\delta-1]+1, L_{v-1}[\delta]+1, L_{v-1}[\delta+1], L_{v-1}[\delta+2]))$\;
        \If{$\ded(X[b-2k+\delta_+\dd c+2k+\delta_-))\le v$}{
            $L'_v[\delta]:=\max(L'_v[\delta],c+2k+\delta_-)$\;
        }
        \If{$L'_v[\delta]\ne -\infty$}{
            $L_v[\delta]:=L'_v[\delta]+\lcp(T(X[a\dd b+c+\delta-L'_v[\delta])), X[L'_v[\delta]\dd d))$;
        }
    }
}
\end{algorithm}
}
This algorithm can be implemented in $\Oh(k^2 \log n)$ time assuming that $X$ is stored in a data structure of \cref{lem:augdynstr}. To process the clusters, we note that there are $\Oh(k^2)$ indices in clusters, so each index of the cluster may be solved with a standard cubic time Dyck edit distance algorithm in $\Oh(k^{2 + 3})$ total time (see Fried et al. \cite{DBLP:conf/soda/FriedGKKPS22} for details).
    By Lemmas~\ref{lem:trapconstr} and \ref{lem:clusterconstr}, we can compute the set of trapezoids, clusters, and decomposition tree in time $\Oh(k^2)$. So in total, we have $\Oh(k^5)$ update time.
\end{proof}

\bibliographystyle{alphaurl}
\bibliography{refs}

\appendix

\section{Missing Proofs of Section~\ref{sec:fastapprox}}\label{sec:missingproofs}

\subsection{Lazy Propagation}\label{sec:lazy}
For our sub-polynomial dynamic Dyck edit distance approximation algorithm, we often wish to find parentheses with specific height conditions, such as twins of a given parenthesis or the index of the parenthesis corresponding to the parent node of a given child node. In general, we wish to quickly answer what we call \emph{minimum height range queries}, queries $\rg(i, h)$ that return the maximum range starting at $i$ such that all parentheses in this range have heights greater than $h$.

In this section, we describe how we can use segment trees to pre-process a parenthesis string $X$ to answer such range queries in logarithmic time, and then show how to use a well-known technique \emph{lazy propagation} to correctly maintain this data structure when $X$ undergoes dynamic updates.

The segment tree is a binary tree. Each leaf node of the segment tree $\T$ for $X$ will correspond to an index $0 \leq i < |X|$.  Leaf nodes will be ordered such that the leftmost leaf node corresponds to index 0 and the rightmost leaf node corresponds to index $|X| - 1$. Each internal node will be the parent of two consecutive nodes from the following level and will correspond to the range of indices in the concatenation of the range of indices of its children.  Additionally, each node $v$ with corresponding range $[i, j]$ stores the minimum height of an index in this range, denoted $minh(v)$. For a leaf node $v$, $minh(v) = h(i)$ where $i$ is the index corresponding to $v$. 

\SetKwFunction{getrange}{GetMinRange}
\SetKwFunction{combine}{Combine}

\begin{algorithm}
    \Fn{\combine$([\ell_1\dd r_1], [\ell_2\dd r_2])$}{
        \eIf{$r_1 == \ell_2 - 1$}
        {
            Return $[\ell_1\dd r_2]$\;
        }
        {
            Return $[\ell_1\dd r_1]$\;
        }
    }
    
    \Fn{\getrange$(v, i, h)$}{
        \If{$i > v.r$}{
            Return $\emptyset$\;
        }
        \If{$i \leq v.\ell$ and $minh(v) > h$}
        {
            Return $[v.\ell\dd v.r]$\;
        }
        \uIf{$i \leq v.\ell$ and $minh(v.left) \leq h$ }{Return \getrange$(v.left, i, h)$}
        \uElseIf{$i \leq \ell$ and $minh(v.left) > h$}
        {
            Return $\combine([v.left.\ell, v.left.r], \getrange(v.right, i, h))$
        }
        \Else{Return $\combine(\getrange(v.left, i, h), \getrange(v.right, i, h))$\;
        }
    }
    \caption{Answering minimum height range queries with segment trees}
    \label{alg:rangequery}
\end{algorithm}

\begin{claim}
\label{clm:segtree}
    For a parenthesis string $X$, Algorithm~\ref{alg:rangequery} answers a minimum height range query $(i, h)$ in $\Oh(\log^2 n)$ time.
\end{claim}

\begin{proof}
    We first build the segment tree as described above for $X$.  We can then pass the root of the tree to the $\getrange$ function along with $\ell = 0, r = |X| - 1$ and the query values $i$ and $h$.  Whenever the range of $v$, the node we are considering in the $\getrange$ function lies completely to the left of the index $i$, we do not add it to our range or consider it any further. Conversely, if the range of $v$ lies fully to the right of $i$, if it has minimum height above $h$ we return the full range to be passed to the $\combine$ function later.  If the range of $v$ lies fully to the right $i$ but the minimum height in the left half of the range is not above $h$, then we only recurse on the left child, which corresponds to the left half of the range.  If instead, the left child range has minimum height fully above $h$, then we recurse on the right child and will return the full left child range concatenated with any parts of the right child found to be in the solution as well.  Finally, if the range of $v$ is not fully to the right of $i$ we just recurse on both children.

    Note that we want the algorithm to return a contiguous range starting from $i$. Therefore, when we have potential pieces of the solution from both a left and right child, we use the $\combine$ function to take the union of the left and right children's partial solutions if the resulting range would be contiguous, or otherwise only take the left child's returned range.

    We call a node \emph{active} if is passed to $\getrange$ at some timestep when Algorithm~\ref{alg:rangequery} is run. We show by induction that there are only $j$ active nodes from level $j$ of the segment tree.  As a base case, we note that there is 1 node, the root, at level 1 in the tree and it is passed directly to the algorithm at the start.  We now assume that there are $j-1$ active nodes at level $j-1$.  We note that at most 1 active node on each level may contain $j$ in its range.   This node may recurse on both its left and right child in the following level. For any node whose range falls completely to the left of $i$, we do not continue recursing on this node's subtree any further.  This is true for any node whose range falls completely to the right of $i$ with minimum height greater than $h$ as well. Finally, for any remaining node, which must have a range falling completely to the right of $i$, the algorithm only recurses on either the left or right child. Therefore, all nodes recurse on at most 1 child except for the node whose range contains $i$, which recurses on 2, i.e., there are at most $j - 1 + 1 = j$ active nodes in the next level of the tree.

     Each active node only requires constant time ignoring their recursive calls. Since the height of a segment tree is $\Oh(\log n)$, there will be $1 + 2 + \ldots + \Oh(\log n)$ active nodes, and in total, the run-time of Algorithm~\ref{alg:rangequery} will be $\Oh(\log^2 n)$.
\end{proof}

While we only care about answering such queries in poly-logarithmic time, we note that improving the run-time for each query to $\Oh(\log n)$ is possible. For any pair of active nodes $v_1, v_2$ on the same level of the segment tree with ranges $[\ell_1 \dd r_1], [\ell_2\dd r_2]$, respectively, such that $i < \ell_1 < \ell_2$, we note that only one of these active nodes needs to recurse on any children.  If $minh(v_1) > h$, then the algorithm returns $[\ell_1\dd r_1]$ as normal without recursing.  If instead $minh(v_1) \leq h$, this means that it is not possible for $[\ell_2\dd r_2]$ to be in the solution to the range query since there is already some parenthesis between $\ell_2$ and $i$ with height smaller than $h$.  We may put this idea into practice by maintaining a queue of active nodes at each level in sorted order according to their ranges.  Before processing an active node $v_2$, we just check the range $[\ell_1\dd r_1]$ and minimum height of the node $v_1$ in the queue to the left of $v_2$.  If $\ell_1 > i$ and $minh(v_1) \leq h$, we do not have to continue recursing on $v_2$ and can return an empty range. Now, by induction again we may show that there are at most 3 active nodes per level of the segment tree.  In the top level of the tree, trivially there is 1 active node.  Now we look at a lower level $j$ and assume there are 3 active nodes. 1 such node $v$ may correspond to a range containing $i$ and so may need to recurse on both children.  Any active node to the left of $v$ in the queue will not recurse on any children since their corresponding range falls completely to the left of $i$.  If there are two active nodes to the right of $v$, one of them will not recurse by our above argument.  The remaining active node to the right of $v$ will only recurse on a single child as in the proof of Claim~\ref{clm:segtree}. Therefore, there are 3 active nodes in level $i + 1$, completing the inductive argument. In total, we will visit less than $3\log n$ active nodes, and so the runtime will be $\Oh(\log n)$.

%AVL tree for self-balancing + lazy propagation for incrementing/decrementing all heights for indices after the update.

We now discuss how to quickly update the segment tree to handle dynamic insertions, deletions, and substitutions in the parenthesis string.  First, we want to make sure the height of the segment tree is always $\Oh(\log n)$.  Initially, this will be true due to our construction where each leaf corresponds to a single index and every parent node concatenates the two ranges of its children.  We may continue to guarantee $\Oh(\log n)$ height by using a self-balancing tree such as an AVL tree to handle when changes happen to $X$.  Every deletion or insertion corresponds to an insertion or deletion of a leaf.  We may additionally need to recompute $minh$ values for all nodes in the path to any inserted, deleted or rotated nodes after a rebalance, which may be done in $\Oh(\log n)$ time as $\Oh(\log n)$ nodes may be affected by a single update to $X$.  Furthermore, an update of $X$ at index $i$ changes the height of all parentheses with indices greater than $i$.  If the update is an inserted opening parenthesis or a deleted closing parenthesis, we must increment the height of all following parentheses by 1.  If the update is a deleted opening parenthesis or inserted closing parenthesis, heights decrease by 1.  If the update is a substitution, this is equivalent to a deletion and then insertion of the matching parenthesis, so the height changes by 2 for all subsequent parentheses in the relevant direction.

To handle efficient increments and decrements to the height of large sets of nodes in the segment tree, we use a well-known technique called \emph{lazy propagation}. At each node, we will add an additional variable to store any increment to its range. When we access the min height of a node's range, we will add the new variable's value to the min height to include any updates to the parenthesis string.  We can propagate any increments in a parent node to its children precisely when we are traversing through the segment tree during a range query and reset the parent node's increment variable to 0. This guarantees that the children have enough information to correctly compute their minimum height without needing to update every node, instead updating only when those nodes are needed. Algorithm~\ref{alg:updateseg} provides pseudo-code of the update function which takes a range of parentheses and the value of the change in height for all parentheses in this range.

\SetKwFunction{update}{Update}

\begin{algorithm}
    \Fn{\update$(v, \ell, r, incr)$}{
        \If{$\ell > r$}
        {
            Return\;
        }
        \eIf{$v.\ell == \ell$ and $v.r == r$}
        {
            $v.incr \leftarrow v.incr + incr$\;
        }
        {
            $\update(v.left, \ell, \min(v.left.r, r))$\;
            $\update(v.right, \max(v.right.\ell, \ell), r)$\;
        }
    }
    \caption{Update the heights in the segment tree when a dynamic edit occurs.}
    \label{alg:updateseg}
\end{algorithm}

\begin{claim}
    Algorithm~\ref{alg:updateseg} takes $\Oh(\log n)$ time.
\end{claim}

\begin{proof}
    The proof is by induction.  An active node is any node passed to $\update$ at some time step of the algorithm. We argue inductively that there are at most 4 active nodes per level of the segment tree. At the root level, this is trivially true.  We assume this is true up to some level $j$.  The leftmost and rightmost nodes may recurse on both children.  However, the inner two nodes will not recurse since they are fully in the range of parentheses that need to be incremented.  Therefore, we just set these two nodes' $incr$ variables accordingly and do not recurse.  Thus, there will only be 4 nodes in level $j+1$.  

    Since there are 4 active nodes per level and the height of the segment tree is $\Oh(\log n)$, we have $\Oh(\log n)$ total run time.
\end{proof}

\end{document}